\documentclass[reqno,12pt]{amsart}

\usepackage[english]{babel}

\usepackage[dvipdfmx]{graphicx}

\usepackage{caption}
\usepackage{subcaption}
\usepackage{bookmark}

\usepackage{booktabs}

\usepackage{hhline,float,tabularx, multirow, amsmath, amstext,amssymb}
\usepackage{enumerate}
\usepackage{mathpazo}
\usepackage{color,colortbl}
\definecolor{blue1}{RGB}{25,25,112}
\usepackage{graphicx}
\usepackage[foot]{amsaddr}
\usepackage{mathrsfs}
\usepackage{latexsym}
\usepackage{fancyhdr}
\usepackage{amsthm}
\usepackage{thmtools}
\usepackage{epigraph}
\usepackage{natbib,har2nat}
\usepackage[resetlabels]{multibib}
\newcites{New}{Appendix References}

\newcommand{\widebar}[1]{\mkern 1.5mu\overline{\mkern-1.5mu#1\mkern-1.5mu}\mkern 1.5mu}

\renewcommand{\Re}{\mathbb{R}}

\makeatletter
\newcommand{\oset}[2]{%
  {\mathop{#2}\limits^{\vbox to -.5\ex@{\kern-\tw@\ex@
   \hbox{\scriptsize #1}\vss}}}}
\makeatother

\declaretheoremstyle[notefont=\bfseries,notebraces={}{},%
    headpunct={},postheadspace=1em]{mystyle}

\usepackage{stmaryrd}
\usepackage[margin=1in]{geometry}
\usepackage{cancel}
\usepackage[onehalfspacing]{setspace}
\usepackage{hyperref}
\hypersetup{bookmarksnumbered=true,
            colorlinks=true,
            citecolor=blue1,urlcolor=blue1,linkcolor=blue1,
          }

\newcommand{\refthm}[1]{\hyperref[#1]{Theorem \ref{#1}}}
\newcommand{\refprop}[1]{\hyperref[#1]{Proposition \ref{#1}}}
\newcommand{\reflem}[1]{\hyperref[#1]{Lemma \ref{#1}}}
\newcommand{\refsec}[1]{\hyperref[#1]{Section \ref{#1}}}
\newcommand{\refsecs}[1]{\hyperref[#1]{Sections \ref{#1}}}
\newcommand{\refass}[1]{\hyperref[#1]{Assumption \ref{#1}}}
\newcommand{\refappendix}[1]{\hyperref[#1]{Appendix \ref{#1}}}
\newcommand{\refobs}[1]{\hyperref[#1]{Observation \ref{#1}}}
\newcommand{\refeg}[1]{\hyperref[#1]{Example \ref{#1}}}
\newcommand{\reffig}[1]{\hyperref[#1]{Figure \ref{#1}}}
\newcommand{\reftab}[1]{\hyperref[#1]{Table \ref{#1}}}
\newcommand{\refcond}[1]{\hyperref[#1]{Condition \ref{#1}}}
\newcommand{\refAfriat}[1]{\hyperref[#1]{Afriat's Theorem}}
\allowdisplaybreaks
\newcommand{\KS}{\citetalias{kitamura2018} }

\newcommand{\AP}{\citetalias{attanasio2020} }

\begin{document}

\author[Deb]{Rahul Deb$^{\between}$}
\email{rahul.deb@utoronto.ca}

\author[Kitamura]{Yuichi Kitamura$^\dagger$}
\email{yuichi.kitamura@yale.edu}

\author[Quah]{John K.-H. Quah$^\ddagger$}
\email{john.quah@jhu.edu}

\author[Stoye]{J\"org Stoye$^\star$\\ 	\today}
\address{$^{\between}$University of Toronto, $^\dagger$Yale University, $^\ddagger$Johns Hopkins University, $^\star$Cornell University.}
\email{stoye@cornell.edu}

\thanks{We are very grateful for comments provided by the editor, the referees and numerous seminar, conference participants. We also thank Samuel Norris and Krishna Pendakur for generously sharing their data with us and John Hoddinott for patiently explaining Progresa. Deb thanks the SSHRC for their continued financial support and the Cowles Foundation for their hospitality. Stoye acknowledges support from NSF grants SES-1260980 and SES-1824375 as well as from Cornell University's Economics Computer Cluster Organization, which was partially funded through NSF Grant SES-0922005. Part of this work was done while Stoye was at Bonn University and the Hausdorff Center for Mathematics. We thank Michael Sullivan and Matthew Thirkettle for excellent research assistance.}

\makeatletter

\renewcommand\section{\@startsection{section}{2}%
        \z@{.7\linespacing\@plus\linespacing}{.5\linespacing}%
        {\scshape\centering}}
\renewcommand\subsection{\@startsection{subsection}{2}%
        \z@{.7\linespacing\@plus\linespacing}{.5\linespacing}%
        {\normalfont\itshape\centering}}

\defcitealias{kitamura2018}{KS}
\defcitealias{McFadden1991}{MR}
\defcitealias{attanasio2020}{AP}

\makeatother
\newtheorem{cor}{Corollary}
\newtheorem{theorem}{Theorem}
\newtheorem*{afriattheorem*}{Afriat's Theorem}
\newtheorem*{mcrichtheorem*}{\MR Theorem}
\newtheorem{lemma}{Lemma}
\newtheorem{proposition}{Proposition}
\newtheorem{observation}{Observation}
\theoremstyle{definition}
\newtheorem{definition}{Definition}[section]
\newtheorem{example}{Example}
\newtheorem{claim}{Claim}
\newtheorem{assumption}{Assumption}
\newtheorem{condition}{Condition}
\newcommand{\argmax}{\operatornamewithlimits{argmax}}
\newcommand{\argmin}{\operatornamewithlimits{argmin}}

\title[Revealed Price Preference]{Revealed Price Preference: Theory and Empirical Analysis}

\begin{abstract}
To determine the welfare implications of price changes in demand data, we introduce a revealed preference relation over prices. We show that the absence of cycles in this relation characterizes a consumer who trades off the utility of consumption against the disutility of expenditure. Our model can be applied whenever a consumer's demand over a strict subset of all available goods is being analyzed; it can also be extended to settings with discrete goods and nonlinear prices.  To illustrate its use, we apply our model to a single-agent data set and to a data set with repeated cross-sections. We develop a novel test of linear hypotheses on partially identified parameters to estimate the proportion of the population who are revealed better off due to a price change in the latter application. This new technique can be used for nonparametric counterfactual analysis more broadly.
\end{abstract}

%With the aim of determining the welfare implications of price change in consumption data, we introduce a revealed preference relation over prices. We show that an absence of cycles in this preference relation characterizes a model of demand where consumers trade-off the utility of consumption against the disutility of expenditure. This model is appropriate whenever a consumer's demand over a {\em strict} subset of all available goods is being analyzed. For the random utility extension of the model, we devise nonparametric statistical procedures for testing and welfare comparisons. The latter requires the development of novel tests of linear hypotheses for partially identified parameters. In doing so, we provide new algorithms for the calculation and statistical inference in nonparametric counterfactual analysis for a general partially identified model. Our empirical applications provide support for the model and demonstrate how our methodology can be used for welfare analysis.

\maketitle
%\thispagestyle{fancy}
%\markboth{}{\scshape Extended Abstract}

%\tableofcontents
%\newpage

\section{Introduction}

A central question in economic analysis is the determination of the welfare effect of price changes. As an example, suppose we observe a consumer's purchases of two goods, gasoline and food, from two separate trips to a grocery store with an on site gasoline retailer. In the first instance $t$, the prices are $p^t=(2,2)$ of gasoline and food respectively and she buys a bundle $x^t=(6,3)$. In her second trip $t'$, the prices are  $p^{t'}=(3,1)$ and she purchases $x^{t'}=(5,4)$.  The most basic welfare question one can ask here is whether the consumer is better off at the prices prevailing at $t$ or at $t'$ (keeping fixed the prices of all other goods she consumes)? In this paper, we introduce a theoretical framework based on revealed preference, along with a nonparametric econometric technique, that would allow us to answer questions of this type.

A typical approach to this problem is to model the consumer as having a quasilinear utility function $\widetilde{U}(x)-p\cdot x$ since, in particular, this allows for simple ``sufficient statistics'' analysis of welfare gains or losses using a Harberger formula (see \citet{chetty2009} and most recently \citet{kleven2021} for an overview of this approach). Of course, the second term ($-p\cdot x$) in the quasilinear utility function captures the fact that the goods being analyzed (food and gasoline in our simple example) do {\em not} constitute the universe of the consumer's consumption; expenditure lowers utility because it reduces the consumption of an outside (numeraire) good.

The point of departure of our analysis is the following simple observation. Without having to model the consumer's preference as quasilinear (or taking any other specific functional form) we can still conclude that she is better off at $t$ compared to $t'$.  This is because $p^{t'}\cdot x^{t'}=19$ whereas $p^{t}\cdot x^{t'}=18$.  In other words, if the prices prevailing at $t'$ were $p^t$ instead of $p^{t'}$, the consumer would be better off since purchasing the same bundle $x^{t'}$ would cost less, leaving the consumer with more money to buy other goods (outside the set of goods analyzed).\footnote{Another way of seeing this is the following.  Suppose $t'$ is a supermarket where the prices are $p^{t'}$ and we observe the bundle $x^{t'}$ being bought by a consumer.  If at supermarket $t$, the prices are $p^t$, then we know that the consumer would prefer this supermarket, since the same purchases at $t'$ would cost less at $t$.}   More generally, the consumer has a {\em preference over prices} that an analyst could at least partially discern from the data: if at observations $t$ and $t'$, we find that $p^{t}\cdot x^{t'}\leq (<) p^{t'}\cdot x^{t'}$, then
\begin{center}
	{\em the consumer has revealed that she (strictly) prefers the price $p^{t}$ to the price $p^{t'}$.}
\end{center}
Welfare comparisons made in this way will only be consistent if the revealed preference relation over prices is free of cycles, a property we call the {\em generalized axiom of price preference} (GAPP). This leads inevitably to the following question: precisely what does GAPP mean for consumer behavior?\vspace{0.1in}

\noindent {\bf Augmented Utility Functions:}\,  To answer this question, we assume that the analyst collects a data set $\mathcal{D}=\{p^t,x^t\}_{t=1}^T$ from a consumer; each observation $t$ consists of the prices $p^t\in\Re^L_{++}$ of $L$ goods (representing some but not all the goods she consumes) and the consumer's demand $x^t\in\mathbb{R}^L_{+}$ at those prices.  We show that GAPP (on $\mathcal{D}$) is both necessary and sufficient for the existence of a strictly increasing function $U:\mathbb{R}^L_+\times \mathbb{R}_{-}\to \mathbb{R}$ that {\em rationalizes $\mathcal{D}$} in the following sense:
$$x^t\in\argmax_{x\in\mathbb{R}^L_+}U(x,-p^t\cdot x)\:\mbox{ for all $t=1,2,...,T$.}$$
The function $U$ should be interpreted as an {\em expenditure-augmented utility function}, where $U(x,-e)$ is the consumer's utility when she acquires $x$ at the cost of $e$. It recognizes that the consumer's expenditure on the observed goods is endogenous and dependent on prices: she could in principle spend more than what she actually spent (note that she optimizes over $x\in \mathbb{R}^L_{+}$) but the trade-off is the dis-utility of greater expenditure. Note that the quasilinear utility function $U(x,-p\cdot x)=\widetilde U(x)-p\cdot x$ is a special case of an augmented utility function.\vspace{0.1in}

The augmented utility model has a number of features that makes it widely applicable and easy to use.  We highlight a few of them.

(1)\, Being more general than the quasilinear model, it does not have some of its overly strong implications on the structure of consumer demand (see \refsec{sec:AUQL}).  In particular, it is broad enough to accommodate phenomena emphasized in the behavioral economics literature, such as reference dependence, mental budgeting and inattention to prices. We briefly describe the first of these here, a more detailed discussion can be found in \refsec{sec:behavioral}. \citet{koszegi2006} and \citet{heidhues2008} argue that consumption decisions can depend not just on the actual prices but also on the prices the consumer \textit{expected} to pay. Specifically, the disutility from spending is greater if the expected price was lower than the sticker price and vice versa.  A simple way they propose of capturing this phenomenon is the following function
\begin{equation*}\label{eq:ref_dep}
U(x,-px)=\widetilde{U}(x)-p\cdot x - F(p\cdot x - \tilde{p}\cdot x).
\end{equation*}
The first two terms capture standard quasilinear preferences whereas the third term captures a general form of reference dependence.\footnote{For a related model of reference prices leading to a similar functional form, see \citet{sakovics2011}.}  In Koszegi and Rabin's terminology, the consumer gets ``gain-loss utility'' by comparing the expenditure $p\cdot x$ she incurs on a bundle $x$ against the expenditure $\tilde{p} \cdot x$ she expected to incur, where $\tilde{p}$ are her reference prices.\footnote{A common choice for $F$ is $F(p\cdot x - \tilde{p}\cdot x)=\max\{\bar{k}(p\cdot x - \tilde{p}\cdot x),0\}+\min\{\underline{k}(p\cdot x - \tilde{p}\cdot x),0\}$ where it is typically assumed that $\bar{k}>\underline{k}>0$ or that  the consumer feels losses relative to the reference point more severely than commensurate gains.}

(2)\,  In this model, a consumer's utility at prices $p$ is given by $\max_{x\in\mathbb{R}^L_+}U(x,-p\cdot x)$, which obviously leads to a ranking or preference on prices.  Going further, it is possible to develop notions analogous to compensating and equivalent variations, which gives us a {\em quantitative} sense of how much one set of prices is ranked above another and could form the basis for interpersonal comparisons (see \refsec{sec:compensation}).

(3)\,  Readers familiar with \refAfriat{thm:Afriat} \citep{afriat1967} will no doubt have already noticed that we are working in a similar framework.  That theorem characterizes a data set $\mathcal{D}=\{p^t,x^t\}_{t=1}^T$ that could be rationalized in the following sense: there is $\widetilde U:\mathbb{R}^L_+\to\mathbb{R}$ such that $\widetilde U(x^t)\geq \widetilde U(x)$ for all $x\in\mathbb{R}^L_+$ that satisfy $p^t\cdot x\leq p^t\cdot x^t$.  The notion of rationalization in our model is distinct from that in \refAfriat{thm:Afriat} (even the utility functions have different domains) and there are data sets that could be rationalized in one sense but not the other.  We explain these differences in \refsecs{sec:Afrsection} and \ref{sec:GAPPvsGARP}.

Empirical researchers who apply \refAfriat{thm:Afriat} must contend with cases where a data set is not exactly rationalizable. They have developed an easily interpretable way of measuring how close a data set is to being rationalized known as the {\em critical cost efficiency index}. In \refsec{sec:index} we develop a similarly intuitive index that should facilitate empirical applications of the augmented utility model.

(4)\,  Our notion of revealed preference over prices is not simply applicable to a Euclidean consumption space.  It applies even when goods can only be consumed in discrete quantities (as is often the case in empirical IO models) or when they are represented by characteristics.   Furthermore,  when prices are nonlinear, it is still possible to compare price systems by asking if an agent could replicate the purchases under one system in another price system.  Requiring non-cycling comparisons in this case leads to a natural extension of GAPP and the augmented utility model, which we explain in  \refsec{sec:nonlinearGAPP}. \vspace{0.1in}

%After characterizing the model, we then show that it has a number of attractive features. First, we demonstrate that, in addition to being able to %rank prices via the revealed price preference relation, we can quantify the welfare impact of price changes in the data. Specifically, we show how %tight bounds can be derived for analogues of the standard compensating and equivalent variations. Second, for data sets that do not satisfy GAPP, we %provide a simple way to measure the extent of the violation. Finally, we show that the model can be generalized to accommodate nonlinear prices, %discrete choice and characteristics models.

\noindent {\bf Random Augmented Utility Model (RAUM):}\, In the second part of the paper, we develop the random version of the augmented utility model, in order to study the demand distribution of a population of consumers drawn from repeated cross-sectional data.  We first devise a test to check if the data are consistent with the RAUM.  We then develop a procedure to \textit{estimate} the \textit{proportion} of consumers who are made better or worse off by a given change in prices; welfare analysis of this kind under general preference heterogeneity is a challenging empirical issue, and has attracted considerable recent research (see, for example, \citet{hausman2016} and its references).

Unlike the case of data collected from a single individual, it is worth noting that, in this case, both model testing and welfare analysis are statistical since we need to account for sampling error inherent in repeated cross sectional data. Our RAUM test uses existing (though recently developed) econometric methods.  On the other hand, to carry out the welfare analysis, we develop new theoretical econometric results; it is worth stressing that this is a stand alone contribution that has applications beyond this paper.

For reasons we shall now explain, testing the RAUM on actual repeated cross-sectional data (such as household survey data) turns out to be a lot more straightforward than testing the random version of the standard budget-constrained utility model where the population is required to be rational in the sense of \refAfriat{thm:Afriat} (defined earlier). We refer to the latter model as RUM (random utility model) for short.  The test for RUM is broadly set out in \citet{McFadden1991}, but two challenges must first be overcome.  First, \citet{McFadden1991} do not account for finite sample issues as they assume that the econometrician observes the population distributions of demand; this hurdle was recently overcome by \citet{kitamura2018} who develop a testing procedure which incorporates sampling error. Second, the test suggested by \citet{McFadden1991} requires the observation of large samples of consumers who face not only the same prices but also make \textit{identical} total  expenditures. This feature is not true of any real observational data where a consumer's demand (and thus total expenditure) on a set of observed goods will typically be price dependent. Thus to implement their test, \citet{kitamura2018} need to first estimate demand distributions at a fixed level of (median) expenditure, which requires the use an instrumental variable technique (with all its attendant assumptions) to adjust for the endogeneity of observed total expenditure. \vspace{0.1in}

In contrast, the RAUM can be tested directly on household survey data, even when the demand distribution at a given price vector implies {\em heterogenous levels of total expenditure across consumers}.\footnote{A bit more formally, it is possible for two demand bundles $x$ and $x'$ in the support of the demand distribution when prices are $p^t$ to satisfy $p^t\cdot x\neq p^t\cdot x'$.}  This allows us to estimate the demand distribution by simply using sample frequencies and we can avoid the above-mentioned additional layer of demand estimation needed for testing RUM.

The reason for this remarkable simplification is somewhat ironic: we show that a data set is consistent with the RAUM if, and only if, a converted version of the data set (which results in identical expenditures at each price) of the type envisaged by \citet{McFadden1991} passes the RUM test suggested by them. In other words, we apply the test suggested by \citet{McFadden1991}, but not for the model they have in mind. This trick also means that we can use, and in a more straightforward way, the econometric techniques in \citet{kitamura2018}.

Assuming that a data set is consistent with our model, we can then evaluate the welfare impact of an observed change in prices. Indeed, if we observe the true distribution of demand at each price, it is possible to impose bounds (based on theory) on the proportion of the population who are revealed better off or worse off following an observed change in prices. Of course, when samples are finite, these bounds instead have to be estimated. To do so, we develop new econometric techniques that allow us to form confidence intervals on the proportion of consumers who are better or worse off; these techniques build on the econometric theory in \citet{kitamura2018} but are distinct from it.

We emphasize that these new econometric techniques can be more generally applied to linear hypothesis testing of parameter vectors that are partially identified, even in models that are unrelated to demand theory (see, for example, \citet{lazzati2018}).  They provide a new method for estimation and inference in nonparametric counterfactual analysis and, since the evaluation of counterfactuals is an important goal of empirical research, they are potentially very useful to practitioners.\vspace{0.1in}

\noindent {\bf Empirical Applications:}\, We use separate data sets to demonstrate how welfare analysis can be done using both the deterministic and random versions of our model. First, we use the deterministic augmented utility model to analyze panel data from the Mexican conditional cash transfer program Progresa. Recently, \citet{attanasio2020} showed that sellers responded to these transfers by altering the nonlinear prices they charge for staples.  We focus our analysis on the untreated households that did not receive cash transfers; we show, via revealed preference over the nonlinear price systems, that these households have tended to benefit from the price changes that occurred during the observation period. This is consistent with the finding in \citet{attanasio2020} that the change in the wealth distribution induced by Progresa led to larger quantity discounts (which favored the untreated households because they are usually better-off and consumed more).

Finally, we show how the RAUM can be used to estimate the welfare impact of the changes in observed prices in repeated cross-sectional data. Specifically, we take the model to two separate national household expenditure data sets from Canada and the U.K. and show that we can meaningfully estimate bounds on the percentage of households who are better and worse off. Even though these bounds are typically only partially identified, the estimated bounds are almost always narrower than ten percentage points and often substantially narrower than that. This demonstrates how to operationalize our novel econometric methodology to conduct inference for counterfactuals.

\section{The Deterministic Model}\label{sec:det_model}

We consider an econometrician who is studying a consumer's demand for $L$ goods.  We assume an idealized environment suitable for partial equilibrium analysis, where the consumer's demand for these goods at different prices are observed, while the consumer's wealth and the prices of other goods are held fixed.\footnote{Under fairly standard (but strong) assumptions, changes to the external environment can be precisely justified by deflating the prices of the $L$ goods (see \refsec{sec:deflprices}).}

Specifically, the econometrician collects a data set with a finite number of observations; each observation $t$ can be represented as $(p^t,x^t)$, where  $p^t\in \Re^L_{++}$ are the prices of the $L$ goods and $x^t\in \Re^L_+$ is the bundle of those goods purchased by the consumer.\footnote{We postpone the discussion of discrete consumption spaces and nonlinear pricing to \refsec{sec:nonlinearGAPP}.}  We denote the data set by $\mathcal{D}:= \{(p^t,x^t)\}_{t=1}^T$. (We shall slightly abuse notation and use $T$ to refer both to the (finite) number of observations and to the set $\{1,\dots,T\}$; similarly, $L$ could denote both the number, and the set, of commodities.)

We begin with a basic question: given ${\mathcal D}$, can the econometrician sign the welfare impact of a price change from $p^t$ to $p^{t'}$? Perhaps the most intuitive welfare comparison that can be made in this setting is as follows: if at prices $p^{t'}$, the econometrician finds that $p^{t'}\cdot x^t< p^t\cdot x^t$ then he may conclude that the agent is better off at the price vector $p^{t'}$ compared to $p^t$. This is because, at the price $p^{t'}$ the consumer can, if she wishes, buy the bundle bought at $p^t$ and she would still have money left over to buy other things, so she must be strictly better off at $p^{t'}$.  This ranking is eminently sensible, but can it lead to inconsistencies?

\begin{example}\label{eg:GARPnotGAPP}
	Consider a two observation data set
	\begin{equation*}
	p^{t}=(2,1)\, , \, x^{t}=(4,0) \; \text{ and } \; p^{t'}=(1,2)\, , \, x^{t'}=(0,1).
	\end{equation*}
	Since $p^{t'}\cdot x^t< p^t\cdot x^t$, it seems that the consumer is better off at prices $p^{t'}$ than at $p^t$; however, it is also true that $p^{t'}\cdot x^{t'}> p^t\cdot x^{t'}$, which gives the opposite conclusion.
\end{example}

\medskip

This example shows that for an econometrician to be able to consistently compare the consumer's welfare at different prices, some restriction has to be imposed on the data set.  To be precise, define the binary relations $\succeq_p$ and $\succ_p$ on ${\mathcal P}:=\{p^t\}_{t\in T}$, that is, the set of price vectors observed in $\mathcal D$, in the following manner:
$$p^{t'} \succeq_p (\succ_p) p^t \text{ if } p^{t'}\cdot x^{t}\leq (<)p^{t}\cdot x^{t}.$$
We say that price $p^{t'}$ is {\em directly (strictly) revealed preferred} to $p^t$ if $p^{t'} \succeq_p (\succ_p) p^t$, that is, whenever the bundle $x^t$ is (strictly) cheaper at prices $p^{t'}$ than at prices $p^t$. We denote the transitive closure of $\succeq_p$ by $\succeq_p^*$, that is, for $p^{t'}$ and $p^t$ in $\mathcal P$, we have $p^{t'}\succeq_p^* p^t$ if there are $t_1$, $t_2$,...,$t_N$ in $T$ such that $p^{t'}\succeq_p p^{t_1}$, $p^{t_1}\succeq_p p^{t_2}$,..., $p^{t_{N-1}}\succeq_p p^{t_N}$, and $p^{t_N}\succeq_p p^{t}$; in this case we say that $p^{t'}$ is {\em revealed preferred} to $p^t$.  If anywhere along this sequence, it is possible to replace $\succeq_p$ with $\succ_p$ then we say that $p^{t'}$ is {\em revealed strictly preferred} to $p^t$ and denote that relation by $p^{t'}\succ^*_p p^{t}$.\footnote{Notice that it makes sense to write $\hat p\succeq_p p^t$ even if $\hat p$ is not in $\mathcal P$, since the demand at $\hat p$ is not needed in the definition revealed preference. Similarly, it is possible to define $\hat p\succ_p p^t$ and the transitive extensions $\hat p\succeq_p^* p^t$ and $\hat p\succ_p^* p^t$.  This observation is useful later on, in Sections \ref{sec:compensation} and \ref{sec:st_welfare}. \label{non-sample-observe}.}  The following restriction, which excludes circularity in the econometrician's assessment of the consumer's wellbeing at different prices, is a bare minimum condition to impose on $\mathcal{D}$.

\begin{definition}\label{def:GAPP}
The data set $\mathcal{D}=\{(p^t,x^t)\}_{t=1}^T$ satisfies the \textit{Generalized Axiom of Price Preference} or \emph{GAPP} if there are no observations $t,t'\in T$ such that $p^{t'}\succeq^*_p p^{t}$ and $p^{t}\succ^*_p p^{t'}$.
\end{definition}

This in turn leads naturally to the following question:  if a consumer's observed demand behavior obeys GAPP, what could we say about her decision making procedure?

\subsection{The Expenditure-Augmented Utility Model}\label{sec:exp_aug_model}

An \emph{expenditure-augmented utility function} (or simply, an {\em augmented utility function}) is a function $U:\Re^L_{+}\times \Re_{-}\to\mathbb{R}$, where $U(x,-e)$ is the consumer's utility when she spends $e$ to purchase the bundle $x$.  We require that $U(x,-e)$ is strictly increasing in the last argument (in other words, utility is strictly decreasing in expenditure), which captures the tradeoff the consumer faces between consuming $x$ and consuming other goods (outside the set $L$).

At a given price $p$, the consumer chooses a bundle $x$ to maximize $U(x,-p\cdot x)$.  We denote the {\em indirect utility at price $p$} by
\begin{equation}\label{Vee}
V(p):=\sup\nolimits_{x\in\Re_+^L} U(x,-p\cdot x).
\end{equation}
If the consumer's augmented utility maximization problem has a solution at every price vector $p\in\Re_{++}^L$, then $V$ is also defined at those prices and this induces a reflexive, transitive, and complete preference over prices in $\Re_{++}^L$.  \vspace{0.1in}

A data set $\mathcal{D}=\{(p^t,x^t)\}_{t=1}^T$ is \emph{rationalized by an augmented utility function} if there exists such a function $U:\Re_+^{L}\times \Re_{-}\to \Re$ with
\begin{align}\label{eqn:newU}
x^t \in \argmax\nolimits_{x\in\Re_+^L} U(x,-p^t\cdot x) \qquad \text{for all } t\in T.
\end{align}

It is straightforward to see that GAPP is necessary for a data set to be rationalized by an augmented utility function.  First, notice that if $p^{t'} \succeq_p p^t$, then $p^{t'} \cdot x^{t}\leq p^{t}\cdot x^t$, and so
$$V(p^{t'})\geq U(x^t,-p^{t'}\cdot x^t)\geq U(x^t,-p^{t}\cdot x^t)=V(p^t).$$
Furthermore, $U(x^t,-p^{t'}\cdot x^t)> U(x^t,-p^{t}\cdot x^t)$ if $p^{t'} \succ_p p^t$, and in that case $V(p^{t'})> V(p^t)$.  Suppose GAPP were not satisfied and there were two observations $t,t'\in T$ such that $p^{t'}\succeq^*_p p^{t}$ and $p^{t}\succ^*_p p^{t'}$.  Then there would exist   $t_1,t_2,\dots,t_N\in T$ such that
$$V(p^{t'})\geq V(p^{t_1})\geq\cdots \geq V(p^{t_N}) \geq V(p^t) > V(p^{t'})$$
which is impossible.

\medskip

Our main theoretical result, which we state next, also establishes the sufficiency of GAPP for rationalization. Moreover, the result states that whenever $\mathcal{D}$ can be rationalized, it can be rationalized by an augmented utility function $U$ with a list of properties that make it  convenient for analysis.

\begin{theorem}\label{thm:GAPP}
	Given a data set $\mathcal{D}=\{(p^t,x^t)\}_{t=1}^T$, the following are equivalent:
	\begin{enumerate}
		\item $\mathcal{D}$ is rationalized by an augmented utility function.
		\item $\mathcal{D}$ satisfies GAPP.
		\item $\mathcal{D}$ is rationalized by an augmented utility function $U$ that is strictly increasing, continuous, and concave. Moreover, $U$ is such that $\max_{x\in\Re_+^L} U(x,-p\cdot x)$ has a solution for all $p\in\Re_{++}^L$.
	\end{enumerate}
\end{theorem}

\subsection{\refAfriat{thm:Afriat} and Proof of \refthm{thm:GAPP}}  \label{sec:Afrsection}

Before presenting the proof of \refthm{thm:GAPP}, it is worth providing a short description of the standard theory of revealed preference and, \refAfriat{thm:Afriat}, its central result. This will be useful not just because we will invoke the result several times but also since it will serve as an important point of contrast for our axiom and results.

The standard theory due to \citet{afriat1967} is built formally on the same primitives as our model: a finite data set of prices and corresponding consumption bundles. Unlike our model however, it is assumed that the observed goods correspond to the \textit{universe} of the consumer's consumption. Formally, a data set $\mathcal{D}$ is said to be \emph{rationalized by a utility function} if there exists a locally nonsatiated\footnote{This means that at any bundle $x$ and open neighborhood of $x$, there is a bundle $x'$ in the neighborhood with strictly higher utility.} utility function $\widetilde{U}:\Re_+^L\to \Re$ such that
\begin{align} \label{eqn:oriU}
x^t \in \argmax\nolimits_{\left\{x\in\Re_+^L:\; p^t\cdot x\leq p^t\cdot x^t\right\} } \widetilde{U}(x) \qquad \text{for all } t\in T.
\end{align}
In words, this criterion asks whether there is a utility function defined over the $L$ observed goods such that the consumer is utility maximizing at every observation $t$ over the fixed budget $p^t\cdot x^t$ corresponding to the observed expenditure.

Of course, data sets (outside of laboratory data) almost never contain the universe of consumed goods and the consumer's \textit{true} budget set is not observed, especially when one takes into account the possibility of borrowing and saving. Given this, when checking if a data set $\mathcal{D}$ can be  rationalized in the sense of (\ref{eqn:oriU}), we are effectively testing whether the consumer is maximizing a sub-utility function $\widetilde U:\Re^L_+\to\Re$ defined specifically on those $L$ goods (or equivalently, has weakly separable preferences).

It should be clear that rationalization in the sense of (\ref{eqn:oriU}) is distinct from rationalization by an augmented utility function.  The augmented utility model specifically takes into account the impact of the prices of these $L$ goods on the consumption of other goods; it is {\em necessarily} a partial equilibrium model, and designed for partial equilibrium welfare analysis of the type carried out in empirical industrial organization or public economics. An example is the study of the welfare impact of a sales tax levied on a subset of goods.

It is possible that a data set $\mathcal{D}$ can be rationalized in both senses, but that does not hold in general.  The precise conditions needed for rationalization  by a utility function are given by \refAfriat{thm:Afriat}, which we now describe.

\medskip

{\em Revealed preference} in Afriat's setting is captured by two binary relations, $\succeq_x$ and $\succ_x$ which are defined on the set of chosen bundles observed in $\mathcal D$, that is, the set ${\mathcal X}:=\{x^t\}_{t\in T}$, as follows:
$$x^{t'} \succeq_x (\succ_x)\, x^t \text{ if } p^{t'}\cdot x^{t'} \geq (>)\, p^{t'}\cdot x^{t}.$$
We say that the bundle  $x^{t'}$ is {\em directly revealed (strictly) preferred} to $x^t$ if $x^{t'} \succeq_x (\succ_x)\, x^t$, that is, whenever the bundle $x^t$ is (strictly) cheaper at prices $p^{t'}$ than the bundle $x^{t'}$. This terminology is intuitive: if the agent is maximizing some locally nonsatiated utility function $\widetilde U:\Re_+^L\to\Re$, then $x^{t'} \succeq_x x^t$ ($x^{t'}\succ_x x^t$) must imply that $\widetilde U(x^{t'})\geq (>)\, \widetilde U(x^t)$.

We denote the transitive closure of $\succeq_x$ by $\succeq_x^*$, that is, for $x^{t'}$ and $x^t$ in $\mathcal X$, we have $x^{t'}\succeq_x^* x^t$ if there are $t_1$, $t_2$,...,$t_N$ in $T$ such that $x^{t'}\succeq_x x^{t_1}$, $x^{t_1}\succeq_x x^{t_2},\dots,x^{t_{N-1}}\succeq_x x^{t_N} \succeq_x x^{t}$, and $x^{t_N}\succeq_x x^{t}$; in this case, we say that $x^{t'}$ is {\em revealed preferred} to $x^t$.  If anywhere along this sequence, it is possible to replace $\succeq_x$ with $\succ_x$ then we say that $x^{t'}$ is {\em revealed strictly preferred} to $x^t$ and denote that relation by $x^{t'}\succ^*_x x^{t}$. Clearly, if $\mathcal D$ is rationalizable by some locally nonsatiated utility function $\widetilde U$, then  $x^{t'}\succeq^*_x (\succ^*_x)\, x^t$ implies that $\widetilde U(x^{t'})\geq (>)\, \widetilde U(x^t)$. This observation in turn implies that a necessary condition for rationalization by a utility function is that the revealed preference relation has no cycles.

\begin{definition}\label{def:GARP}
	A data set $\mathcal{D}=\{(p^t,x^t)\}_{t=1}^T$ satisfies the \textit{Generalized Axiom of Revealed Preference} or \emph{GARP} if there are no observations $t,t'\in T$ such that $x^{t'}\succeq^*_x x^{t}$ and $x^{t}\succ^*_x x^{t'}$.
\end{definition}

The main insight of \refAfriat{thm:Afriat} is to show that this condition is also \textit{sufficient} (the formal statement can be found in the online \refappendix{sec:afriat}).

\medskip

Having described \refAfriat{thm:Afriat}, we are now in a position to prove \refthm{thm:GAPP}.

\medskip

\noindent {\sc \textbf{Proof of \refthm{thm:GAPP}.}}
We will show that $(2)\implies (3)$. We have already argued that $(1)\implies (2)$ and $(3)\implies (1)$ by definition.

Choose a number $M>\max_{t} p^t\cdot x^{t}$ and define the augmented data set $\widetilde{\mathcal{D}}=\{(p^t,1),(x^t,M-p^t\cdot x^t)\}_{t=1}^T$.  This data set augments $\mathcal D$ since we have introduced an $L+1^{\text{th}}$ good, which we have priced at 1 across all observations, with the demand for this good equal to $M-p^t\cdot x^t$.

The crucial observation to make here is that
$$(p^t,1) (x^t,M-p^t\cdot x^t) \geq (p^t,1) (x^{t'},M-p^{t'}\cdot x^{t'}) \text{ if and only if } p^{t'}\cdot x^{t'} \geq p^t\cdot x^{t'},$$
which means that
$$(x^t,M-p^t\cdot x^t) \succeq_x  (x^{t'},M-p^{t'}\cdot x^{t'})  \text{ if and only if } p^t\succeq_p  p^{t'}.$$
Similarly,
$$(p^t,1) (x^t,M-p^t\cdot x^t) > (p^t,1) (x^{t'},M-p^{t'}\cdot x^{t'}) \text{ if and only if } p^{t'}\cdot x^{t'} > p^t\cdot  x^{t'},$$
and so
$$(x^t,M-p^t\cdot x^t) \succ_x (x^{t'},M-p^{t'}\cdot x^{t'})  \text{ if and only if } p^t \succ_p p^{t'}.$$
Consequently, $\mathcal{D}$ satisfies GAPP if and only if $\widetilde{\mathcal{D}}$ satisfies GARP. Applying \refAfriat{thm:Afriat} when $\widetilde{\mathcal{D}}$ satisfies GARP, there is $\widetilde{U}:\Re^{L+1}\rightarrow \Re$ (notice that $\widetilde U$ is defined on $\Re^{L+1}$ and not just $\Re^{L+1}_{+}$; see Remark 3 in \refappendix{sec:afriat}) such that
\begin{equation}\label{justify1}
(x^t,M-p^t\cdot x^t) \in \argmax\nolimits_{\{(x,m)\in \Re^L_+\times \Re\, :\, p^t\cdot x + m\leq M\}} \widetilde{U}(x,m) \qquad \text{for all } t\in T.
\end{equation}
The function $\widetilde U$ can be chosen to be strictly increasing, continuous, and concave, and the lower envelope of a finite set of affine functions.  Clearly, the augmented utility function $\widebar U:\Re^L_+\times \Re_{-}\rightarrow \Re$ defined by $\widebar U(x,-e):=\widetilde{U}(x,M-e)$ is strictly increasing in $(x,-e)$, continuous, concave and rationalizes $\mathcal{D}$.

Define $\widehat U:\Re_+^L\times \Re_{-}\rightarrow \Re$ by
\begin{equation}\label{eq:UoX}
	\widehat U(x,-e):= \widebar U(x,-e)-h(\max\{0,e-M\}),
\end{equation}where $h:\Re_+\rightarrow \Re$ is a differentiable function satisfying $h(0)=0$, $h'(k)>0$, $h''(k)\geq 0$ for $k\in \Re_+$, and $\lim_{k\rightarrow \infty} h'(k)=\infty$. (For example, $h(k)=k^3$.)  Like $\widebar U$, the function $\widehat U$ is strictly increasing in $(x,-e)$, continuous and concave and $x^t$ solves $\max_{x\in \Re^L_+} \widehat U(x,-p^t  x)$ (because $\widehat U(x,-e)\leq \widebar U(x,-e)$ for all $(x,-e)$, and $\widehat U(x^t,-p^t\cdot x^t)=\widebar U(x^t,-p^t\cdot x^t)$). Furthermore, for every $p\in\Re^L_{++}$, $\argmax_{x\in X} \widehat U(x,-p\cdot x)$ is nonempty.\footnote{Choose a sequence $x^n\in \Re^L_+$ such that $\widehat U(x^n,-p\cdot x^n)$ tends to $\sup_{x\in \Re^L_+} \widehat U(x,-p\cdot x)$ (which we allow to be infinity).  It is impossible for  $p\cdot x^n\rightarrow \infty$ because the piecewise linearity of $U(x,-e)$ in $x$ and the assumption that $\lim_{k\to \infty}h'(k)\to\infty$ implies that $\widehat U(x^{n},-p\cdot x^n)\rightarrow -\infty$. So the sequence $p\cdot x^n$ is bounded, which in turn means that there is a subsequence of $x^n$ that converges to $x^\star\in \Re^L_+$.  By the continuity of $\widehat U$, we obtain $\widehat U(x^\star,-p\cdot x^\star)= \sup_{x\in \Re^L_+} \widehat U(x,-p\cdot x)$.} \hfill $\blacksquare$

\medskip

We end this section by noting that GARP imposes testable restriction distinct from GAPP. This is immediate from \refeg{eg:GARPnotGAPP} and can be seen from \reffig{fig:GARPnotGAPP_A} which plots not just the observed consumption bundles but also the corresponding budget sets (derived from the observed prices and expenditures).
\begin{figure}[h]
	\includegraphics[trim=3cm 11cm 6.3cm 6.9cm,clip=true,scale=.75,page=3]{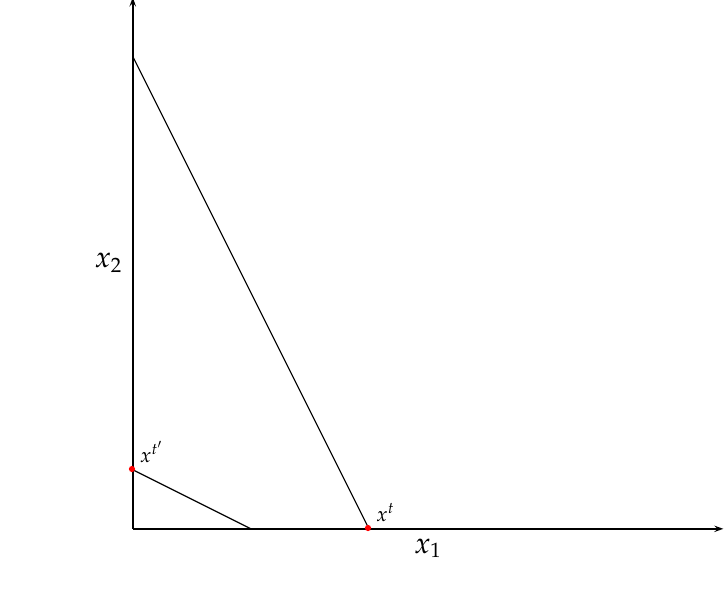}
	\caption{Choices that do not allow for consistent welfare predictions but satisfy GARP.}
	\label{fig:GARPnotGAPP_A}
\end{figure}
As we argued, GAPP does not hold in this example but, since the budget sets do not even cross, it is immediate to conclude that GARP does. We defer the description of the exact relation between the two criteria to \refsec{sec:GAPPvsGARP}.

From this point onwards, when we refer to `rationalization' without additional qualifiers, we shall mean rationalization by an augmented utility function, that is, in the sense given by (\ref{eqn:newU}) rather than in the sense given by (\ref{eqn:oriU}).

Up to now, we have motivated our model by showing that it is the utility representation of a basic axiom requiring consistent price comparisons. In the next two subsections, we provide direct motivation for the augmented utility function itself by arguing that it contains, as special cases, several distinct (standard and behavioral) preference-modeling approaches.

\subsection{`Standard' consumer theory and the augmented utility function}\label{sec:AUQL}

Perhaps the clearest motivation for our model is to think of it as a generalization of the quasilinear utility model, in which the consumer derives utility $\widetilde{U}(x)$ from the bundle $x$ and maximizes utility net of expenditure, that is, she chooses $x$ to maximize
\begin{equation} \label{yuyee}
	U(x,-e):=\widetilde{U}(x)-e,
\end{equation}
where $e=p\cdot x$.  There is a familiar textbook way of justifying this objective function by fitting it within the constrained optimization model of standard consumer theory. This is to think of the consumer as having a utility function $\widebar U$ defined over $L+1$ goods, with the last `outside' good entering additively and linearly into the utility function, so that $\widebar U(x,z)=\widetilde U(x)+z$. Assuming that the consumer has a total wealth of $W$, the utility of purchasing a bundle $x\in\mathbb{R}^L_+$ is then
$$\widebar U(x,W-p\cdot x)=\widetilde U(x)-p\cdot x+W.$$
Ignoring boundary issues, the consumer is effectively maximizing (\ref{yuyee}).

Even though the quasilinear model is widely used in partial equilibrium analysis, it is well known that the complete absence of income effects makes it unsuitable for certain empirical applications. For this reason, it is also common to remove the linear structure on $\widebar U$ while retaining the assumption that all outside consumption opportunities can be represented by a {\em single} outside good; this is true, for example, in the literature on modeling the demand for differentiated goods.\footnote{For example, in \citet{berry1995} and in \citet{nevo2000}, $\widebar U$ is additively separable between the $L$ goods and the outside good; in the former, the utility of consuming $y$ units of the outside good is $\alpha \ln y$, for some $\alpha>0$, whereas in the latter it is $\alpha y$ (in other words, $\widebar U$ is quasilinear). In \citet{bhattacharya2015}, $\widebar U$ is allowed to be a general function defined on $L+1$ goods. In models of differentiated goods, the consumption space is typically assumed to be discrete rather than $\Re^L_+$, but the augmented utility model is still applicable in that context (see \refsec{sec:nonlinearGAPP}). \label{discretefoot}}  In this case, the utility of purchasing a bundle $x\in\mathbb{R}^L_+$ is $\widebar U(x,W-p\cdot x)$; notice that, provided $W$ is fixed, we can think of the consumer as maximizing an augmented utility function: simply let $U(x,-e)=\widebar U(x,W-e)$.

Obviously, a consumer's outside consumption opportunities would in reality involve more than one good, and the prices of those outside goods could change as well. Within the familiar constrained-optimal model of consumer theory, there are known conditions that justify the representation of  those consumption opportunities by a representative good (with its corresponding price index).   This is explained in detail in  \refsec{sec:deflprices}.

Finally, it is worth mentioning that the augmented utility function could also capture, as a special case, quasilinear utility maximization subject to certain constraints. One such example is consumption with a subsistence constraint, which we describe in more detail in the empirical application in \refsec{sec:progresa}. Loosely speaking, we can capture constraints on $(x,-e)$ with an augmented utility $U(x,-e)$ that assigns very low values at $(x,-e)$ that violate the constraint.

\subsection{Behavioral preferences captured by the augmented-utility model}\label{sec:behavioral}

The central feature of the augmented utility model is that consumers experience disutility from expenditure.  As we explained in the previous subsection, this disutility could be interpreted in a purely opportunity cost sense -- more expenditure on the consumed goods imply less money available for other goods.  In this understanding, the augmented utility function is a reduced form of a broader 'true' utility function defined on all goods.

However, it is also reasonable to think of the augmented utility function in another way: that the consumer has -- \textit{directly} -- a preference over bundles of the observed $L$ goods and their associated expenditure, which she has developed as a way of guiding her purchasing decisions.  Thus it is the basic object of analysis and not the reduced form of something more fundamental.  This understanding of choice behavior is exploited in the behavioral economics literature and the following quote from \citet{prelec1998} is effectively a description of the augmented utility function:
\begin{quote}
	each time a consumer engages in an episode of consumption, we assume she asks herself: ``How much is this pleasure costing me?'' The answer to this question is the imputed cost of consumption. This imputed cost is ``real'' in the sense that it actually detracts from consumption pleasure.
\end{quote}
In this understanding, the disutility of expenditure is still related to opportunity cost, but the relationship is more flexible than what is permitted in a classical framework.\footnote{Another paper that spells out remarkably clearly this approach to modelling consumer decisions is  \citet{friedman2015}, though the authors primarily have in mind the quasilinear utility model.}

In the Introduction, we described one example of behavioral preferences (reference-dependent preferences) that could be captured by an augmented utility function. In the remainder of this section, we describe how our model relates to two other prominent themes in the behavioral literature.\medskip

\begin{center}
	{\em Inattention to Prices and Expenditure}
\end{center}

The public economics literature following \citet{chettylooney2009} has observed that consumers often misperceive prices: in their context, shoppers at grocery stores do not internalize the price effect of taxes. This literature is summarized in a recent survey by \citet{gabaix2019} who argues that many behavioral biases often take the form of inattention. Our model naturally captures a version of the inattention to prices discussed in \citet{bordalo2013} and \citet{gabaix2014}. Here a consumer faced with a price $p$ perceives the expenditure associated with a bundle $x$ as $f(x,p\cdot x)$, where $f$ is increasing in the true expenditure $p\cdot x$ and could potentially depend on $x$.  With this misperception, and assuming that the consumer has a quasilinear preference, she then chooses $x\in\Re^L_+$ to maximize
\begin{equation}\label{eqn:inattention_pref}
U(x,-p\cdot x)=\widetilde{U}(x)-f(x,p\cdot x)
\end{equation}
A special case of this model is where a consumer has a default price $p^d$ and misperceives the actual price $p$ to be $a p + (1-a) p^d$ where $a\in[0,1]$ is the `attention parameter.'  The perceived expenditure is then $f(x,p\cdot x)=ap\cdot x+(1-a)p^d\cdot x$.  More generally, the model accommodates $f(x,p\cdot x)=a(x)p\cdot x+(1-a(x))p^d\cdot x$, where the attention parameter $a(x)\in [0,1]$ varies across bundles.\footnote{This formulation of perceived expenditure is more general than \citet{gabaix2014} in that it allows the attention parameter $a$ to depend on $x$ but is less general in that the parameter does not vary across goods.}  This is a natural extension since, among other things, it allows a consumer to be more attentive to her actual expenditure if she is purchasing large bundles compared to small ones (so that $a(x)$ tends to 1 when $x$ is large). Yet another possibility is that the consumer is not completely sensitive to every dollar increase in expenditure but pays more attention only when certain thresholds are crossed; this would correspond to the case where $f$ depends only on the expenditure $e=p\cdot x$ and has the shape of a step function of expenditure.

Clearly, inattention as modeled by (\ref{eqn:inattention_pref}) is an instance where the agent has an augmented utility function, even though it will typically not be quasilinear (in actual expenditure).

It is also worth mentioning that using an augmented utility function (such as \eqref{eqn:inattention_pref}) to capture price inattention is particularly apt because, as \citet{gabaix2014} notes, the numeraire serves as ``the shock absorber that adjusts to the budget constraint.'' The alternative is to model the consumer as having {\em both} price misperception over a given set of goods {\em and} a budget on those goods that must be satisfied, which inevitably leads to the added complication of modelling how the agent adjusts her intended demand when she realizes it violates (because prices are misperceived) the budget constraint at the true prices.\footnote{\citet{gabaix2014} proposes one way to deal with this issue.}

\newpage

\begin{center}
	{\em Budgeting and Mental Budgeting}
\end{center}

As we discussed in \refsec{sec:AUQL}, a common approach to partial equilibrium analysis is to add a numeraire as an additional good and assume that the agent has a (standard) utility function and budget set defined on the $L+1$ goods, with price and income information used to determine the level of the numeraire consumed. Of course such an approach could only work when income information is available and that is not always the case.\footnote{Several widely used data sets, such as supermarket scanner panel data, that contain rich information on purchases, do not have accurate measures of income. Here, income information is typically the category (income ranges) that households self report when applying for loyalty cards (and so the information becomes out of date).}  Even when this information is available, it is strictly speaking not the right value to use as the global budget if the consumer can save and borrow to a significant degree (as acknowledged, for example, in \citet{hausman2016}).  More generally, figuring out what really constitutes `the budget' is not always straightforward, even in a classical setting.

Regularities highlighted by behavioral economists add a further wrinkle to the concept of a budget. It has been widely observed that households do not always treat money as fungible and instead create separate accounts for various categories of goods \cite{thaler1999}. This is not only true for consumption decisions (see, for instance, \citet{hastings2013, hastings2018}) but also for savings decisions, which is why consumers often save more when they have access to commitment savings options (important theoretical and empirical contributions are \citet{amador2006} and \citet{feldman2010}, \citet{dupas2013} respectively).

Now consider a researcher who is trying to model the demand for a set of $L$ goods which form a subset of all the goods consumed by an agent. If mental accounting effects are important, the researcher will have to allow for the fact that he cannot observe how goods are categorized by the agent, nor does he know what really constitutes the mental budget from which the agent is drawing her expenditure (on the $L$ observed goods and their perceived alternatives). In this situation, the augmented utility framework provides a natural way to model the demand for those $L$ goods: it is consistent with constrained utility maximization incorporating an outside good (see \refsec{sec:AUQL}) but {\em does not require} the researcher to take a stand on the (unobserved, mental) budget from which the agent is drawing her expenditure.\footnote{Here we are assuming that the data $\mathcal{D}=\{(p^t,x^t)\}_{t\in T}$ are collected over a period where the mental budget for the $L$ observed goods and their alternatives is stable. Changing mental budgets would manifest itself as violations of GAPP (see \refeg{ex:mentalbudget}).}

\section{Properties of the augmented utility model}\label{sec:properties}

In this section, we explore various aspects of the augmented utility model, beginning with a discussion of the relationship between GAPP and GARP.  We then go on to discuss welfare analysis in the augmented utility framework. Since one would not expect data sets to be completely consistent with the augmented utility model, we discuss how departures from GAPP could be measured. Lastly, we discuss how prices could be deflated in this model to account for general changes in the price level.

\subsection{Comparing GAPP and GARP}\label{sec:GAPPvsGARP}

Recall that \refeg{eg:GARPnotGAPP} in \refsec{sec:det_model} is an example of a data set that obeys GARP but fails GAPP. We now present an example of a data set that satisfies GAPP but fails GARP.

\begin{example}\label{eg:GAPPnotGARP}
	Consider the data set consisting of the following two choices:
	\begin{equation*}
	p^{t}=(2,1)\, , \, x^{t}=(2,1) \; \text{ and } \; p^{t'}=(1,4)\, , \, x^{t'}=(0,2).
	\end{equation*}
	
	\begin{figure}[h]
		\includegraphics[trim=4.5cm 11cm 6.3cm 6.9cm,clip=true,scale=.75,page=1]{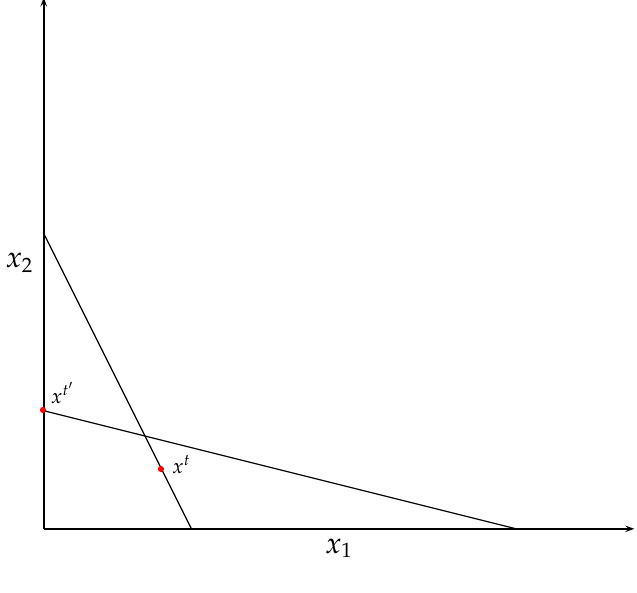}
		\caption{Choices that satisfy GAPP but not GARP}
		\label{fig:GAPPnotGARP_A}
	\end{figure}
These choices, as shown in \reffig{fig:GAPPnotGARP_A}, violate GARP as $p^{t}\cdot  x^{t}=5 > 2= p^{t}\cdot  x^{t'}$ ($x^t\succ_x x^{t'}$) and $p^{t'}\cdot  x^{t'}=8 > 6= p^{t'}\cdot x^{t}$ ($x^{t'}\succ_x x^{t}$). However, these choices satisfy GAPP as $p^{t'}\cdot x^{t'}=8 > 2=p^{t}\cdot x^{t'}$ ($p^t\succ_p p^{t'}$) but $p^{t}\cdot x^{t}=5 \ngeq 6 =p^{t'}\cdot x^{t}$ ($p^{t'}\nsucceq_p p^{t}$).
\end{example}

\medskip

While GAPP and GARP are not in general the same conditions, they coincide in any data set where $p^t\cdot x^t=1$ for all $t\in T$.   This is because  $x^t\succeq_x(\succ_x) x^{t'}$ if and only if $p^t\succeq_p (\succ_p) p^{t'}$ since both conditions are equivalent to $1\geq(>)\, p^t\cdot x^{t'}$.  Given a data set ${\mathcal D}=\{(p^t,x^t)\}_{t=1}^T$, we define the {\em iso-expenditure version of $\mathcal D$} as another data set $\breve {\mathcal D}:=\left\{(p^t,\breve{x}^t)\right\}_{t=1}^T$, such that $\breve{x}^t=x^t/(p^t\cdot x^t)$. This new data set has the feature that $p^t\cdot \breve{x}^t=1$ for all $t\in T$. Notice that the revealed price preference relations $\succeq_p$, $\succ_p$ remain unchanged when consumption bundles are scaled. Thus a data set obeys GAPP if and only if its iso-expenditure version obeys GAPP, which in this case is equivalent to GARP.\footnote{There is an analogous `GARP-version' of \refprop{prop:GAPPscaling} and that this observation (or some close variation of it) has been exploited before in the literature (see, for example, \citet{sakai1977}).  Suppose ${\mathcal D}=\{(p^t,x^t)\}_{t=1}^T$ obeys GARP.  Then GARP holds even if each observed price vector $p^t$ is arbitrarily scaled.  In particular, $\mathcal D$ obeys GARP if and only if $\hat{\mathcal D}=\{(\hat p^t,x^t)\}_{t\in T}$, where $\hat {p}^t=p^t/(p^t\cdot x^t)$, obeys GARP (equivalently, GAPP) since $\hat p^t\cdot x^t=1$ for all $t\in T$.  The latter perspective is useful because it highlights the possibility of applying \refAfriat{thm:Afriat} on $\hat{\mathcal D}$, {\em in the space of prices} (in other words, with the roles of prices and bundles reversed).  This immediately gives us a different, `dual' rationalization of $\mathcal D$ in terms of indirect utility, that is, there is a continuous, strictly decreasing, and convex function $\widetilde{V}:\mathbb{R}_{++}^L\to\mathbb{R}$ such that $\hat p^t\in {\arg\min}_{\{p\in\mathbb{R}_{++}^L\,:\,p\cdot x^t\geq 1\}}\widetilde{V}(p)$.  For an application of this observation, see \citet{brown2000}.}  The next proposition gives a more detailed statement of these observations.

\begin{proposition}\label{prop:GAPPscaling}
	Let ${\mathcal D}=\{(p^t,x^t)\}_{t=1}^T$ be a data set and let $\breve{\mathcal D}=\{(p^t,\breve{x}^t)\}_{t\in T}$, where $\breve{x}^t=x^t/(p^t\cdot x^t)$. Then the revealed preference relations $\succeq_p^*$ and $\succ_p^*$ on ${\mathcal P}=\{p^t\}_{t=1}^T$ and the revealed preference relations $\succeq_x^*$ and $\succ_x^*$ on ${\breve{\mathcal{X}}}=\{\breve{x}^t\}_{t=1}^T$ are related in the following manner:
	\begin{enumerate}
		\item $p^t\succeq^*_p p^{t'}$ if and only if $\breve{x}^t\succeq^*_x \breve{x}^{t'}$.
		\item $p^t\succ^*_p p^{t'}$ if and only if $\breve{x}^t\succ^*_x \breve{x}^{t'}$.
	\end{enumerate}
	As a consequence, $\mathcal D$ obeys GAPP if and only if its iso-expenditure version, $\breve{\mathcal D}$, obeys GARP.
\end{proposition}\medskip

\noindent {\bf Proof.}\, Notice that
$$p^t\cdot \frac{x^t}{p^t\cdot x^t} \geq  p^t\cdot \frac{x^{t'}}{p^{t'}\cdot x^{t'}} \, \iff \, p^{t'}\cdot x^{t'} \geq  p^t\cdot x^{t'}.$$
The left side of the equivalence says that $\breve{x}^t\succeq_x \breve{x}^{t'}$ while the right side says that $p^t\succeq_p p^{t'}$.  This implies (1) since $\succeq_p^*$ and $\succeq_x^*$ are the transitive closures of $\succeq_p$ and $\succeq_x$ respectively.  Similarly, it follows from
	$$p^t \cdot \frac{x^t}{p^t \cdot x^t} >  p^t \cdot \frac{x^{t'}}{p^{t'}\cdot x^{t'}} \, \iff \, p^{t'}\cdot x^{t'} >  p^t\cdot x^{t'}$$
	that $\breve{x}^t\succ_x \breve{x}^{t'}$ if and only if $p^t\succ_p p^{t'}$, which leads to (2). The claims (1) and (2) together guarantee that there is a sequence of observations in $\mathcal D$ that lead to a GAPP violation if and only if the analogous sequence in $\breve{\mathcal D}$ lead to a GARP violation. \hfill $\blacksquare$
\medskip

As an illustration, compare the data sets in \reffig{fig:GARPnotGAPP_A} and \reffig{fig:GAPPnotGARP_A} to the iso-expenditure data sets in \reffig{fig:GARPnotGAPP_B} and \reffig{fig:GAPPnotGARP_B}.  It can be clearly observed that the iso-expenditure data in \reffig{fig:GARPnotGAPP_B} contains a GARP violation (which implies it does not satisfy GAPP) whereas the data in \reffig{fig:GAPPnotGARP_B} does not violate GARP (and, hence, satisfies GAPP).

\begin{figure}[h]
	\centering
	\begin{subfigure}[b]{.49\linewidth}
		\includegraphics[trim=4.6cm 11cm 6.3cm 6.9cm,clip=true,scale=.7,page=4]{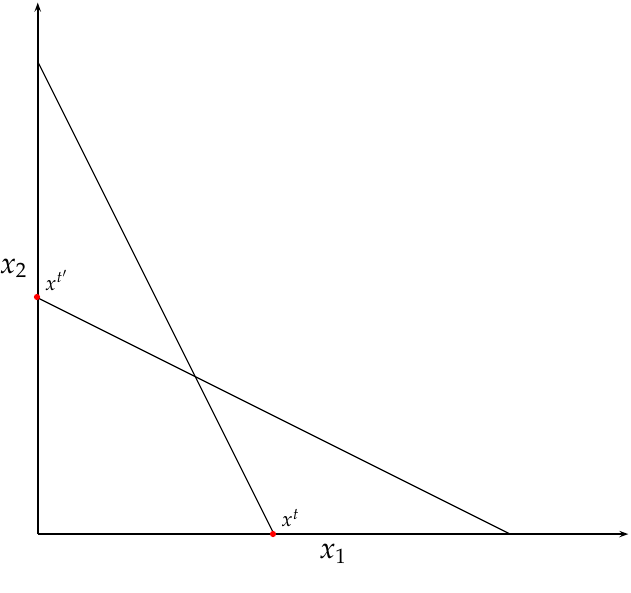}
		\caption{\refeg{eg:GARPnotGAPP}}
		\label{fig:GARPnotGAPP_B}
	\end{subfigure}
	\begin{subfigure}[b]{.49\linewidth}
		\includegraphics[trim=4.6cm 11cm 6.3cm 6.9cm,clip=true,scale=.7,page=2]{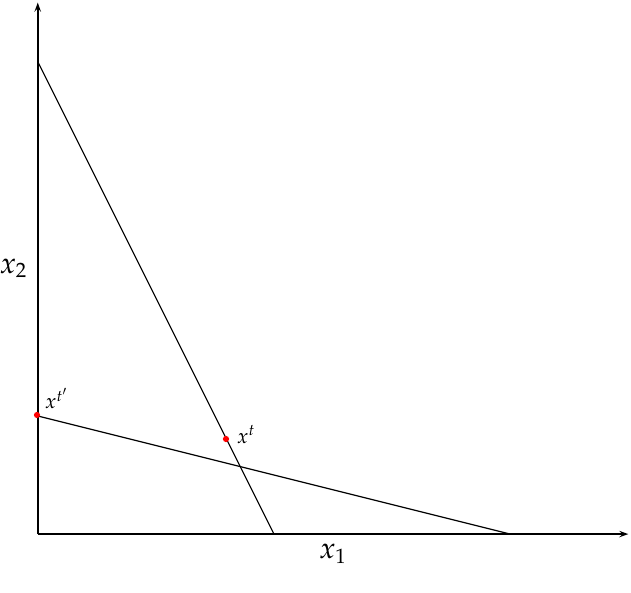}
		\caption{\refeg{eg:GAPPnotGARP}}
		\label{fig:GAPPnotGARP_B}
	\end{subfigure}
	\caption{Expenditure-Normalized Choices}
\end{figure}

A consequence of \refprop{prop:GAPPscaling} is that the augmented utility model can be tested in two ways: we can either test GAPP directly or we can test GARP on its iso-expenditure version.  If we are simply interested in testing GAPP on a single-agent data set $\mathcal D$, normalization brings no advantage: the test is computationally straightforward in either case and involves the construction of their (respective) revealed preference relations and checking for acyclicity. However, as we shall see in \refsec{sec:RAUMX}, iso-expenditure scaling plays an important role in the test we develop (on repeated cross-sectional demand data) for the random utility version of the augmented utility model.

While GARP and GAPP are distinct properties, they are not mutually exclusive and it is possible for a data set to satisfy both. For example, if $\mathcal{D}=\{(p^t,x^t)\}_{t=1}^T$ is collected from a consumer who is maximizing a quasilinear augmented utility function, then it will satisfy both GAPP and GARP.\footnote{When $U$ has the form (\ref{yuyee}), $x^t$ maximizes $U(x,-p^t\cdot x)$ only if $x^t$ maximizes $\widetilde U(x)$ in $\{x\in\Re^L_+:p^t\cdot x\leq p^t\cdot x^t\}$.  Thus $\mathcal{D}$ must also obey GARP. A broader class of augmented utility functions that satisfy both GAPP and GARP is given in \refsec{sec:modGPGP}.}  When both properties are satisfied, then an analyst could make use of either property when making predictions of demand at an out-of-sample price; the two properties will then typically lead to different set predictions.   We discuss this in greater detail in \refsec{sec:GAPPGARP} of the online appendix, which also contains more discussion of the relationship between revealed preferences under GAPP and under GARP.

In light of \refprop{prop:GAPPscaling} and the fact that the revealed price preference relation is not affected by scaling consumption bundles, it is natural to wonder about the relationship between the testable implications of the augmented-utility model and the constrained-optimization model (as in (\ref{eqn:oriU})) restricted to homothetic preferences. A data set that can be rationalized in the latter sense\footnote{For the precise characterization, see \citet{varian1983}.} will have the feature that it must satisfy GARP for \textit{any} arbitrary scaling of consumption bundles and thus will satisfy GAPP.  By contrast, a data set that satisfies GAPP must only satisfy GARP for the particular scaling that equalizes expenditure across observations.  In other words, GAPP is a less stringent property; that it is {\em strictly} less stringent is clear from \refeg{eg:GAPPnotGARP}, which satisfies GAPP but violates GARP and therefore cannot be rationalized in Afriat's sense (as given by (\ref{eqn:oriU})) for any locally nonsatiated preference, let alone a homothetic preference.\footnote{\refeg{ex:yat} in the online appendix contrasts demand predictions using the augmented utility model and the constrained-optimization model (both with and without imposing homotheticity on the preference).}
 \vspace*{0.05in}

\subsection{Preference over Prices}

We know from \refthm{thm:GAPP} that if $\mathcal D$ obeys GAPP then it can be rationalized by an augmented utility function with an indirect utility that is defined at all price vectors in $\Re_{++}^L$.  It is straightforward to check that any indirect utility function $V$ as defined by (\ref{Vee}) has the following two properties:
\begin{itemize}
\item[(a)]  it is {\em nonincreasing in $p$}, in the sense that if $p'\geq p$ (element by element) then $V(p')\leq V(p)$, and
\item[(b)]  it  is {\em quasiconvex in $p$}, in the sense that if $V(p)=V(p')$, then $V(\beta p+(1-\beta)p')\leq V(p)$ for any $\beta\in [0,1]$.
\end{itemize}

\medskip

Any rationalizable data set $\mathcal D$ could potentially be rationalized by many augmented utility functions, with each one leading to a different indirect utility function.  We denote this set of indirect utility functions by $\mathbf{V}({\mathcal D})$.  We have already observed that if $p^{t}\succeq_p^* (\succ_p^*)\, p^{t'}$ then $V(p^{t})\geq (>)\, V(p^{t'})$ for any $V\in\mathbf{V}({\mathcal D})$; in other words, the conclusion that the consumer prefers the prices $p^t$ to $p^{t'}$ is {\em nonparametric} in the sense that it is independent of the precise augmented utility function used to rationalize $\mathcal D$.  The next result (proved in  \refappendix{sec:pricewelfare}) says that, without further information on the augmented utility function, this is {\em all} the information on the consumer's preference over prices in $\mathcal P$ that we can glean from the data. Thus, in our nonparametric setting, the revealed price preference relation contains the most detailed information for welfare comparisons.

\begin{proposition}\label{prop:p_welfare}
	Suppose ${\mathcal D}=\{(p^t,x^t)\}_{t=1}^T$ is rationalizable by an augmented utility function. Then for any $p^{t}$, $p^{t'}$ in $\mathcal P$:
	\begin{enumerate}
		\item $p^t\succeq^*_p p^{t'}$ if and only if $V(p^t)\geq V(p^{t'})$ for all $V\in \mathbf{V}({\mathcal D})$.
		\item $p^t\succ_p^* p^{t'}$ if and only if $V(p^t)> V(p^{t'})$ for all $V\in \mathbf{V}({\mathcal D})$.
	\end{enumerate}
\end{proposition}\vspace*{0.05in}

\subsection{Compensation for a price change}\label{sec:compensation}

In standard consumer theory, the compensating and equivalent variations are two ways of quantifying the welfare impact of a price change (see \citet{mascolell1995}, Chapter 3.I). We now argue that analogues exist for the augmented utility model and that bounds for them can be recovered from the data.

Let $U$ be the consumer's augmented utility function. Suppose that the initial price is $p^{t_1}$ and it changes to $p^{t_2}$, leading to a change in consumption from $x^{t_1}$ to $x^{t_2}$.   Then we can find $\mu_c$ such that
\begin{equation}\label{CV}
\max\nolimits_{x\in \Re_+^L} U(x,-p^{t_2}\cdot x-\mu_c)=V(p^{t_1}).
\end{equation}
Note that $\mu_c$ is unique since $U$ is strictly increasing in the last argument. We could think of $\mu_c$ as the lump sum transferred {\em from} the consumer (if it is positive) or {\em to} the consumer (if it is negative) after the price change that will make her just indifferent between the situation before and after the change.

Suppose we interpret $U$ as arising from an overall utility function $\widetilde{U}(x,z)$ (that depends on the observed goods $x$ and the level $z$ of an outside good), given the consumer's wealth of $M$, so that $U(x,-e)=\widetilde{U}(x,M-e)$.  Since $\mu_c$ solves (\ref{CV}), it will also satisfy
$$\max\nolimits_{\{x\in \Re_+^L:\; p^{t_2}\cdot x\leq M-\mu_c\}} \widetilde{U}(x,(M-\mu_c)-p^{t_2}\cdot x)\; = \; \widetilde{U}(x^{t_1},M-p^{t_1}\cdot x^{t_1}).$$
In other words, $\mu_c$ is the reduction in total wealth that will leave the consumer's overall utility at $p^{t_2}$ the same as it was at $p^{t_1}$.  Thus, with this particular interpretation of the augmented utility function, $\mu_c$ coincides with what is called the {\em compensating variation} in standard consumer theory. For this reason, we shall also refer to $\mu_c$, defined by (\ref{CV}), as the compensating variation.

Pushing the analogy further, it is possible to use the compensating variation in our model in the same way it is typically used. For example, the price change from $p^{t_1}$ to $p^{t_2}$ may benefit some consumers while hurting others. The Kaldor criterion would deem this change an overall improvement if the sum of the compensating variations across all consumers is positive since it guarantees that those who benefit from the price change could, in principle, compensate the losers and still be better off.

In a similar way, we can define the {\em equivalent variation} as the value $\mu_e$ that solves
\begin{equation}\label{EV}
\max\nolimits_{x\in \Re_+^L} U(x,-p^{t_1}\cdot x+\mu_e)=V(p^{t_2})
\end{equation}
If $U(x,-e)=\widetilde{U}(x,M-e)$ then $\mu_e$ also solves
$$\max\nolimits_{\{x\in \Re_+^L: \; p^{t_2}\cdot x\leq M+\mu_e\}} \widetilde{U}(x,(M+\mu_e)-p^{t_1}\cdot x)=\widetilde{U}(x^{t_2},M-p^{t_2}\cdot x^{t_2}).$$
In other words, $\mu_e$ coincides with the equivalent variation as it is usually defined.\vspace*{.1in}

Now suppose a data set $\mathcal{D}$ obeys GAPP and contains the observation $(p^{t_1},x^{t_1})$. What can we say about the compensating variation of a price change from $p^{t_1}$ to $p^{t_2}$ (where the latter may or may not be a price observed in $\mathcal{D}$)?  There will typically be a range of these values since there is more than one augmented utility function that rationalizes $\mathcal{D}$. Nonetheless, it is possible to obtain a \textit{tight} lower bound for the set of possible compensating variation values. Formally, this is given by
$$\inf\{\mu_c:\mbox{$\mu_c$ solves (\ref{CV}) for some augmented utility function $U$ that rationalizes $\mathcal{D}$}\}.$$
Abusing terminology somewhat, we shall denote this term simply  by $\inf({\mu}_{c})$.

We now describe how to compute this bound.\footnote{We leave the reader to carry out the analogous exercise for the equivalent variation.} Let $S\subset T$ be the set of observations such that $s\in S$ if $p^{s}\succeq^*_p p^{t_1}$.  This set is nonempty since it contains $p^{t_1}$ itself.  For each $s\in S$, there is $m_c^s$ such that
\begin{equation}\label{justnice}
p^{t_2}\cdot x^s+m_c^s=p^sx^s.
\end{equation}
We claim that for any $U$ that rationalizes $\mathcal{D}$, the compensating variation $\mu_c\geq m_c^s$. This is because if $m<m_c^s$, then $m\neq \mu_c$ for any utility function rationalizing $\mathcal{D}$.  Indeed,
\begin{multline*}
\max\nolimits_{x\in \Re_+^L} U(x,-p^{t_2}\cdot x-m)\geq U(x^s,-p^{t_2}\cdot x^s-m)> U(x^s,-p^{t_2}\cdot x^s-m_c^s)\\
=U(x^s,-p^s\cdot x^s)\geq U(x^{t_1},-p^{t_1}\cdot x^{t_1})=V(p^{t_1}).
\end{multline*}
Thus $\inf(\mu_c)\geq m_c^s$ for all $s\in S$.  In fact, it is possible to obtain a stronger conclusion:
\begin{equation}  \label{justnice2}
\inf(\mu_c)= \max\{m_c^s:\mbox{$m_c^s$ satisfies (\ref{justnice}) for some $s\in S$}\}.
\end{equation}
Since the right side of this equation can be easily computed from the data, we have found a practical way of calculating $\inf(\mu_c)$.

\medskip

Notice that if $p^{t_2}$ is revealed preferred to $p^{t_1}$ (equivalently, that there is $s'\in S$ such that $m_c^{s'}\geq 0$),\footnote{Recall that $p^{t_2}\succeq_p^* p^{t_1}$ makes sense even if $p^{t_2}$ is not observed in the data set; see footnote \ref{non-sample-observe}.}  then $\inf(\mu_c)\geq 0$; in other words, at $p=p^{t_2}$, a lump sum {\em tax} of $\inf(\mu_c)$ will leave the agent no worse off than at $t_1$ and potentially better off.  On the other hand, if $p^{t_2}$ is {\em not} revealed preferred to $p^{t_1}$, that is, for every $s\in S$, we have $m_c^s< 0$, then $\inf(\mu_c)< 0$; in other words, at $p=p^{t_2}$, a lump sum {\em transfer} of $\inf(\mu_c)$ to the agent will leave the agent no worse off than at $t_1$ and potentially better off.

We provide a fuller discussion on the compensating variation, including a proof of (\ref{justnice2}), in \refappendix{sec:morecomvar}.

\subsection{Measuring departures from rationality}\label{sec:index}

Empirical studies that apply \refAfriat{thm:Afriat} frequently find GARP violations. A common way of measuring the {\em extent} of such violations is to compute the {\em critical cost efficiency index} \cite{afriat1973}. This refers to the largest $e\in (0,1]$ such that $\mathcal{D}$ can be rationalized in the following sense: there is a locally-nonsatiated utility function $\widetilde{U}$ such that $\widetilde{U} (x^t)\geq \widetilde{U}(x)$ for all $x$ in the `shrunken' budget set $B^t_{e}=\left\{x\in\Re_+^L:\; p^t \cdot x\leq ep^t \cdot x^t\right\}$.  Rationality is imperfect if $e<1$ since the consumer behaves as though she ignores bundles $x'$ that satisfy $e p^{t}\cdot x^{ t}<p^{t}\cdot x'\leq p^{t}\cdot x^{t}$ and, there could be some observation $\hat t$ and bundle $x'$ in this range for which $\widetilde{U}(x')> \widetilde{U}(x^{\hat t})$. Importantly, the calculation of the critical cost efficiency index is straightforward and is facilitated by a modified version of GARP.

There is a similar way of measuring the extent to which a data set $\mathcal{D}$ fails to be rationalized by an augmented utility function. For a given $\vartheta\in (0,1]$, there is a weaker version of the GAPP test that allows us to determine whether there is an expenditure-augmented utility $U:\Re_+^{L}\times \Re_{-}\to \Re$ such that, at each observation $t$,
$$U(x^t,-p^t\cdot x^t)\geq U(x,-\vartheta^{-1}p^t\cdot x)\:\mbox{ for all $x\in\Re_+^L$.}$$
If there is, we say that $\mathcal{D}$ is $\vartheta$-rationalized by an augmented utility function. Notice that if $\mathcal{D}$ can be $\vartheta$-rationalized then it can be $\vartheta'$-rationalized for any $\vartheta'<\vartheta$, since $U$ is strictly decreasing in expenditure.   The consumer who is $\vartheta$-rational (for $\vartheta<1$) may have only limited or bounded rationality in the sense that there could be a bundle $x'$ and an observation $\hat t$ such that
$$U(x^{\hat t},-p^{\hat t}\cdot x')>U(x^{\hat t},-p^{\hat t}\cdot x^{\hat t})\geq U(x',-\vartheta^{-1}p^{\hat t}\cdot x').$$
In other words, the consumer fails to recognize that bundle $x'$ is superior to $x^{\hat t}$ at $t=\hat t$ because she has inflated (by $\vartheta^{-1}$) the expenditure of purchasing $x'$. Any data set can be $\vartheta$-rationalized for some $\vartheta\in (0,1]$ and the supremum $\vartheta^*$ over these values provides a natural measure of rationality which we shall refer to as the {\em rationality index}.

The following proposition establishes a connection of our rationality index with the critical cost efficiency index.

\begin{proposition}\label{prop:GAPP-GARP}
Let ${\mathcal D}=\{(p^t,x^t)\}_{t=1}^T$ be a data set and let $\breve{\mathcal D}=\{(p^t,\breve{x}^t)\}_{t\in T}$, where $\breve{x}^t=x^t/(p^t\cdot x^t)$, be its iso-expenditure version. Then $\vartheta^*$ is the rationality index for $\mathcal{D}$ if and only if it is the critical cost efficiency index for $\breve{\mathcal D}$.
\end{proposition}

A consequence of this result is that the rationality index inherits the ease of computation of the critical cost efficiency index. In \refappendix{sec:rat-indices}, we  provide instructions on this computation including in more general environments with nonlinear prices. The proof of  \refprop{prop:GAPP-GARP} can be found in \refappendix{sec:thetaGPGP}.

\subsection{Deflating prices}  \label{sec:deflprices}

When the data set $\mathcal{D}=\{(p^t,x^t)\}_{t=1}^T$ is collected over an extended period, it is possible that there are  changes in the prices of all goods, including goods outside the ones observed.  Thus the nominal value of expenditure may no longer be an accurate measure of the opportunity cost of expenditure.  A simple way of taking this into account is to deflate the prices of the $L$ goods with a general price index.   In other words, one could check if $\widetilde{\mathcal{D}}=\{(p^t/k^t,x^t)\}_{t=1}^T$ obeys GAPP, where $k^t\in \mathbb{R}_{++}$ is an index of the general price level. If it does, it would mean that there is an augmented utility function $U$ that rationalizes the data after deflation; in other words,
$$x^t\in\argmax_{x\in\Re^L_+}U\left(x,-\frac{p^t\cdot x}{k^t}\right)\:\mbox{ for all $t\in T$.}$$

This simple way of accounting for general price changes could be precisely justified when the augmented utility function is the reduced form of a larger constrained optimization problem.  Indeed, suppose that the consumer is maximizing an overall utility $\widetilde{U}(x,y)$ that depends both on the observed bundle $x$ and on a bundle $y$ of other goods, subject to a global budget of $M$. Formally, the consumer maximizes $\widetilde{U}(x,y)$ subject to $p\cdot x+q\cdot y\leq M$, where $q$ are the prices of goods $y$. Keeping $q$ and $M$ fixed, $U(x,-e)$ is defined as the greatest overall utility the consumer can achieve by choosing $y$ optimally, subject to expenditure $M-e$ and conditional on consuming $x$, that is,
\begin{equation}\label{interpret}
U(x,-e)=\max\{\widetilde{U}(x,y):\mbox{$y\geq 0$ and $q\cdot y\leq M-e$}\}.
\end{equation}
At the prices $p^t$ for the observed goods and $q$ for the outside goods, the consumer chooses a bundle $(x^t,y^t)$ to maximize $\widetilde{U}$ subject to $p^t\cdot x+q\cdot y\leq M$.  Then $\mathcal{D}=\{(p^t,x^t)\}_{t=1}^T$ will obey GAPP, since $x^t$ maximizes $U(x,-p^t\cdot x)$, with $U$ as defined by (\ref{interpret}).

Now suppose that the prices of the other goods are changing.  Consider the simplest case where these prices move up or down proportionately, so they are $k^t q$ at observation $t$, for some scalar $k^t>0$. Furthermore, assume that the agent's global budget at $t$ also increases by a factor $k^t$, which means that the consumer's nominal wealth is keeping pace with price inflation.  Then at observation $t$, the consumer maximizes $\widetilde{U}(x,y)$ subject to $(x,y)$ obeying
$$p^t\cdot x+k^t q\cdot y\leq k^t M.$$
Dividing this inequality by $k^t$, we see that the consumer's choice is identical to the case where the price of the observed goods is $p^t/k^t$, with external prices and total wealth constant at $q$ and $M$ respectively.  Therefore, the data set with deflated prices,  $\widetilde{\mathcal{D}}=\{(p^t/k^t,x^t)\}_{t=1}^T$ obeys GAPP.

In the case where the relative prices of the outside goods change, it is still possible to derive a price index which ensures that GAPP holds after deflating $p^t$, but this requires stronger assumptions on $\widetilde{U}$. We discuss this in detail in \refappendix{sec:indices}.

\section{General consumption spaces and nonlinear pricing}\label{sec:nonlinearGAPP}

So far we have assumed that the consumption space is $\Re^L_+$ and that prices are linear, but in fact neither feature is crucial to our main result. In this section, we assume that the space from which the consumer chooses her consumption is $X\subseteq \Re_+^L$.  We define a {\em price system} as a map $\psi: X\to\mathbb{R}_+$, where $\psi (x)$ is the cost of purchasing $x\in X$. Of course, a special case of a price system is $\psi (x)=p\cdot x$ but the more general formulation with $\psi$ allows for quantity discounts, bundle pricing and other pricing features that can be important in certain contexts (such as our empirical application in \refsec{sec:progresa}).

We assume that both the price system and the bundle chosen by the consumer are observed. Formally, a data set is a collection $\mathcal{D}=\{(\psi^t,x^t)\}_{t=1}^T$.  This data set is rationalized by an augmented utility function $U:X\times\mathbb{R}_{-}\to \mathbb{R}$ if
\begin{equation}\label{nonlinear-ratio}
x^t\in \argmax_{x\in X}U(x,-\psi^t(x))\:\mbox{ for all $t\in T$.}
\end{equation}

The notion of revealed preference over prices can be extended to a revealed preference over price systems.  We say that $\psi^{t'}$ is {\em directly revealed preferred} ({\em directly revealed strictly preferred}) to $\psi^t$ if $\psi^{t'} (x^{t})\leq (<)\psi^{t}(x^{t})$; we denote this by $\psi^{t'} \succeq_p (\succ_p) \psi^t$.  We denote the transitive closure of $\succeq_p$ by $\succeq_p^*$, that is, $\psi^{t'}\succeq_p^* \psi^t$ if there are $t_1$,  $t_2,\ldots,t_N$ in $T$ such that $\psi^{t'}\succeq_p \psi^{t_1}$, $\psi^{t_1}\succeq_p \psi^{t_2},\ldots,\psi^{t_{N-1}}\succeq_p \psi^{t_N}$, and $\psi^{t_N}\succeq_p \psi^{t}$; in this case we say that $\psi^{t'}$ is {\em revealed preferred} to $\psi^t$.  If anywhere along this sequence, it is possible to replace $\succeq_p$ with $\succ_p$ then we denote that relation by $\psi^{t'}\succ^*_p \psi^{t}$ and say that $\psi^{t'}$ is {\em strictly revealed preferred} to $\psi^t$. It is straightforward to check that if, $\mathcal{D}$ can be rationalized by an augmented utility function, then it obeys the following generalization of GAPP to price systems:
\begin{quote}
{\em there do not exist observations $t,t'\in T$ such that $\psi^{t'}\succeq^*_p \psi^{t}$ and $\psi^{t}\succ^*_p \psi^{t'}$.}
\end{quote}

\medskip

The following theorem asserts that the converse is also true and that, under further conditions, we can guarantee that the data can be rationalized by an augmented utility function with additional properties. The proof of this result (in fact of a more general result allowing for errors) is in \refappendix{sec:ultimate}.

\begin{theorem}\label{thm:GAPP_nonlinear}
A data set $\mathcal{D}=\{(\psi^t,x^t)\}^T_{t=1}$ can be rationalized by an augmented utility function if and only if satisfies GAPP.

Furthermore, suppose that $\mathcal D$ satisfies GAPP, $X$ is closed and that, for all $t\in T$, the price systems have the following properties:\, (i) $\psi^t$ is a continuous function; (ii) for any number $M$, $\{x\in X:\psi^t(x)\leq M\}$ is a compact set; and (iii) $\psi^t$ is strictly increasing in $x_K$ for some $K\subseteq L$.\footnote{This means that if $x'$, $x\in X$, $x'\neq x$, $x_{\ell}=x_{\ell}'$ for all ${\ell}\notin K$, and $x'_{\ell}\geq x_{\ell}$ for all ${\ell}\in K$, then $\psi^t(x')>\psi^t(x)$.}  Then, for any closed set $Y\subseteq\Re^L_+$ containing $X$, there is a continuous augmented utility function $U:Y\times\Re_{-}\to\Re$ that rationalizes $\mathcal D$, with $U(x,-e)$ strictly increasing in $x_K$.
\end{theorem}
\noindent Remarks:\, (1)\,  Note that condition (ii) is a weak assumption requiring that there be no arbitrarily large bundles with a bounded price.  (2) By definition, an augmented utility function is strictly decreasing in expenditure, but in certain cases it may be natural to require $U$ to be strictly increasing in $x_K$ for some set $K$ (which can be empty). The theorem says that this is possible, so long as the price systems are also strictly increasing in $x_K$. (3) Lastly, the theorem guarantees that the domain of the augmented utility function can be larger than $X$.  For reasons which will be clear later in this section, this is natural in certain applications. However when we say that $U$ rationalizes the data, we mean that (\ref{nonlinear-ratio}) holds and, in particular, $x^t$ need not be optimal in $Y$.\medskip

The literature on mental accounting has emphasized the possibility of actors in the economy manipulating the mental budgets of agents.  The following example shows how a nonlinear GAPP test can be used to detect such phenomena.

\begin{example}\label{ex:mentalbudget}
A store initially prices two goods at $p^t=(1,2)$ and a shopper purchases $x^t=(10,20)$ from the store.  The store introduces a scheme where regular customers (such as this shopper) receives a voucher of 12 dollars to be used for the purchase of the store's products; prices are changed to $p^{t'}=(2,3/2)$ and the shopper buys $x^{t'}=(20,20)$.\footnote{If good 1 is cheap to procure, this scheme is advantageous to the store, since in the first instance, the shopper spends 50 dollars while in the second, she spends 58 (net of the voucher).}  What is the impact of the gift voucher?

Since the value of the voucher is small in terms of total income, the shopper could spread this reward widely across all purchases (including purchases from other stores) and this should result in no (or at least a very small) impact on demand for the store's products. On the other hand, she may have a mental budget for purchases at that store, and the voucher represents an appreciable increase in {\em that} mental budget by 12 dollars.

A revealed preference analysis supports the latter hypothesis.  Indeed, if we ignore the voucher, the data are not compatible with the maximization of an augmented utility function since $p^t\cdot x^t=50=p^{t'}\cdot x^t$ and $p^{t'}\cdot x^{t'}=70>60=p^t\cdot x^{t'}$, which violates GAPP. On the other hand, at observation $t'$, we could model the shopper as mentally discounting 12 dollars from her expenditure at the shop.  In formal terms, the price system at $t'$ is a function $\psi^{t'}(x)=\max\{p^{t'}\cdot x-12,0\}$, so $\psi^{t'}(x^{t'})=58$. In this case, we have $\psi^{t'}\succ^*_p\psi^t$ (where $\psi^t(x)=p^tx$), but it is no longer the case that $\psi^t\succeq^*_p\psi^{t'}$ since $\psi^t(x^{t'})=60>\psi^{t'}(x^{t'})=58$.  So the data is now compatible with GAPP, but with a nonlinear pricing system based on the shopper's mental accounting.\footnote{Notice (in connection with our discussion of mental accounting in \refsec{sec:behavioral}) that the total mental budget of the shopper remains unknown, though the researcher observes an event that has altered that budget.}
\end{example}

\subsection{Discrete consumption spaces}\label{sec:discretespace}

Below are four instances where \refthm{thm:GAPP_nonlinear} could be applied.

\medskip

\noindent (1)\, Suppose that the consumption space consists of $L$ goods of which the first $K$ can only be consumed in discrete quantities (as in the model of \citet{polisson2013}, for example). The consumption space is then $X=\mathbb{N}^{K}\times \Re^{L-K}_+$, where $\mathbb{N}$ is the set of natural numbers. \refthm{thm:GAPP_nonlinear} is applicable whether or not prices are linear. Suppose that each good has a price $p_i>0$. Since $X$ is closed and the price system $\psi^t(x)=p^t\cdot x$ is strictly increasing in $x$, \refthm{thm:GAPP_nonlinear} guarantees that, if GAPP holds, then there is a continuous augmented utility function that is strictly increasing in $x\in X$ and rationalizes $\mathcal D$.

\medskip

\noindent (2)\, Another natural choice environment is one where the consumer is deciding on buying a subset of objects from a set with $L$ items. Then each subset could be represented as an element of $X=\{0,1\}^L$. For $x\in X$, the $\ell^{\text{th}}$ entry $x_{\ell}$ equals 1 if and only if the $\ell^{\text{th}}$ object is chosen. If only certain subsets are permissible, as in the case of discrete choice, then $X$ would be a strict subset of $\{0,1\}^L$. The price system $\psi$ gives the price of different bundles of goods. Let $e_{\ell}$ denote the vector with 1 in the $\ell^{\text{th}}$ entry and zero everywhere else. Then $\psi(e_{\ell})$ is the price of purchasing good ${\ell}$ alone. The price system is nonlinear if $\psi (x)\neq \sum_{\ell=1}^L x_{\ell}\psi(e_{\ell})$ for some $x\in X$.

\medskip

\noindent (3)\, In empirical models of demand for differentiated goods, it is common to model each good as embodying a set of characteristics (see \citet{nevo2000}). For example, if each good is a type of breakfast cereal, then the characteristics could be the calories, fiber content etc. Suppose that there are $L$ characteristics and let $Y_{\ell}\subseteq\Re_+$ be the set of values that characteristic $\ell$ can take. Then, the characteristics space is $Y=\times_{\ell=1}^L Y_{\ell}$.\footnote{If characteristic $1$ naturally takes on continuous values (such as calories) then we let $Y_{\ell}=\Re_+$. Characteristic 2 could be the brand. Suppose there are two brands, then $Y_2=\{1,2\}$, and so on.} There are $I$ goods, with good $i$ having characteristics $x^i\in Y$.  Assuming (as is common in these models) that a consumer purchases only one good, the consumption space is $X=\{x^i\}_{i=1}^I$ and a price system $\psi:X\to\Re_{++}$ is just a list of prices for the different goods.

In this context, it is natural to model the consumer with an augmented utility function defined on characteristics and expenditures $Y\times \Re_{-}$, even though the products available to her are only those in $X$. Furthermore, among the characteristics, there could be those where higher values are unambiguously better, in which case the researcher could be interested in guaranteeing that utility is strictly increasing in those characteristics.
\refthm{thm:GAPP_nonlinear} allows for these considerations.  If $\mathcal D$ obeys GAPP then it can be rationalized by a continuous augmented utility function $U:Y\times\Re_{-}\to\Re$. Additionally, for a set of characteristics $K\subseteq L$, one could guarantee that $U(y,-e)$ is strictly increasing in $y_K$ so long as $\psi^t(x)$ is strictly increasing in $x_K$, for all $t$.

In models of differentiated goods, it is also common to allow for the introduction of new goods and for changes to a product's characteristics.\footnote{These changes could be substantive (for example, a change to a breakfast cereal formula) or it could be a change in advertising expenditure that serves as a proxy for a change in a product's public profile.}  Obviously, changes to a product's characteristics could potentially lead to a change in the product's utility which, unless taken into account by the test, could lead to a spurious rejection of augmented utility-maximization.  In formal terms, these changes can be captured by allowing the set of alternatives to depend on $t$; in \refsec{sec:differentiatedgoods}, we explain how it is possible to modify the GAPP test in \refthm{thm:GAPP_nonlinear} to account for changes of this type.

\subsection{Characteristics models with continuous consumption spaces}

We assume that the space of characteristics is $Y=\Re^L_+$, with each product $i$ represented by a vector of characteristics $x^i\in Y$.  We allow these goods to be bought in bundles, so the consumption space is the convex cone $X$ generated by $\{x^i\}_{i=1}^I$.\footnote{For a GARP-based test of a model of this type, see \citet{blow2008}.}  We assume that the vectors $\{x^i\}_{i=1}^I$ are linearly independent; this guarantees that for each $\hat x\in X$, there is a {\em unique} bundle of goods, $\hat{\alpha}=(\hat{\alpha}_i)_{i=1}^I\in\Re^I_+$ such that $\sum_{i=1}^I \hat{\alpha}_i x^i=\hat x$.  We denote $\hat{\alpha}$ by $\alpha(\hat x)$. Let $p^t\in\Re^I_{++}$ be the prices of the $I$ goods at observation $t$. To obtain the bundle $x\in X$, the consumer needs to spend $\psi^t(x)=p^t \cdot \alpha(x)$.

At observation $t$, the researcher observes $p^t$ and the consumer's purchases $\alpha^t\in\Re^I_{+}$. We assume the researcher knows $\{x^i\}_{i=1}^I$ and so he can work out the consumption in characteristics space, $x^t=\sum_{i=1}^I \alpha^{t}_{i}x^i$, as well as the price system $\psi^t$.  \refthm{thm:GAPP_nonlinear} guarantees that if $\mathcal{D}=\{(\psi^t,x^t)\}_{t\in T}$ satisfies GAPP then it can be rationalized by a continuous augmented utility function $U:\Re^L_+\times\Re_{-}\to\Re$.  So long as $\psi^t(x)$ is strictly increasing in $x\in X$ for each $t$, we can also ensure that $U(y,-e)$ is increasing in the characteristics $y$.

\section{The Random Augmented Utility Model}\label{sec:RAUMX}

In this section, we develop the random version of the expenditure-augmented utility model. We first describe our test procedure for this model via a simple example.

\subsection{An Illustrative Example}\label{sec:Example_Random}

\begin{example}\label{eg:StTesting}

Suppose we have repeated cross-sectional data consisting of the demand of a population of ten consumers at two price vectors.  This is illustrated in \reffig{fig:Data} where the collection of points in the left and right panels indicate the demand bundles at $p^{t}=(2,1)$ and $p^{t'}=(1,2)$ respectively. The lines in \reffig{fig:Data} indicate relative prices.

\begin{figure}[h]
	\centering
	\begin{subfigure}[b]{.5\linewidth}
		\includegraphics[trim=4.6cm 9.8cm 8.2cm 10.5cm,clip=true,scale=.8,page=13]{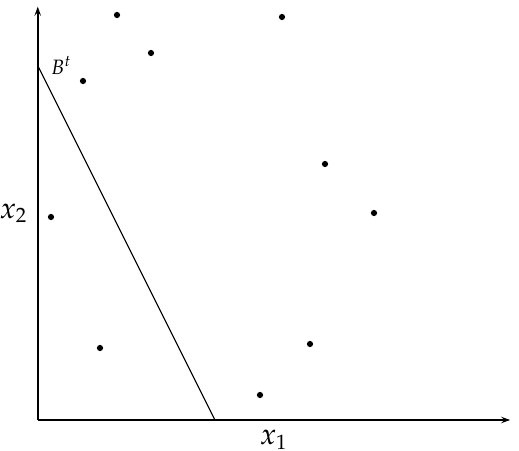}
		\caption{Choices under prices $p^{t}$}
		\label{fig:obs1}
	\end{subfigure}%
	\begin{subfigure}[b]{.5\linewidth}
		\includegraphics[trim=4.6cm 9.8cm 8.2cm 10.5cm,clip=true,scale=.8,page=14]{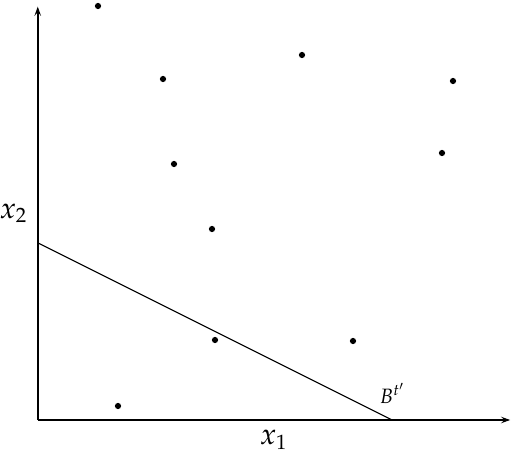}
		\caption{Choices under prices $p^{t'}$}
		\label{fig:obs2}
	\end{subfigure}
	\caption{The Data Set}
	\label{fig:Data}
\end{figure}

Since we assume this is a cross-sectional data set, the econometrician cannot match consumption bundles across the two panels by consumer identity. The question we wish to address is whether this data set can be {\em rationalized}, by which we mean the following.

\begin{quote} Can we match the choices at $t$ with those at $t'$, forming ten distinct pairs, such that each pair can be rationalized by an augmented utility function (or, equivalently, satisfies GAPP)? \end{quote}
This exercise illustrates the problem we address in this section:  given the empirical demand distributions in different periods, is there a time-invariant distribution over consumer types that could explain these empirical distributions, subject to the restriction that each type is an augmented utility function maximizer.

An obvious way of answering the above question would be to consider all possible partitions into pairs and test GAPP on each pair.\footnote{Note that when we formally define our test in the next subsection, the choice distribution will be assumed to be atomless. The simple finite matching analogy in this section, while inexact, is meant to provide the intuition our methodology.} This approach, however, is not practical when the population at each observation is large and when there are more than two periods. Fortunately, there is a different procedure that works in general, which we now explain.

\refprop{prop:GAPP-GARP} tells us that a pair $\mathcal{D}=\{(p^t,x),(p^{t'},x')\}$ created by choosing bundle $x$ from observation $t$ and $x'$ from observation $t'$ obeys GAPP if and only if its {\em iso-expenditure} analog, $\breve{\mathcal{D}}=\{(p^t,\breve{x}), (p^{t'},\breve{x}')\}$ obeys GARP, where $\breve{x}=x/p\cdot x$ and $\breve{x}'=x'/p\cdot x'$ are scaled versions of $x$ and $x'$ that satisfy $p^t\cdot \breve{x}= p^{t'}\cdot \breve{x}=1$. This scaling is demonstrated in Figures \ref{fig:Btchoices} and \ref{fig:Bt'choices}. \reffig{fig:Patches} shows the scaled bundles from both observations superimposed onto a single picture. This figure also labels different partitions of the budget lines and we use this notation in what follows.

\begin{figure}[h]
	\centering
	\begin{subfigure}[b]{.5\linewidth}
		\includegraphics[trim=4.6cm 11.5cm 8.2cm 8.8cm,clip=true,scale=.9,page=5]{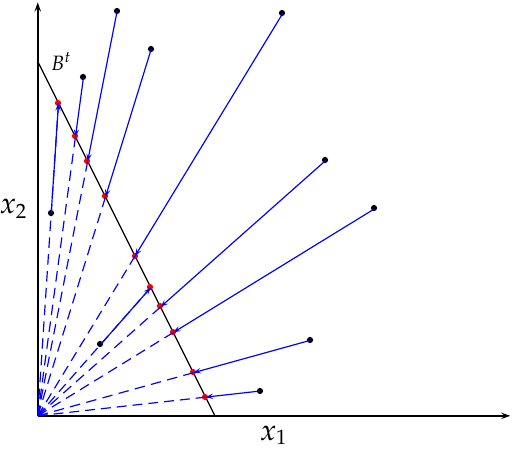}
		\caption{Period $t$}
		\label{fig:Btchoices}
	\end{subfigure}%
	\begin{subfigure}[b]{.5\linewidth}
		\includegraphics[trim=4.6cm 11.5cm 8.2cm 8.8cm,clip=true,scale=.9,page=6]{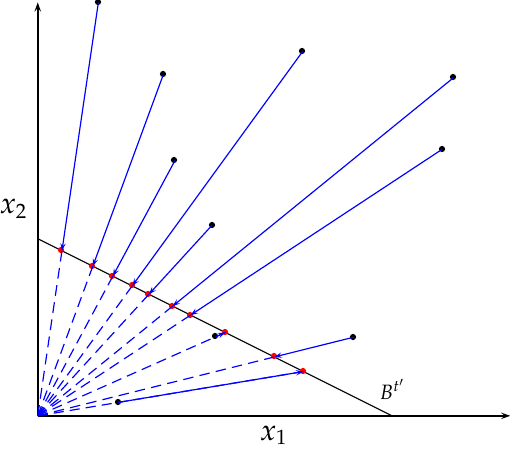}
		\caption{Period $t'$}
		\label{fig:Bt'choices}
	\end{subfigure}
	\begin{subfigure}[b]{.5\linewidth}
		\includegraphics[trim=4.6cm 9.8cm 6cm 10.6cm,clip=true,scale=.9,page=15]{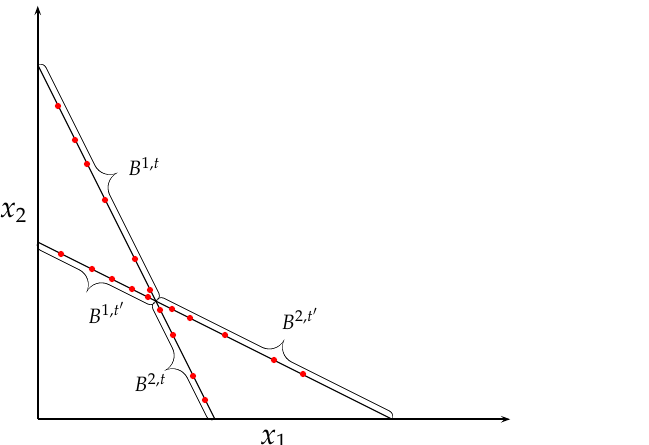}
		\caption{Patches generated by budget intersections}
		\label{fig:Patches}
	\end{subfigure}
	\caption{Observed and Rescaled Choices}
	\label{fig:ChoiceSegments}
\end{figure}

Now recall that, if $\breve{\mathcal{D}}$ satisfies GARP, then it is {\em not} possible for $\breve{x}$ to lie on $B^{2,t}$ and for $\breve{x}'$ to lie on $B^{1,t'}$ (see \reffig{fig:ChoiceSegments}(c)). Instead $\breve{\mathcal{D}}$ must belong to one of the following three types: $\breve{x}$ lies on $B^{1,t}$ and $\breve{x}'$ lies on $B^{2,t'}$;  $\breve{x}$ lies on $B^{1,t}$ and $\breve{x}'$ lies on $B^{1,t'}$; or $\breve{x}$ lies on $B^{2,t}$ and $\breve{x}'$ lies on $B^{2,t'}$.  These cases are graphically depicted in \reffig{fig:nu1}, \reffig{fig:nu2} and \reffig{fig:nu3} respectively.

\begin{figure}[h]
	\centering
	\begin{subfigure}[b]{.5\linewidth}
		\includegraphics[trim=4.6cm 11.5cm 8.2cm 8.8cm,clip=true,scale=.8,page=7]{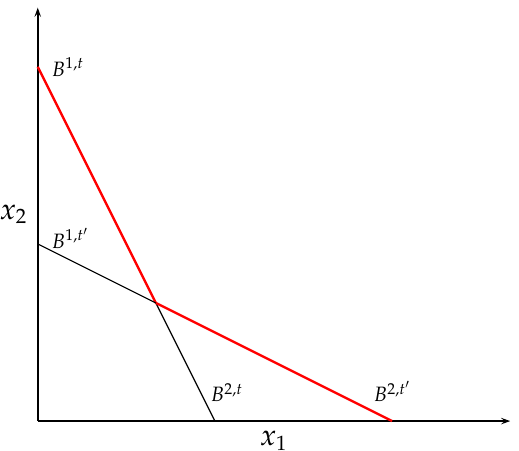}
		\caption{Proportion of this Rational Type: $\nu_1$}
		\label{fig:nu1}
	\end{subfigure}%
	\begin{subfigure}[b]{.5\linewidth}
		\includegraphics[trim=4.6cm 11.5cm 8.2cm 8.8cm,clip=true,scale=.8,page=8]{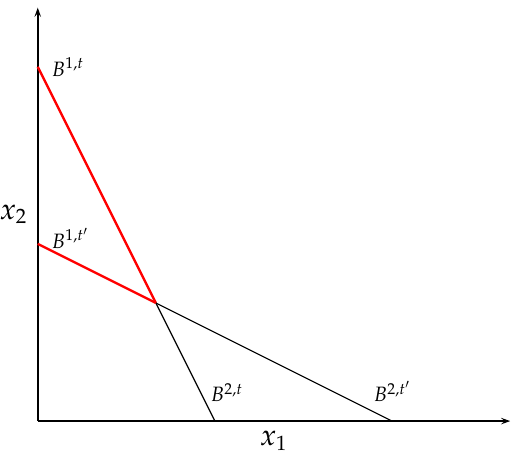}
		\caption{Proportion of this Rational Type: $\nu_2$}
		\label{fig:nu2}
	\end{subfigure}\newline
	\begin{subfigure}[b]{.5\linewidth}
		\includegraphics[trim=4.6cm 11.5cm 8.2cm 8.8cm,clip=true,scale=.8,page=9]{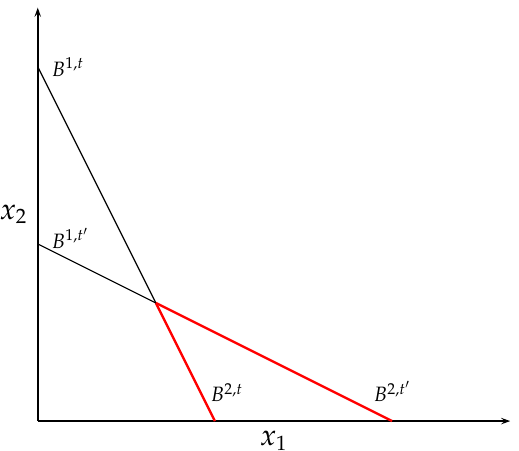}
		\caption{Proportion of this Rational Type: $\nu_3$}
		\label{fig:nu3}
	\end{subfigure}
	\caption{Set of Rational Types}
	\label{fig:rational_types}
\end{figure}

Denoting the fraction of each of these GAPP-consistent consumer types in the population by $\nu_1$, $\nu_2$ and $\nu_3$ respectively, together they must generate the observed proportion of choices on the segments $B^{1,t},\; B^{2,t},\;B^{1,t'}$ and $B^{2,t'}$.  \reffig{fig:choice_dis} demonstrates the proportion of choices in terms of the $\nu$s, while \reffig{fig:DeterminingPi} displays the empirical proportion of choices on each segment (after scaling), which we denote by
\begin{equation}\label{eq:labor}
\hat{\pi}=\left(\hat{\pi}^{1,t},\hat{\pi}^{1,t'},\hat{\pi}^{2,t},\hat{\pi}^{2,t'}\right)'=\left(\frac{3}{5},\frac{2}{5},\frac{1}{2},\frac{1}{2}\right)'.
\end{equation}
(For instance, $\hat{\pi}^{1,t}=\frac35$ because six of the ten rescaled demand bundles lie on $B^{1,t}$.)  Therefore, a necessary condition for rationalizing the data is that there are $\nu$s that solve
\begin{equation}\label{eq:nus}
\nu_1+\nu_2=\hat{\pi}^{1,t},\; \nu_1+\nu_3=\hat{\pi}^{2,t'},\; \nu_2=\hat{\pi}^{1,t'}, \; \nu_3=\hat{\pi}^{2,t}.
\end{equation}

Two observations should be immediate from this process. The first is that there could be data where the system (\ref{eq:nus}) has no solution; this occurs when $\hat{\pi}^{1,t}-\hat{\pi}^{1,t'}\neq \hat{\pi}^{2,t'}-\hat{\pi}^{2,t}$.  The second is that when the values of $\hat{\pi}$ are given by (\ref{eq:labor}), the solution to (\ref{eq:nus}) is
\begin{equation}\label{eq:nus2}
\nu=(\nu_1,\nu_2,\nu_3)'=\left(\frac{1}{10},\frac{1}{2},\frac{2}{5}\right)'.
\end{equation}

To confirm that the data in this example can indeed be rationalized, it remains for us to pair up the demand bundles at the two observations. To do this, arbitrarily choose $1 (=\nu_1\times 10)$ pair of choices that lie on $B^{1,t}$ and $B^{2,t'}$; any 5 ($=\nu_2\times 10$) pairs that lie on $B^{1,t}$ and $B^{1,t'}$; and the remaining 4 ($=\nu_3\times 10$) pairs on $B^{2,t}$ and $B^{2,t'}$. Clearly each pair satisfies GARP and thus the original (un-scaled) pair satisfies GAPP.

\begin{figure}[h]
	\centering
	\begin{subfigure}[b]{.49\linewidth}
		\includegraphics[trim=4.6cm 11.5cm 8.2cm 8.8cm,clip=true,scale=.8,page=10]{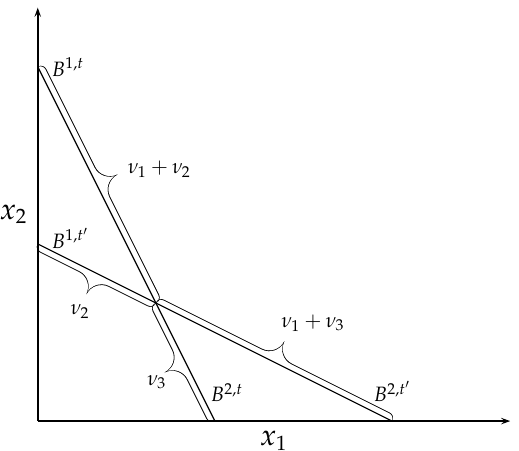}
		\caption{Resulting Distribution of Choices}
		\label{fig:choice_dis}
	\end{subfigure}
	\begin{subfigure}[b]{.49\linewidth}
		\includegraphics[trim=4.6cm 11.5cm 6cm 8.8cm,clip=true,scale=.9,page=12]{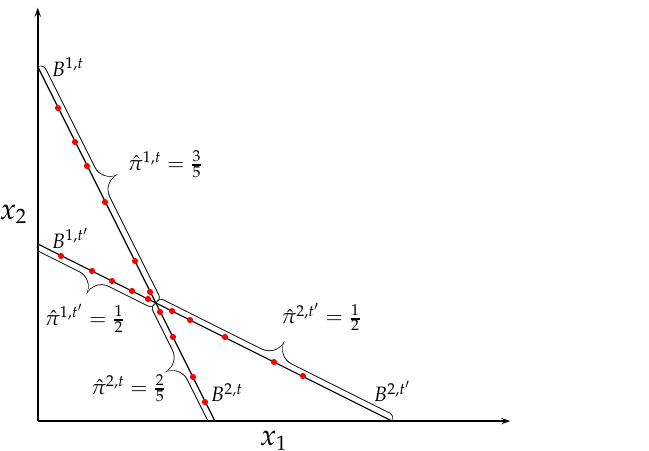}
		\caption{The Scaled Empirical Choice Probabilities $\hat{\pi}$}	
		\label{fig:DeterminingPi}
	\end{subfigure}
	\caption{Choice Distribution and Empirical Frequency}
	\label{fig:Rationalization}
\end{figure}

\end{example}

\subsection{Rationalization by Random Augmented Utility}\label{sec:RatRAUM}

The starting point of our analysis is a {\em repeated cross-sectional data set}, $\mathscr{D}:=\{(p^t,\mathring{\pi}^t)\}_{t=1}^T$, where each observation consists of the prevailing price $p^t$ and the distribution of demand in the population at that price, represented by a probability measure $\mathring{\pi}^t$ on $\Re^L_+$.  An example of $\mathscr{D}$ is the data set depicted in \reffig{fig:Data} where the probability measure corresponds to the empirical distribution of demand bundles.  The following definition generalizes the notion of rationalization considered in that example.

\begin{definition}
The repeated cross-sectional data set $\mathscr{D}=\{(p^t,\mathring{\pi}^t)\}_{t=1}^T$ is {\em rationalized} by the {\em random augmented utility model} (RAUM) if there exists a probability space $(\Omega, \mathcal{F},\mu)$ and a random variable $\chi:\Omega\to (\mathbb{R}^{L}_{+})^T$ such that, almost surely, $\{(p^t,\chi^t(\omega))\}_{t\in T}$ can be  rationalized by an augmented utility function (equivalently, obeys GAPP) and
\begin{equation}\label{eqn:condi}
\mathring{\pi}^t(Y)=\mu(\{\omega\in\Omega: \chi^t(\omega)\in Y\})\:\mbox{ for any measurable $Y\subseteq \Re^L_+$.}
\end{equation}
\end{definition}
In this definition, one could interpret $\Omega$ as the population of consumers and $\chi^t(\omega)$ as the demand of consumer type $\omega$ at observation $t$ (when the prevailing price is $p^t$); all consumer types in the support of $\mu$ are required to be consistent with the augmented utility model and the observed distribution of demand at each observation $t$, given by $\mathring{\pi}^t$,  must coincide with that induced by the distribution $\mu$ over consumer types. Alternatively, the model can also be interpreted as one where each individual's augmented utility changes over time but in such a way that the population distribution is stationary.

In \refeg{eg:StTesting}, the repeated cross-sectional data set has two observations, where the probability distributions are simply uniform distributions on finite support. A RAUM-rationalization involves matching observations in $t$ with those in $t'$, so that each pair obeys GAPP. In the general case with $T$ observations, the function $\chi$ solves a $T$-fold matching problem, where each group $\{\chi^t(\omega)\}_{t\in T}$ (along with the associated prices) satisfies GAPP and agrees with the observations (that is, \eqref{eqn:condi} is satisfied).\footnote{It is straightforward to check that, with two observations, finding a rationalization is equivalent to finding a zero-cost solution to the transportation problem (see \citet{galichon2011}) where the cost of a pair of bundles is 0 if it obeys GAPP and 1 otherwise.}

\medskip

We shall now explain the general procedure for deciding if a given repeated cross-sectional data set $\mathscr{D}$ can be RAUM-rationalized.  This procedure mimics our solution to \refeg{eg:StTesting}.  For ease of exposition, we impose the following assumption on the data.

\begin{assumption}  \label{ass:indif}
Let $B^t :=\{x\in \Re^L_+ \,:\,p^t\cdot x=1\}$ be the budget plane at prices $p^t$ and expenditure 1.  For all $t, t'\in T$ with $B^t\neq B^{t'}$,
$$\mathring{\pi}^t\left(\left\{x\in\Re^L_+:\frac{x}{p^t\cdot x}\in B^t\:\mbox{ and }\frac{x}{p^{t'}\cdot x}\in B^{t'}\right\}\right)=0.$$
\end{assumption}

\medskip

This assumption states that the probability of a bundle lying, after re-scaling, at the intersection of two budget planes is zero.  The assumption is not required for any of our results but it is convenient because it simplifies the exposition.\footnote{If we allow for mass at budget intersections, then we would have to include them in our definition of patches. This is notationally cumbersome but once included our arguments (and  \refthm{thm:StGAPPTest}) remain correct.} It is always satisfied if $\mathring{\pi}^t$ is absolutely continuous with respect to the Lebesgue measure and is unlikely to be violated in any application with a continuous consumption space and linear prices.

\medskip

Let $\{B^{1,t},\dots, B^{I_t,t}\}$ denote the collection of subsets, or \textit{patches}, of $B^t$ where each subset has as its boundaries the intersection of $B^t$ with other budget sets and/or the boundary planes of the positive orthant. These are the higher-dimensional and multi-period analogs to the line segments in \reffig{fig:Patches}.  Formally, for all $t\in T$ and $i_t\neq i'_t$, each set in $\{B^{1,t},\dots, B^{I_t,t}\}$ is closed and convex and satisfies the following conditions:
\begin{enumerate}
	\item[(i)] $\cup_{1\leq i_t \leq I_t} B^{i_t,t}=B^t$,
	\item[(ii)] $\text{int}(B^{i_t,t}) \cap B^{t'} =\phi$ for all $t'\neq t$ that satisfy $B^t\neq B^{t'}$ (where $\text{int}(B^{i_t,t})$ denotes the relative interior of $B^{i_t,t}$),
	\item[(iii)] $B^{i_t,t}\cap B^{i'_t,t}\neq\phi$ implies that $B^{i_t,t}\cap B^{i'_t,t}\subset B^{t'}$ for some $t'\neq t$ that satisfies $B^t\neq B^{t'}$.
\end{enumerate}\vspace{-0.05in}

\medskip

For the patch $B^{i_{t},t}$, we let
\begin{equation}\label{eqn:pie}
\pi^{i_t,t}:=\mathring{\pi}^t\left(\left\{x\in\Re^L_+: \frac{x}{p^t\cdot x}\in B^{i_t,t}\right\}\right).
\end{equation}
In words, $\pi^{i_t,t}$ is the probability that a bundle lies in $B^{i_t,t}$ after re-scaling.  Note that, even though there may be $i_t$, $i'_t$ for which $B^{i_t,t}\cap B^{i'_t,t}$ is nonempty, \refass{ass:indif} guarantees that $\sum_{i_t=1}^{I_t}\pi^{i_t,t}=1$.  We denote by $\pi^t$ the vector $(\pi^{i_t,t})_{i_t=1}^{I_t}$ and by $\pi$ the column vector $(\pi^1,\pi^2,\ldots,\pi^T)'$.  We refer to $\pi$ as the {\em vector of observed patch probabilities}.

\medskip

Consider a single-agent data set of the form $\mathcal{D}=\{(p^t,x^t)\}_{t=1}^T$.  Given $\mathcal{D}$, we can define its {\em iso-expenditure version}, which is $\breve{\mathcal{D}}=\{(p^t,\breve{x}^t)\}_{t=1}^T$, where $\breve{x}^t=x^t/p^t\cdot x^t$ (so $p^t\cdot \breve{x}^t=1$ for all $t$).  Suppose that $\breve{x}^t$ does not lie on the intersection of budget planes, that is, there is $i^t$ such that $\breve{x}^t\in \text{int}(B^{i^t,t})$.  We make two important observations.  First, \refprop{prop:GAPPscaling} tells us that $\mathcal{D}$ satisfies GAPP if and only if $\breve{\mathcal{D}}$ satisfies GARP.  Second, if $\mathcal{D}$ satisfies GAPP then so does $\mathcal{D}'=\{(p^t,y^t)\}_{t\in T}$ if $y^t$ has the property that its re-scaled version $\breve{y}^t$ satisfies  $\tilde y^t\in\text{int}(B^{i^t,t})$; this is because the revealed preference relations (over the bundles $\tilde y^t$) are determined only by where $\breve{y}^t$ lies on the budget set relative to its intersection with another budget.

It follows from these observations that we may classify all single-agent data sets that obey GAPP according to the patch occupied by the scaled bundle $\breve{x}^t$ at each $B^t$.  In formal terms, each $\mathcal{D}$ that obeys GAPP is associated with an iso-expenditure $\breve{\mathcal{D}}$ that obeys GARP, which is in turn associated with a vector $a = \left(a^{1,1},\dots, a^{I_T,T}\right)$ where
\begin{equation}\label{eqn:A_def}
a^{i_t,t}=\left\{\begin{array}{cl}1 & \text{ if } \breve{x}^t \in B^{i_t,t}, \\ 0 & \text{ otherwise.}\end{array}\right.
\end{equation}
Thus, for the observed prices, we have partitioned the collection of all deterministic data sets obeying GAPP (of which there could be infinitely many) into a {\em finite} number of distinct classes or types, based on its associated vector $a$. We denote this set of vectors by $\mathcal{A}$.
We use $A$ to denote the matrix whose columns consist of every $a\in\mathcal{A}$, arranged in an arbitrary order; we refer to $A$ as the {\em matrix of  GARP-consistent types}.

\medskip

In Example \ref{eg:StTesting}, all the deterministic data sets that obey GAPP must correspond to one of three types (as depicted in \reffig{fig:rational_types}) and
\begin{equation}\label{eq:matrixA}
A=\left(\begin{array}{c c c} 1 & 1 & 0 \\ 0 & 0 & 1 \\ 0 & 1 & 0 \\ 1 & 0 & 1\end{array}\right).
\end{equation}
(Each column in $A$ describes the types in \reffig{fig:rational_types}: from left to right the columns capture the types in Figures \ref{fig:nu1}, \ref{fig:nu2} and \ref{fig:nu3} respectively.)

Given a repeated cross-sectional data set $\mathscr{D}$, we can construct $\mathcal{A}$ and the matrix of GARP-consistent types $A$. Suppose that this data set can be rationalized by some distribution $\mu$.  Let $\nu_a$ denote the mass of consumers of type $a$ in the population, that is
$$\nu_a=\mu\left(\left\{\omega\in\Omega: \frac{\chi^t(\omega)}{p^t\cdot \chi^t(\omega)}\in B^{i_t,t}\:\mbox{ if $a^{i_t,t}=1$, for all $t\in T$}\right\}\right).$$
At a given observation $t$, let $\mathcal{A}^{i_{t},t}=\{a: a^{i_{t},t}=1\}$; this is the subset of GARP-consistent types that have their re-scaled demands in the patch $B^{i_{t},t}$ at observation $t$. The proportion of the population whose types belong to $\mathcal{A}^{i_{t},t}$ is
$$\mu\left(\left\{\omega\in\Omega: \frac{\chi^t(\omega)}{p^t\cdot \chi^t(\omega)}\in B^{i_{t},t}\right\}\right)=\sum_{a\in \mathcal{A}^{i_t,t}}\nu_a=\sum_{a\in \mathcal{A}}\nu_a \, a^{i_t,t}.$$
Since $\mathscr{D}$ is rationalized by $\mu$, setting $Y=\{x\in\Re^L_+: x/(p^t\cdot x)\in B^{i_{t},t}\}$ in (\ref{eqn:condi}), we obtain
\begin{equation}\label{eqn:primacondi}
\pi^{i_{t},t}=\sum_{a\in \mathcal{A}}\nu_a \, a^{i_t,t}
\end{equation}
where $\pi^{i_{t},t}$ is defined by \eqref{eqn:pie}.  In other words, the observed probability of choices that land on $B^{i_t,t}$ after scaling must equal to the mass of GARP-consistent types implied by $\mu$.  This condition must hold for all patches $B^{i_t,t}$, so \eqref{eqn:primacondi} can be more succinctly written as $A\nu=\pi$, where $\nu$ is the column vector $(\nu_a)_{a\in\mathcal{A}}$. (Recall that $\pi$ is the vector of observed patch probabilities.)
In \refeg{eg:StTesting}, $A$ is given by \eqref{eq:matrixA}, $\pi$ is given by \eqref{eq:labor} and the solution $\nu$ by \eqref{eq:nus2}.

\medskip

To recap, given a data set $\mathscr{D}$, we calculate the matrix of GARP-consistent types $A$ and the vector of patch probabilities $\pi$. A necessary condition for $\mathscr{D}$ to be rationalized by RAUM is that there is $\nu\in\Delta^{|\mathcal{A}| -1}$ that solves $A\nu=\pi$.  It turns out that this condition is also sufficient: if $\nu$ exists, then we can find a RAUM-rationalization of $\mathscr{D}$ where the proportion of the population with type $a$ is precisely $\nu_a$.  The details of this final step are in the \refappendix{sec:appendix_RAUM}. The next result summarizes this discussion.

\begin{theorem}\label{thm:StGAPPTest}
Let $\mathscr{D}=\{(p^t,\mathring{\pi}^t)\}_{t=1}^T$ be a repeated cross-sectional data set obeying \refass{ass:indif}. Then $\mathscr D$ can be RAUM-rationalized if and only if there exists a $\nu\in \Delta^{|\mathcal{A}| -1}$ such that $A\nu=\pi$.
\end{theorem}

We end this section by contrasting the RAUM test with that of the {\em classic random utility model} (or RUM for short). The typical data environment for the latter is one where each observation consists of a distribution of choices on a given \textit{constraint set} (which varies across observations). In that environment, \citet{McFadden1991} and \citet{McFadden2005} observe that the problem of testing RUM can be discretized. \KS operationalize this insight in the case where constraint sets are linear budget sets. In that context, requiring choices from the same constraint set simply means that $\mathscr{D}=\{(p^t,\mathring{\pi}^t)\}_{t=1}^T$ is {\em iso-expenditure}, in the sense that if $x$ is in the support of $\mathring{\pi}^t$ then $p^t\cdot x=1$. \KS demonstrates that an iso-expenditure data set $\mathscr{D}$ can be RUM-rationalized if and only if $A\nu=\pi$ for some $\nu\in\Re^{|\mathcal{A}|}_+$ (where $A$ is the matrix of GARP-consistent types and $\pi$ is the vector of patch probabilities).

Notice that \refthm{thm:StGAPPTest} recovers the result of \KS as a corollary.  RAUM-rationalization guarantees the existence of a distribution over types that is consistent with the observations (that is, \eqref{eqn:condi} holds), with $\{(p^t,\chi^t(\omega))\}_{t\in T}$ satisfying GAPP almost surely. With the iso-expenditure condition, GAPP and GARP are equivalent properties, which means (by \refAfriat{thm:Afriat}) that there is is a strictly increasing utility function $\widetilde{U}_{\omega}:\Re^L_+\to\Re$ with $\widetilde{U}_{\omega}(\chi^t(\omega))\geq \widetilde{U}_{\omega}(x)$ for all $x\in B^t$; this is precisely what is needed for a  RUM-rationalization. Of course, it is also clear from our proof of \refthm{thm:StGAPPTest} that we are building on \citetalias{kitamura2018}, since our proof strategy involves (in effect) the following three steps: (i) converting $\mathscr{D}$ into an iso-expenditure data set $\breve{\mathscr{D}}$ (obtained from $\mathscr{D}$ simply by scaling demands); (ii) noticing that $\widetilde{\mathscr{D}}$ can be RAUM-rationalized if and only if $\breve{\mathscr{D}}$ can be RUM-rationalized; and (iii) then relying on the characterization of RUM-rationalization in \citetalias{kitamura2018}.

\medskip

Since a population of heterogenous consumers typically do not have identical expenditures, an actual data set would not typically be iso-expenditure. In order to test RUM, \KS found it necessary to estimate an iso-expenditure data set $\breve{\mathscr{D}}$ from the true data set $\mathscr{D}$, which in turn requires an additional econometric procedure with all its attendant assumptions.  In contrast, as we have established in \refthm{thm:StGAPPTest}, the RAUM has the important empirical feature that it can be {\em directly tested} on data sets that are not iso-expenditure.

\subsection{Welfare Comparisons}\label{sec:st_welfare}

Since the test for rationalizability involves finding a distribution $\nu$ over different types, it is possible to use this distribution for welfare analysis.  To be specific, suppose that a government is contemplating a change in sales tax that could lead to prices changing from its current value $p^{t'}$ to $\hat p$. Relevant to the government's re-election prospects is the proportion of consumers who will be better off as a result of this price change.\footnote{We would like to thank an anonymous referee for suggesting this motivation.}  Our methods allow us to obtain information on this proportion.

To be specific, suppose the analyst has access to a data set $\mathscr{D}$ that contains among its observations $(p^{t'},\mathring\pi^{t'})$, i.e.,  the prevailing prices and the demand distribution. To determine the welfare effect of a price change from $p^{t'}$ to $\hat p$, let $\mathbb{1}_{\hat p\succeq^*_p p^{t'}}$ denote the row vector with its length equal to the number of rational types ($|\mathcal{A}|$), such that the $j^{\text{th}}$ element is 1 if $\hat p\succeq^*_p p^{t'}$ for the rational type corresponding to column $j$ of $A$ and 0 otherwise.\footnote{Even though $\hat p$ is not among the observed prices, one could still define $\hat p\succeq^*_p p^{t'}$; see footnote \ref{non-sample-observe}.}  In words, $\mathbb{1}_{\hat p\succeq^*_p p^{t'}}$ enumerates the set of rational types for which $\hat p$ is revealed preferred to $p^{t'}$. For a rationalizable data set $\mathscr{D}$, \refthm{thm:StGAPPTest} guarantees that
\begin{equation}\label{eqn:wel_lb}
	\underline{\mathcal{N}}_{\hat p\succeq^*_p p^{t'}}:=\,\begin{array}{c} \min_{\nu} \; \mathbb{1}_{\hat p\succeq^*_p p^{t'}}\, \nu, \\ \text{ subject to } \; A\nu=\pi, \end{array}
\end{equation}
is the lower bound on the proportion of consumers who are revealed better off at prices $\hat p$ compared to $p^{t'}$, while the upper bound is
\begin{equation}\label{eqn:wel_ub}
	\overline{\mathcal{N}}_{\hat p\succeq^*_p p^{t'}}:=\,\begin{array}{c} \max_{\nu} \; \mathbb{1}_{\hat p\succeq^*_p p^{t'}}\, \nu, \\ \text{ subject to } \; A\nu=\pi. \end{array}
\end{equation}

Since \eqref{eqn:wel_lb} and \eqref{eqn:wel_ub} are both linear programming problems (which have solutions if, and only if, $\mathscr{D}$ is rationalizable), they are easy to implement and computationally efficient.  Suppose that the solutions are $\underline\nu$ and $\overline{\nu}$ respectively; then for any $\beta\in [0,1]$, $\beta\underline{\nu}+(1-\beta)\overline \nu$ is also a solution to $A\nu=\pi$ and, in this case, the proportion of consumers who are revealed better off at $\hat p$ compared to $p^{t'}$ is exactly $\beta\,\underline{\mathcal{N}}_{\hat p\succeq^*_p p^{t'}}+(1-\beta)\,\overline{\mathcal{N}}_{\hat p\succeq^*_p p^{t'}}$.  In other words, the proportion of consumers who are revealed better off can take any value in the interval $[\underline{\mathcal{N}}_{\hat p\succeq^*_p p^{t'}}\, ,\, \overline{\mathcal{N}}_{\hat p\succeq^*_p p^{t'}}]$.

\refprop{prop:p_welfare} tells us that the revealed preference relations are tight, in the sense that if, for some consumer, $\hat p$ is not revealed preferred to $p^{t'}$ then there exists an augmented utility function which rationalizes her consumption choices and for which she strictly prefers $p^{t'}$ to $\hat p$.  Given this, we know that, amongst all rationalizations of $\mathscr D$, $\underline{\mathcal{N}}_{\hat p\succeq^*_p p^{t'}}$ is also the infimum on the proportion of consumers who are better off at $\hat p$ compared to $p^{t'}$.

The following proposition summarizes these observations.

\begin{proposition} \label{prop:gocompare}
	Let $\mathscr D=\{(p^t,\mathring{\pi}^t)\}_{t=1}^T$ be a repeated cross-sectional data set that satisfies \refass{ass:indif} and is rationalized by the RAUM.  Then, for every $\eta \in [\underline{\mathcal{N}}_{\hat p\succeq^*_p p^{t'}}\, ,\, \overline{\mathcal{N}}_{\hat p\succeq^*_p p^{t'}}]$, there is a rationalization of $\mathscr D$ for which $\eta$ is the proportion of consumers who are revealed better off at $\hat p$ compared to $p^{t'}$.
	
	Furthermore,  $\underline{\mathcal{N}}_{\hat p\succeq^*_p p^{t'}}$ is the infimum of the proportion of consumers who are better off at $\hat p$ compared to $p^{t'}$, among all the rationalizations of $\mathscr{D}$.
\end{proposition}

It may be helpful to consider how \refprop{prop:gocompare} applies in \refeg{eg:StTesting}.  In that case, there are three GAPP-consistent types with a {\em unique} $\nu$ that solves $A\nu=\pi$ (see (\ref{eq:nus2})).  Of the three types, $p^t\succeq^*_p p^{t'}$ holds only for type 2 (see \reffig{fig:rational_types}) and thus the proportion of consumers who are revealed better off at $p^t$ compared to $p^{t'}$ is $\nu_2=1/2$.  Formally, we have $\mathbb{1}_{p^t\succeq^*_p p^{t'}}=(0,1,0)$, $\mathbb{1}_{p^t\succeq^*_p p^{t'}}\nu =1/2$, and $\underline{\mathcal{N}}_{ p^t\succeq^*_p p^{t'}}=\overline{\mathcal{N}}_{p^t\succeq^*_p p^{t'}}=1/2$.\footnote{Of the other two types in the population, type 3 (with $\nu_3=2/5$) are revealed better off at $p^{t'}$ compared to $p^t$, while type 1 consumers could be either better or worse at $p^t$ compared to $p^{t'}$. Therefore, across all rationalizations of that data set, the proportion of consumers who are better off (but not necessarily revealed better off) at $p^t$ compared to $p^{t'}$ can be as low as $1/2$ and as high as $1-2/5=3/5$.}

\section{Statistical Test of RAUM, and Inference for Counterfactuals} \label{sec:metrics}

This section outlines our econometric methodologies. First, \refsec{sec:econ_test} provides a statistical test of the RAUM (presented in \refsec{sec:RAUMX}). Second, and more importantly, \refsec{sec:econ_wel} develops a new methodology for obtaining asymptotically uniformly valid confidence intervals for counterfactual objects. This result applies to a general class of random utility models, including the RAUM.  It can be used for statistical analyses of welfare comparisons and we use it for that purpose in our empirical study in \refsec{sec:application_RAUM}.

\subsection{Testing the Random Augmented Utility Model}\label{sec:econ_test}

Recall from \refthm{thm:StGAPPTest} that, given a set of prices and corresponding demand distributions $\mathscr{D}=\{(p^t,\mathring{\pi}^t)\}_{t=1}^T$ and an implied vector $\pi$ of choice probabilities on rescaled and discretized budgets, a test of the random augmented utility model is a test of
\begin{eqnarray}
	H_0: && \exists \nu\in \Delta^{|\mathcal{A}| -1} \text{ such that } A\nu=\pi \notag \\
	\Longleftrightarrow && \min_{\nu\in \Re^{|\mathcal{A}|}_+}[\pi-A \nu]'\Omega[\pi-A \nu]=0, \label{eqn:hyp_main}
\end{eqnarray}
where $\Omega$ is a positive definite matrix and where the equivalence was noted and exploited in \citetalias{kitamura2018}.\footnote{The strategy to configure $H_0$ as a quadratic program also appears in \citet{dePaula2018}, albeit for a different program and in a different context.}

In practice, we estimate $\pi$ by its sample analog $\hat{\pi}=(\hat{\pi}^1,\dots,\hat{\pi}^T)$ obtained by rescaling the empirical distribution of choices $\{x^t_{n_t}\}_{n_t=1}^{N_t}$ where $N_t$ is the number of observed choices in the data in period $t$.
This gives rise to test statistic
\begin{equation}\label{eqn:test_stat}
	J_N:=N\min_{\nu\in \Re^{|\mathcal{A}|}_+}[\hat{\pi}-A \nu]'\Omega[\hat{\pi}-A \nu],
\end{equation}
where $N=\sum_{t=1}^T N_t$ denotes the total number of observations.
Computing appropriate critical values for this test is delicate because the limiting distribution of $J_N$ depends discontinuously on nuisance parameters. We use the modified bootstrap procedure proposed by \citetalias{kitamura2018}.

\subsection{Inference for Counterfactuals in a General Class of Random Utility Models}\label{sec:econ_wel}

A counterfactual quantity in a random utility model can be generally regarded as a function of the underlying distribution $\nu$ of individual preferences.  This section focuses on the case where this mapping is linear, so that we are concerned with statistical inference for $\theta = \rho \cdot \nu$, where $\rho \in \Re^{|\mathcal{A}|}$ is a known vector which varies with the counterfactual of interest. Our analysis of welfare comparisons in \refsec{sec:st_welfare} falls into this framework, by letting $\theta$ be the proportion of consumers who are revealed better off at prices $\hat p$ compared to $p^{t'}$, with $\rho = \mathbb{1}_{{\hat p}\succeq^*_p p^{t'}}$.  It is worth emphasizing that the methodology developed in this section has broad applicability: it can be used to study other random utility models (such as the model in \citet{KS19}) and to investigate other objects of interest in random utility models; for example, \citet{lazzati2018} applies our technique to estimate the proportion of non-strategic players in a game.

Note that  $\theta$ is partially identified as follows:
\begin{equation*}\label{eqn:general}
	\theta  \in \Theta_I \quad \text{where} \quad \Theta_I := \{\rho \cdot \nu| \nu \geq 0, A \nu=\pi\}.
\end{equation*}
Our confidence interval inverts a test of
\begin{eqnarray}
	\pi \in \mathcal S(\theta) \quad \text{where} \quad  \mathcal S(\theta) := \left\{A \nu \, |\,   \rho \cdot \nu = \theta, \; \nu \in \Delta^{|\mathcal{A}| -1} \right\} \label{eq:v-test}
\end{eqnarray}
or equivalently,
$$\min_{\nu\in \Delta^{|\mathcal{A}| -1},  \, \theta = \rho \cdot \nu}[\pi-A \nu]'\Omega[\pi-A \nu]=0.$$
The test statistic is a scaled sample analog
\begin{eqnarray}
	J_N(\theta)
	&=& N \min_{\nu\in \Delta^{|\mathcal{A}| -1},  \, \theta = \rho \cdot \nu}[\hat \pi-A  \nu]'\Omega[\hat \pi-A \nu] \notag
	\\
	&=&N \min_{\eta \in \mathcal S(\theta)}	[\hat
	\pi
	- \eta]^{\prime }\Omega \lbrack \hat
	\pi- \eta]. \label{eqn:J_N_hyp}
\end{eqnarray}
Once again, the naive bootstrap fails to deliver valid critical values for \eqref{eqn:J_N_hyp}, since its asymptotic distribution changes discontinuously, depending on the location of $\pi$ relative to the polytope $\mathcal S(\theta)$.  This is akin to the nonstandard nature of the inference in  \eqref{eqn:test_stat}, though a simple application of the modified bootstrap algorithm in \KS does not work, as their method relies on, among other things, the polytope $\{A  \nu:\nu \geq 0\}$ being a cone.  This is not necessarily the case for counterfactual analysis, and we need to deal with $\mathcal S(\theta)$ without relying on conical properties.  In this section we develop a new algorithm that guarantees asymptotic validity  for inference concerning general counterfactuals.

That said, as in \citetalias{kitamura2018}, we do gain an insight from Weyl-Minkowski duality. In \refappendix{sec:appendix_metrics}, we show that there exist nonstochastic matrices $B$, $\tilde B$ and a nonstochastic vector-valued function $d(\theta)$ such that $\pi \in \mathcal S(\theta)$ if, and only if,
\begin{equation}\label{eq:h-test}
	B\pi \leq 0,\;  \tilde B\pi = d(\theta)\; \text{ and }\; {\bf 1} \cdot \pi = 1,
\end{equation}
where $\bf 1$ is the $I$-vector of ones where $I=\sum_{t=1}^T I_t$ is the total number of patches. Thus, in principle this is a linear (in)equality testing problem. There is a rich literature on such problems. However, we cannot directly invoke that literature because we cannot compute $(B,\tilde B)$ in practice for a problem with a relevant scale.

While we therefore need to work with representation \eqref{eq:v-test}, representation \eqref{eq:h-test} is useful. It illustrates that the inference problem is non-standard; in particular, the limiting distribution of the test statistic depends on how close to binding each of the constraints encoded in $(B,\tilde B,d(\theta))$ is. From analogy to the moment inequalities literature, it also pretty much implies that the constraints' slackness cannot be pre-estimated with sufficient accuracy; the reason being that it enters the test's asymptotic representation scaled by $\sqrt{N}$. However, we also know that certain existing procedures which shrink the estimated slack of all inequalities to zero before computing the distribution of $J_N$ will work. Our proposal is inspired by these but must implement the idea with the computationally feasible representation \eqref{eq:v-test} instead of \eqref{eq:h-test}, which is only theoretically available.  This means that we cannot calculate the empirical slack, which is explicit in (the empirical version of) representation \eqref{eq:h-test} but not in  \eqref{eq:v-test}, the very reason why a new method is called for.

Intuitively, we contract (or ``tighten") the polytope $\mathcal S(\theta)$ toward a point in its relative interior, thereby effectively (but non-obviously) reducing the empirical slack in any inequality constraint.  This forces all the constraints with small slacks to be binding after ``tightening". Note that, unlike in \citetalias{kitamura2018}, we face substantial added complications because  (i) we need to deal with a non-conical $\mathcal S(\theta)$, and (ii) the appropriate way to tighten the polytope $\mathcal S(\theta)$ varies with the value of $\theta$ through the dependence  of $\mathcal S(\theta)$ on $\theta$. This leads to a \textit{restriction-dependent tightening} approach which we now describe in broad strokes.

\medskip

Choose a sequence $\tau_N$ such that $\tau _{N}\downarrow 0$ and $\sqrt{N}\tau _{N}\uparrow \infty $ (we make a specific proposal in the appendix) and define
$$
\mathcal S_{\tau_N}(\theta) := \{A\nu\; | \; \rho \cdot \nu = \theta, \nu \in \mathcal V_{\tau_{N}}(\theta)\},
$$
where $\mathcal V_{\tau_{N}}(\theta)$ is obtained by appropriately constricting $\Delta^{|\mathcal{A}|-1}$; in particular, some components of $\nu$ are forced to be boundedly above $0$. Note that $\mathcal S_{\tau_N}(\theta)$ depends on $\theta$ through the equation $\rho \cdot \nu = \theta$ but also because, as the notation suggests, the construction of $\mathcal V_{\tau_{N}}(\theta)$ will change with $\theta$, a key feature of our algorithm. The definition of $\mathcal V_{\tau_{N}}(\theta)$ for general $\rho$ is rather involved and thus deferred to \refappendix{sec:appendix_metrics}, but it considerably simplifies for binary $\rho$ as in our application.

The set $\mathcal S_{\tau_N}(\theta)$ replaces $\mathcal S(\theta)$ in the bootstrap population. The precise algorithm proceeds as follows.  For each $\theta \in \Theta$:
\begin{itemize}
	\item[(i)] Compute the \textit{$\tau_N$-tightened restricted} estimator of the empirical choice distribution
	\begin{equation*}
		\hat{\eta}_{\tau _{N}}:=\argmin_{\eta \in \mathcal S_{\tau_{N}}(\theta)}	N[\hat{\pi} - \eta]^{\prime }\Omega \lbrack \hat{\pi}-\eta].
	\end{equation*}		
	\item[(ii)] Define the \textit{$\tau_N$-tightened recentered} bootstrap estimators
	\begin{equation*}
		\hat \pi^{*(r)}_{\tau_N} := \hat \pi^{*(r)} - \hat \pi + \hat \eta_{\tau_N},
		\quad r = 1,...,R,
	\end{equation*}
	where $\hat{\pi}^{*(r)}$ is a bootstrap analog of $\hat{\pi}$ and $R$ is the number of bootstrap samples. For instance, in our application, $\hat{\pi}^{*(r)}$ is generated by the simple nonparametric bootstrap of choice frequencies.
	\item[(iii)] For each $r = 1,...,R$, compute
	\begin{eqnarray*}
		J_{N,\tau_{N}}^{*(r)}(\theta)
		=\min_{\eta \in \mathcal S_{\tau_{N}}(\theta)}	N[\hat
		\pi^{*(r)}_{\tau_N}
		- \eta]^{\prime }\Omega \lbrack \hat
		\pi^{*(r)}_{\tau_N} - \eta].
	\end{eqnarray*}	
	\item[(iv)] Use the empirical distribution of $	J_{N,\tau_{N}}^{*(r)}(\theta)$ to obtain the critical value for $J_{N}(\theta)$.
\end{itemize}
A confidence interval for $\theta$ collects values of $\theta$ that are not rejected.

\medskip

\refthm{thm:validity} below establishes asymptotic validity of the above procedure.  Let
$$
\mathcal F := \left\{(\theta,\pi) \; \left| \; \theta \in \overline{\vartheta}, \pi \in \mathcal S(\theta) \cup \mathcal P \right. \right\}
$$
where ${\mathcal{P}}$ denote the set of all $\pi$ that satisfy \refcond{condition 1} in \refappendix{sec:appendix_metrics}.

\begin{theorem}
	\label{thm:validity} Choose $\tau _{N}$ so that $\tau _{N}\downarrow 0$ and $\sqrt{N}%
	\tau _{N}\uparrow \infty $. Also, let $\Omega $ be diagonal. Then under Assumptions \ref{ass:proportions} and \ref{ass:sampling} stated in \refappendix{sec:appendix_metrics},
	\begin{equation*}
		\liminf_{N\rightarrow \infty }\inf_{(\theta,\pi) \in \mathcal F}\Pr
		\{J_N(\theta)\leq \hat{c}_{1-\alpha }\}=1-\alpha,
	\end{equation*}%
	where $0\leq \alpha \leq \frac{1}{2}$ and $\hat{c}_{1-\alpha }$ is the $1-\alpha $ quantile of $J_{N,\tau_{N}}^*$.
\end{theorem}

The proof of \refthm{thm:validity} is in \refappendix{sec:appendix_metrics}.

\section{Empirical Applications}\label{sec:EmpApp}

In this section, we present two separate applications meant to demonstrate how both the deterministic and random versions of our model can be tested and  employed for welfare analysis.

\subsection{Augmented utility model: testing and welfare analysis on Progresa data}\label{sec:progresa}

We apply the deterministic model to the Progresa-Oportunidades data set, a workhorse of the treatment evaluation literature. Progresa was a conditional cash transfer program aimed at poor communities in Mexico. The program was remarkable in that it was rolled out in random order so the causal effect of the cash transfers could be studied. For brevity, we do not describe the program in detail; information on the program is widely available including in the paper we discuss next.

Our application builds on recent work of \citet{attanasio2020} (henceforth \citetalias{attanasio2020}) who analyze whether the program led to changes in the market prices for basic staples: rice, kidney beans, and sugar. This is an important question because the welfare effect of these transfers would clearly depend in part on their impact on prices. While the previous literature had documented that \textit{average} prices were not affected by the program \cite{hoddinott2000}, \AP argue that sellers charge nonlinear prices and that these nonlinear price schedules had changed.

Because treatment was randomized across villages but means-tested at the household level, some households faced a changing price schedule but no shock to their own income. In our study, we focus our attention on these households because we can be more confident that their augmented utility functions are unchanged across the observation periods.  Our objectives are, firstly, to test the augmented utility model and, secondly, to evaluate the welfare impact of price changes using that model.  This data set is well suited for analysis using our deterministic model because its panel structure means that we can study each household separately. Following \citetalias{attanasio2020}, we consider nonlinear prices, which allows us to implement the results in \refsec{sec:nonlinearGAPP}.

The theoretical part of \AP derives the optimal (nonlinear) pricing schedule under the assumption that there is a heterogenous population of households, with each household maximizing a quasilinear utility function, subject to a subsistence constraint.  This constraint stipulates that a household needs a certain minimum number of calories, which can be obtained from either the observed bundle $x$ or the numeraire; given $x$, the minimum amount of the numeraire good needed to meet the calorie threshold is denoted by $\underline{z}(x)$. Thus the household can only choose among those bundles $x$ such that $\psi (x)+\underline{z}(x)\leq M$, where $\psi$ is the price system and $M$ is household wealth. It is worth noting that the augmented utility framework is sufficiently flexible to accommodate this behavior.  Indeed, the  household could be thought of as maximizing an augmented utility function of the modified-quasilinear form
$$U(x,-e)=\widetilde U(x)- \mathbf{K}(e+\underline{z}(x)-M)\, e,$$
where $\mathbf{K}(w)=1$ if $w\leq 0$ and $\mathbf{K}(w)$ is a very large positive number if $w>0$. In this way, any $(x,-e)$ (a bundle and its associated expenditure) that leads to a violation of the subsistence constraint incurs a very large disutility and so will never be chosen.

%Because of the additional constraint, we cannot model this consumer as maximizing an unconstrained quasilinear utility function, but we can model her as maximizing an augmented utility $U(x,-e)$. To see this, recall that we only require $U$ to be strictly increasing in its second argument, that is, it need not be increasing in $x$. Therefore, $U$ can become low at values of $x$, $\psi(x)$ that violate the subsistence constraint and so these values will never be optimal.\footnote{Replacing $e=\psi(x)$, the constraint is $e+\underline{z}(x)\leq M$ and so, for fixed $x$, if $e$ violates the constraint, so will all higher values of $e$. Thus, we can use an augmented utility $U(x,-e)$ to model this consumer because we can assign low values to $U(x,-e)$ at $x$, $e$ that violate the constraint without compromising monotonicity of $U$ in the second argument.} In other words, the generality of our augmented utility framework allows us to capture behavior that is equivalent to a consumer with quasilinear preferences facing additional constraints.

We work with \citetalias{attanasio2020}'s data and refer to them for a detailed explanation. Compared to their analysis, we restrict ourselves to the narrower definition of village (``locality") because the larger units of analysis (``municipality") may not be contained in either the treatment or the control group.  Also, because we are interested in intertemporal within-village price variation, we estimate separate price schedules for the same village in different waves as opposed to one price schedule (estimated across waves) per village. This necessitates being slightly more permissive about data needs, and we estimate prices for all village-good-wave triples that have $20$ or more (as opposed to $75$ or more) observations. We follow \AP in rejecting data for villages where prices strictly increase with quantity sold and where there is insufficient variation in quantities purchased.

\medskip

We estimate the  price schedule for good $i$ in village $v$ at wave $t$ by applying Ordinary Least Squares to
\begin{equation}\label{eq:tariffs}
	\log (\psi_{vti}(q_{vtih})) = b_{vti0} + b_{vti1}\log(x_{vtih}) + \varepsilon_{vtih}.
\end{equation}
Here $h$ indexes households and $\psi_{vti}(q_{vtih}) =\mathbb{E}[p_{vti}(x_{vtih}) |x_{vtih}]\,x_{vtih}$, where $p_{vti}(x_{vtih})$ is the unit price corresponding to quantity $x_{vtih}$, $\varepsilon$ is measurement error, and the expected value is taken over the empirical distribution of reported unit prices corresponding to the same quantity purchased of good $i$ in village-wave $(v, t)$. This is exactly Equation (15) in \AP except for being estimated at a less aggregated level.

\medskip

We test GAPP on households that:
\begin{itemize}
	\item were not eligible for Progresa transfers,\footnote{Specifically, we include households that were not eligible for Progresa in any of the survey waves and households that were eligible in 2003 {\em only}; for the latter, we exclude their 2003 observations from the test. Eligibility increased dramatically in the 2003 wave, so totally excluding those households would lead to many fewer observations.}
	\item were observed for more than one wave in treated villages for which we estimate price schedules for each of rice, kidney beans, and sugar,
	\item purchased at least one of these goods.
\end{itemize}
In our final sample, this leaves us with $2488$ households in $177$ villages.\footnote{For 554 of these households we have two observations, for 840 households we have three, for 934 households we have four, and for 160 households we have five.  There are so few with five observations because many households were enrolled into the program in the final wave and thus removed from our sample.}

We emphasize that GAPP is not vacuously satisfied on these data. Recall that GAPP cannot be violated when two price systems $\psi$, $\psi'$ are ranked, in the sense that $\psi(x)\geq\psi'(x)$ for all $x\in \mathbb{R}^L_+$. Of the $20556$ possible combinations of pairs of waves encountered by  households in the data, about $4\%$ have this feature, and only $20$ out of $2488$ households exclusively face such price pairs and therefore satisfy  GAPP vacuously.  Nonetheless, $83\%$ of households pass the GAPP test. Most violations were small in the sense of the rationality index $\vartheta$ (defined in \refsec{sec:index}) being close to $1$: fewer than $1\%$ of households were below $.9$, and fewer than $4\%$ were below $.95$.

\begin{table}
	\begin{center}
		\begin{tabular}{ccccccc}
			& & \textbf{10/98} & \textbf{03/99} & \textbf{11/99} & \textbf{11/00} & \textbf{2003}\\[1.5ex]
			\textbf{10/98} &&  & .035 & .024 & .006 & 0\\
			\textbf{03/99} && .913 &  & .198 & .052 & .015\\
			\textbf{11/99} && .936 & .686 &  & .105 & .034\\
			\textbf{11/00} && .981 & .914 & .847 &  & .240\\
			\textbf{2003} && 1 & .980 & .927 & .520 &
		\end{tabular}
		\caption{Fraction of GAPP rationalizable consumers revealed preferring the row wave to the column wave.}
		\label{table:RP}
	\end{center}
\end{table}
\begin{table}
	\begin{center}
		\begin{tabular}{rcccccc}
			& & \textbf{03/99} & \textbf{11/99} & \textbf{11/00} \\[1.5ex]
			\textbf{75th percentile} && 5.36  & 7.26 & 11.65 \\
			\textbf{Median} && 3.27 & 4.66 & 6.98   \\
			\textbf{25th percentile} && 1.58 & 2.44 & 4.12 			
		\end{tabular}
		\caption{Lower bound of the compensating variation, with 10/98 as the base}
		\label{table:CV2}
	\end{center}
\end{table}

We carried out some illustrative welfare analysis, the results of which are displayed in Tables \ref{table:RP} and \ref{table:CV2}. \reftab{table:RP} displays the fractions of GAPP-compliant households that reveal prefer a given wave to another wave. Specifically, each cell in the table corresponds to the fraction of GAPP-rationalizable consumers who reveal prefer (directly or indirectly) the price system in the row wave to the price system in the corresponding column wave.\footnote{Note that the (indirect) revealed preference relation $\succeq^*_p$ uses demand information at {\em all} waves in each binary comparison; see the definition of $\succeq^*_p$ in \refsec{sec:nonlinearGAPP}.}  Notice that the data indicates a strong tendency to prefer price systems in later waves. For example, 91.3\% of households reveal prefer prices in 03/99 to those in 10/98; the same is true even more strongly when 10/98 is compared against later waves.

To have a sense of the scale of this welfare improvement over time, we calculate, for each household, the lower bound on the compensating variation, with the price system faced by the household at 10/98 as the base.\footnote{The formula for the lower bound when prices are nonlinear is in \refsec{sec:morecomvar}.}  These values are then ranked.  The results are displayed in Table \ref{table:CV2}.  Since more than 90\% of households reveal prefer (price systems at) subsequent waves to 10/98, the lower bound of the compensating variation must be positive for more than 90\% of households.  For example, between 03/99 and 10/98, the median compensating variation is 3.27; thus, based on its observed behavior, one could remove 3.27 from this household in 03/99 and still leave it as well off in 03/99 as in 10/98.  Note that the values in this table are not small, given that the household median expenditure in 10/98 on the items considered is 27.48.

These results are consistent with \citetalias{attanasio2020}'s finding that the change in the income distribution induced by Progresa caused a change in sellers' intensity of price discrimination. As a result, poorer households faced higher average prices and wealthier households faced lower ones; since Progresa was means-tested, untreated households fall into the latter category. Thus, the general equilibrium effects of the program could be the reason for the welfare improvements observed in untreated households.

\subsection{RAUM: Testing and welfare analysis on household expenditure data}\label{sec:application_RAUM}

We test the RAUM and conduct welfare analyses on two repeated cross-sectional data sets: the U.K. Family Expenditure Survey (FES) and the Canadian Surveys of Household Spending (SHS). Our aim is to show that the data supports the model and to demonstrate that the estimated welfare bounds are informatively tight.

We first present the analysis for the FES, which is widely used in the nonparametric demand estimation literature (for instance, by \citet{blundell2008}, \citetalias{kitamura2018}, \citet{hoderlein2014}, \citet{adams2020}, and \citet{Kawaguchi}). In the FES, about 7000 households are interviewed each year and they report their consumption expenditures in different commodity groups. Following \citet{blundell2008} we derive the real consumption level for each commodity group by deflating it with a price index for that group (which is taken from the annual Retail Prices index).  Again following \citet{blundell2008}, we restrict attention to households with cars and children, leaving us with roughly 25\% of the original data. We implement tests for $3$, $4$, and $5$ composite goods. The coarsest partition of $3$ goods---food, services, and nondurables---is precisely what is examined by \citet{blundell2008} (and we use their replication files).  As in \citetalias{kitamura2018}, we introduce more commodities by first separating out clothing and then alcoholic beverages from the nondurables.

The data set we have is the sample analog of $\mathscr{D}=\{(p^t,\mathring{\pi}^t)\}_{t=1}^T$ as defined in \refsec{sec:RatRAUM}. It is worth reiterating the point that we made at the end of \refsec{sec:RatRAUM}: even though this data set is {\em not} iso-expenditure, we can directly test the RAUM on this data; this contrasts with testing the RUM on this data, which cannot be done directly and must involve a further procedure to estimate an iso-expenditure data set.

We implement the test in blocks of 6 years, i.e., we set $T=6$.  We avoid covering a longer period partly due to the computational demands of calculating $A$ (the matrix of GARP-consistent types; see (\ref{eqn:A_def})),\footnote{That said, new techniques developed in \cite{smeulders2021} have significantly reduced the computational demands of the problem.} but also because a time-invariant distribution of augmented utility functions is only plausible over shorter time horizons, for example because of long term first-order changes to the U.K. income distribution \citep{jenkins2016}.

\reftab{table:FES_Test} displays our results, with columns correspond to different blocks of 6 years and rows contain the values of the test statistic and the corresponding p-values.  The test statistic $J_N$ is defined by (\ref{eqn:test_stat}), with the identity matrix serving as $\Omega$. Notice that for the year block 90-95, the test statistic is zero; this means that the sample distribution $\hat\pi$ satisfies the rationality condition in \refthm{thm:StGAPPTest} exactly. That is, there is a distribution $\nu$ on GARP-consistent types such that $\hat\pi=A\nu$.  Apart from this case, the sample distribution does not exactly satisfy the rationality condition and so the test statistic is strictly positive; nonetheless it is very clear from the p-values that, overall, our model is not rejected by the FES data.

\begin{table}	
	\hspace*{-2.5cm}
	\begin{subtable}{\textwidth}
		\begin{tabular}{ccc|c|c|c|c|c|c|c|c|c}
			\hline
			&  & \multicolumn{10}{c}{Year Blocks} \\
			\cmidrule(l){3-12}
			& 	&	75-80	&	76-81	&	77-82	&	78-83	&	79-84	&	80-85	&	81-86	&	82-87	&	83-88	&	84-89	\\
			\cmidrule(l){3-12}
			\multirow{2}{*}{3 Goods}  & Test Statistic ($J_N$)	&	0.337	&	0.917	&	0.899	&	0.522	&	0.018	&	0.082	&	0.088	&	0.095	&	0.481	&	0.556	\\
			& p-value	&	0.04	&	0.34	&	0.55	&	0.59	&	0.99	&	0.67	&	0.81	&	0.91	&	0.61	&	0.48	\\
			\cmidrule(l){3-12}
			\multirow{2}{*}{4 Goods} & Test Statistic ($J_N$)	&	0.4	&	0.698	&	0.651	&	0.236	&	0.056	&	0.036	&	0.037	&	0.043	&	0.043	&	0.232	\\
			& p-value	&	0.25	&	0.58	&	0.63	&	0.91	&	0.96	&	0.99	&	0.96	&	0.95	&	0.99	&	0.68	\\
			\cmidrule(l){3-12}
			\multirow{ 2}{*}{5 Goods} & Test Statistic ($J_N$)	&	0.4	&	0.687	&	0.705	&	0.329	&	0.003	&	0.082	&	0.088	&	0.104	&	0.103	&	0.144	\\
			& p-value	&	0.3	&	0.66	&	0.68	&	0.88	&	0.999	&	0.96	&	0.79	&	0.85	&	0.9	&	0.83 \\
			\cmidrule(l){3-12}
			\hline \\
		\end{tabular}
	\end{subtable}
	\hspace*{-2.5cm}
	\begin{subtable}{\textwidth}
		\begin{tabular}{ccc|c|c|c|c|c|c|c|c|c}
			\hline
			&  & \multicolumn{10}{c}{Year Blocks} \\
			\cmidrule(l){3-12}
			&	&	85-90	&	86-91	&	87-92	&	88-93	&	89-94	&	90-95	&	91-96	&	92-97	&	93-98	&	94-99	\\
			\cmidrule(l){3-12}
			\multirow{2}{*}{3 Goods}  & Test Statistic ($J_N$)	&	0.027	&	1.42	&	2.94	&	1.51	&	1.72	&	0	&	0.313	&	0.7	&	0.676	&	0.26	\\
			& p-value	&	0.69	&	0.3	&	0.18	&	0.24	&	0.21	&	1	&	0.59	&	0.48	&	0.6	&	0.83	\\
			\cmidrule(l){3-12}
			\multirow{2}{*}{4 Goods}  & Test Statistic ($J_N$)	&	0.227	&	0.025	&	0.157	&	0.154	&	0.004	&	1.01	&	0.802	&	0.872	&	0.904	&	0.604	\\
			& p-value	&	0.48	&	0.96	&	0.8	&	0.73	&	0.97	&	0.21	&	0.31	&	0.57	&	0.65	&	0.74	\\
			\cmidrule(l){3-12}
			\multirow{2}{*}{5 Goods}  & Test Statistic ($J_N$)	&	0.031	&	0.019	&	0.018	&	0.019	&	0.023	&	0.734	&	0.612	&	0.643	&	0.634	&	0.488	\\
			& p-value	&	0.85	&	0.98	&	0.97	&	0.91	&	0.83	&	0.22	&	0.4	&	0.72	&	0.78	&	0.79 \\
			\cmidrule(l){3-12}
			\hline
		\end{tabular}
	\end{subtable}
	\caption{Test Statistics and  p-values for sequences of $6$ budgets of the FES. Bootstrap size is $R=1000$.}
	\label{table:FES_Test}
\end{table}

\begin{table}
	\begin{tabular}{ccc}
		\hline
		Comparison & Estimated Bounds & Confidence Interval \\
		$p^{1976}\succ^*_p p^{1977}$ & $[.150,.155]$ &  $[.13,.183]$ \\
		$p^{1977}\succ^*_p p^{1976}$ & $\{.803\}$ & $[.784,.831]$ \\
		$p^{1979}\succ^*_p p^{1980}$ & $[.517,.530]$ & $[.487,.56]$ \\
		$p^{1980}\succ^*_p p^{1979}$ & $\{.463\}$ &  $[.436,.497]$ \\
		\hline
	\end{tabular}
	\caption{Estimated bounds and confidence intervals for the proportion of consumers who reveal prefer one price to another one in the FES data. Data used are for 1975-1980. Bootstrap size is $R=1000$.}
	\label{table:FES_Bounds}
\end{table}

We also estimated the bounds $[\underline{\mathcal{N}}_{p^t\succeq^*_p p^{t'}}\, ,\, \overline{\mathcal{N}}_{p^{t}\succeq^*_p p^{t'}}]$ (as defined by (\ref{eqn:wel_lb}) and (\ref{eqn:wel_ub})) on the proportion of households that are revealed better off at prices $p^t$ than at prices $p^{t'}$. For brevity, we present a few representative estimates using data for the years 1975-1980 in \reftab{table:FES_Bounds}.  The column `Estimated Bounds' are the bounds obtained by calculating $\mathbb{1}_{p^{t}\succeq^*_p p^{t'}}\, \nu$ from the (not necessarily unique) values of $\nu$ that minimize the test statistic (\ref{eqn:test_stat}).  In two cases this estimate is unique while it is not in the other two cases.  Applying the procedure set out for calculating confidence intervals in \refsec{sec:econ_wel}, we obtain the intervals displayed (which must necessarily contain the estimated bounds). It is worth noting that the width of these intervals is less than .1 throughout, so they are quite informative.\footnote{Note that, even if the true values of the proportion of the population satisfying $p^{t}\succ^*_p p^{t'}$ and $p^{t'}\succ^*_p p^{t}$ are known, they will typically add up to strictly less than 1 because, for part of the population, there will be no revealed preference relation between $p^{t}$ and $p^{t'}$. For example, type 1 consumers in \refeg{eg:StTesting} have no revealed preference relation between $p^{t_1}$ and $p^{t_2}$.}

For our second empirical application using Canadian data, we use the replication kit of \citet{norris2013,norris2015}. Like the FES, the SHS is a publicly available, annual data set of household expenditures in a variety of different categories.  We study annual expenditure in 5 categories that constitute a large share of the overall expenditure on nondurables: food purchased (at home and in restaurants), clothing and footwear, health and personal care, recreation, and alcohol and tobacco. The SHS data is rich enough to allow us to analyze the data separately for the nine most populous provinces: Alberta, British Columbia, Manitoba, New Brunswick, Newfoundland, Nova Scotia, Ontario, Quebec, and Saskatchewan. The number of households in each province-year range from $291$ (Manitoba, 1997) to $2515$ (Ontario, 1997). We use province-year prices indices (as constructed by \citet{norris2015}) and deflate them using province-year CPI data from Statistics Canada to get real price indices.

\begin{table}
	\hspace*{-1cm}
	\begin{tabular}{ccc|c|c|c|c|c|c|c}
		\hline
		Province &  & \multicolumn{8}{c}{Year Blocks} \\
		\cmidrule(l){3-10}
		& & 97-02 & 98-03 & 99-04 & 00-05 & 01-06 & 02-07 & 03-08 & 04-09 \\
		\cmidrule(l){3-10}
		\multirow{2}{*}{Alberta} & Test Statistic ($J_N$) & .07 & 0 & 0 & 0 & 0 & 0 & .003 & 4.65\\
		& p-value & .94 & 1 & 1 & 1 & 1 & 1 & .98 & .04 \\
		\cmidrule(l){3-10}
		\multirow{2}{*}{British Columbia} & Test Statistic ($J_N$) &.89 & .56 & .48 & .07 & .05 & 6.23 & 8.87 & 8.71 \\
		& p-value & .47 & .47 & .98 & .96 & .97 & .05 & .02 & .01 \\
		\cmidrule(l){3-10}
		\multirow{2}{*}{Manitoba} & Test Statistic ($J_N$) & 0 & 0 & 0 & 0 & 0 & 0 & .01 & .01 \\
		& p-value & 1 & 1 & 1 & 1 & 1 & 1 & 1 & 1 \\
		\cmidrule(l){3-10}
		\multirow{2}{*}{New Brunswick} & Test Statistic ($J_N$) & .08 & .05 & 0 & 0 & 0 & .60 & .58 & .57 \\
		& p-value & .94 & .94 & 1 & 1 & 1 & .58 & .79 & .68 \\
		\cmidrule(l){3-10}
		\multirow{2}{*}{Newfoundland} & Test Statistic ($J_N$) & .10 & .32 & .29 & .29 & .38 & 3.08 & 2.30 & 2.08 \\
		& p-value & .85 & .90 & .91 & .87 & .81 & .21 & .35 & .27 \\
		\cmidrule(l){3-10}
		\multirow{2}{*}{Nova Scotia} & Test Statistic ($J_N$) & .05 & .03 & 0 & 0 & 0 & 0 & .93 & 1.02 \\
		& p-value & .97 & .98 & 1 & 1 & 1 & 1 & .69 & .58 \\
		\cmidrule(l){3-10}
		\multirow{2}{*}{Ontario} & Test Statistic ($J_N$) & .064 & .040 & .035 & 0 & 0 & 0 & 0 & 0 \\
		& p-value & .98 & .95 & .91 & 1 & 1 & 1 & 1 & 1 \\
		\cmidrule(l){3-10}
		\multirow{2}{*}{Quebec} & Test Statistic ($J_N$) & .11 & 0 & 0 & 0 & 0 & .51 & .54 & .49 \\
		& p-value & .88 & 1 & 1 & 1 & 1 & .67 & .67 & .65 \\
		\cmidrule(l){3-10}
		\multirow{2}{*}{Saskatchewan} & Test Statistic ($J_N$) & 0 & 0 & 0 & 0 & 0 & .02 & .02 & 0 \\
		& p-value & 1 & 1 & 1 & 1 & 1 & 1 & 1 & 1 \\
		\cmidrule(l){3-10}
		\hline
	\end{tabular}
	\caption{Test Statistics and  p-values for sequences of $6$ budgets of the SHS. Bootstrap size is $R=1000$.}
	\label{table:SHS_Test}
\end{table}

\begin{table}
	\begin{tabular}{ccc}
		\hline
		Comparison & Estimated Bounds & Confidence Interval \\
		$p^{1998}\succ^*_p p^{2001}$ & $\{.099\}$ &  $[.073,.125]$ \\
		$p^{2001}\succ^*_p p^{1998}$ & $\{.901\}$ & $[.875,.927]$ \\
		$p^{1999}\succ^*_p p^{2002}$ & $[.299,.341]$ & $[.272,.385]$ \\
		$p^{2002}\succ^*_p p^{1999}$ & $[.624,.701]$ &  $[.594,.728]$ \\
		\hline
	\end{tabular}
	\caption{Estimated bounds and confidence intervals for the proportion of consumers who reveal prefer one price to another one in the SHS data. Data used are for 1997-2002 in British Columbia. Bootstrap size is $R=1000$.}
	\label{table:SHS_Test}
\end{table}

\reftab{table:SHS_Test} displays the test statistics and associated p-value for each province and every 6 year block. Compared to the FES data, there are two notable differences. The first is that many more test statistics are exactly zero; that is, the observed choice frequencies are rationalized by the random augmented utility model. The second is that, for a small proportion of year blocks, there are statistically significant positive test statistics (in particular, the last three columns for British Columbia). Nonetheless, the p-values taken together do not reject the model if multiple testing is taken into account; for example, step-down procedures would terminate at the first step (that is, Bonferroni adjustment). Finally, we can also estimate the proportion of the population with a revealed preference for one year's prices over another.  We provide an illustration in  \reftab{table:SHS_Test}; notice that the confidence intervals are informative, with a width no greater than 0.15.

\section{Conclusion} \label{sec:conclusion}

We propose a revealed price preference relation that generates a nonparametric ranking of price vectors; a consistency (no-cycles) condition in this relation characterizes an augmented utility model in which consumers get utility from consumption and disutility from expenditure. This model is a natural generalization of quasilinearity and, furthermore, captures some prominent behavioral models of consumption. The model is also flexible enough to accommodate nonlinear prices, discrete choice and other consumption environments.  We develop the theoretical basis for welfare analysis in our model.

We generalize our model to a random utility context which is suitable for welfare analysis using repeated cross-sectional (as opposed to single-agent) data and show how to statistically test this random augmented utility model. A strength of this model is that it can be directly taken to household expenditure data in contrast to the standard random utility model which requires an additional round of estimation to account for the endogeneity of expenditure. We develop novel econometric theory to determine the proportion of consumers who are made better or worse off by a price change. This theory---which derives bounds on linear transforms of partially identified vectors---is a standalone contribution which has broader applications beyond those in this paper.

Finally, we operationalize both the deterministic and random versions of our model in separate applications to single-agent and repeated cross-sectional data. We confirm that our model is supported by data and can be used for meaningful welfare analysis.

%\begin{bibunit}[plain]
%\putbib[GAPP_Refs2]
%\end{bibunit}

%\end{document}

%\bibliographystyle{econometrica}
%\bibliography{GAPP_Refs}

\newpage

\appendix

\renewcommand\theequation{A.\arabic{equation}}
\renewcommand\thefigure{A.\arabic{figure}}
\renewcommand\thetable{A.\arabic{table}}
\renewcommand\theproposition{A.\arabic{proposition}}
\renewcommand\thelemma{A.\arabic{lemma}}
\renewcommand\theexample{A.\arabic{example}}

\setcounter{equation}{0}
\setcounter{figure}{0}
\setcounter{proposition}{0}
\setcounter{lemma}{0}
\setcounter{table}{0}

\renewcommand{\thesection}{A.\arabic{section}}
\renewcommand\thetheorem{A.\arabic{theorem}}
\numberwithin{theorem}{section}

\begin{center}
* * * * * * \\
ONLINE APPENDIX\\
* * * * * *
\end{center}\medskip

\section{GAPP and GARP}\label{sec:GAPPGARP}

In this section, we first state and explain Afriat's Theorem.  After that we cover a number of topics on GAPP and GARP and their relationship: augmented utility functions that lead to both properties holding in a data set (\refsec{sec:modGPGP});  demand predictions at out-of-sample prices under GAPP and under GARP (\refsec{sec:DDGPGP}); and on reconciling differing revealed preference relations under GAPP and GARP (\refsec{sec:welGPGP}).

\subsection{Afriat's Theorem}\label{sec:afriat}

Recall that, given a data set $\mathcal{D}=\{(p^t,x^t)\}_{t=1}^T$, a utility function $\widetilde{U}:\Re_+^L\to \Re$ is said to rationalize $\mathcal{D}$ if, for all $t\in T$, we have $\widetilde U(x^t)\geq \widetilde U(x)$ for all
$x\in\left\{x\in\Re_+^L:\; p^t\cdot x\leq p^t\cdot x^t\right\}$; in other words, $x^t$ is the bundle that maximizes $\widetilde U$ among all bundles that cost $p^t\cdot x^t$ or less. \refAfriat{thm:Afriat} characterizes those data sets that can be rationalized in this sense.   Below is the formal statement of \refAfriat{thm:Afriat} along with some remarks that relate this theorem to results in the paper.

\begin{afriattheorem*}[\citet{afriat1967}]\label{thm:Afriat}
	Given a data set $\mathcal{D}=\{(p^t,x^t)\}_{t=1}^T$, the following are equivalent:
	\begin{enumerate}
		\item $\mathcal{D}$ can be rationalized by a locally nonsatiated utility function.
		\item $\mathcal{D}$ satisfies GARP.
		\item $\mathcal{D}$ can be rationalized by a strictly increasing, continuous, and concave utility function.
	\end{enumerate}
\end{afriattheorem*}

\noindent {\sc Remark 1.}\, That (1) implies (2) is clear, given the definition of GARP (see \refsec{sec:Afrsection} in the main paper). The substantive part of \refAfriat{thm:Afriat} is the claim that (2) implies (3).  Standard proofs (see, for instance, \citet{fostel2004} or \citet{quah2014}) work by showing that a consequence of GARP is that there exist numbers $\phi^t$ and $\lambda ^t>0$ (for all $t\in T$) that solve the so-called {\em Afriat inequalities}
\begin{equation}\label{Afri}
\phi^{t'} \leq \phi^{t}+ \lambda^{t}\, p^{t} \cdot (x^{t'}-x^{t})\:\mbox{ for all ${t'}\neq t$.}
\end{equation}
Once this is established, it is straightforward to show that
\begin{equation}\label{AfU}
\widetilde U(x)=\min_{t\in T}\left\{ \phi^t+ \lambda^t\, p^t \cdot (x-x^t)\right\}
\end{equation}
rationalizes $\mathcal D$, with the utility of the observed consumption bundles satisfying $\widetilde U(x^t)=\phi^t$.  The function $\widetilde U$ is the lower envelope of a finite number of strictly increasing affine functions, and so it is strictly increasing, continuous, and concave. A remarkable feature of this theorem is that while GARP follows simply from local nonsatiation of the utility function, it is nonetheless sufficient to guarantee that $\mathcal D$ is rationalized by a utility function with significantly stronger properties. Our results \refthm{thm:GAPP} and \refthm{thm:GAPP_nonlinear} share this feature.

\medskip

\noindent {\sc Remark 2.}\,  To be precise, GARP guarantees that there is preference $\succsim$ (i.e., a complete, reflexive, and transitive binary relation) on ${\mathcal X}$ that extends the (potentially incomplete) revealed preference relations $\succeq^*_x$ and $\succ^*_x$ in the following sense: if $x^{t'}\succeq_x^* x^t$, then $x^{t'}\succsim x^t$ and if $x^{t'}\succ_x^* x^t$ then $x^{t'}\succ x^t$. One could then proceed to show that, for \textit{any} such preference $\succsim$, there is in turn a utility function $\widetilde U$ that rationalizes $\mathcal D$ and extends $\succsim$ (from $\mathcal{X}$ to $\Re^L_{+}$) in the sense that $\widetilde U(x^{t'})\geq (>) \widetilde U(x^t)$ if $x^{t'}\succsim (\succ) x^t$ (see \citet{quah2014}). This has implications on the inferences one could draw from the data.  If $x^{t'}\not\succeq_x^* x^t$ (or if $x^{t'}\succeq_x^* x^t$ but $x^{t'}\not\succ^*_x x^t$) then it is {\em always possible} to find a preference extending the revealed preference relations such that $x^t\succ x^{t'}$ (or $x^{t'}\sim  x^{t}$ respectively).\footnote{We use $x^{t'}\sim x^{t}$ to mean that $x^{t'}\succsim x^t$ and $x^{t}\succsim x^{t'}$.}   Therefore, $x^{t'}\succeq_x^* (\succ^*_x) x^t$ if and only if every locally nonsatiated utility function rationalizing $\mathcal D$ has the property that $\widetilde U(x^{t'})\geq (>) \widetilde U(x^t)$.

Similarly, we show in \refprop{prop:p_welfare} that the revealed price preference relation contains the most detailed information for welfare comparisons in our model.
\medskip

\noindent {\sc Remark 3.}\, A feature of \refAfriat{thm:Afriat} that is less often remarked upon is that in fact $\widetilde U$, as given by (\ref{AfU}), is well-defined, strictly increasing, continuous, and concave on the domain $\Re^L$, rather than just the positive orthant $\Re_+^L$.  Furthermore,
\begin{equation} \label{eqn:iU2}
x^t \in \argmax_{\left\{x\in\Re^L:\; p^t \cdot x\leq p^t\cdot  x^t\right\}} \widetilde{U}(x)\qquad \text{for all } t\in T.
\end{equation}
In other words, $\widetilde{U}$ can be extended beyond the positive orthant and $x^t$ remains optimal under $\widetilde{U}$ in the set
$\left\{x\in\Re^L:\; p^t \cdot x\leq p^t\cdot  x^t\right\}$. (Compare \eqref{eqn:iU2} with \eqref{eqn:oriU}.) We utilize this feature when we apply \refAfriat{thm:Afriat} in our proof of \refthm{thm:GAPP}.

\subsection{Models that satisfy both GAPP and GARP}\label{sec:modGPGP}

Suppose that a data $\mathcal{D}=\{(p^t,x^t)\}_{t+1}^T$ is collected from a consumer who is maximizing an augmented utility function of the form
\begin{equation}\label{formform}
U(x,-e)=h(\widetilde U(x),-e),
\end{equation}
where $h$ is strictly increasing (in both its arguments) and $\widetilde U:\Re^L_+\to\Re$ is strictly increasing.  In this case, obviously the data set obeys GAPP, but it must also obey GARP, because if $x^t$ maximizes $U$ then $x^t$ also maximizes $\widetilde U$ in the set $\{x\in\Re^L_+:p^t\cdot x\leq p^t\cdot x^t\}$.  Thus GAPP and GARP are not mutually exclusive properties and to say that a data set satisfies one is not to say that it violates the other; depending on the issue being studied, the analyst could exploit GAPP, or GARP, or perhaps even both in conjunction.

An interesting question worth investigating is the characterization of those data sets $\mathcal{D}$ generated by consumers who maximize an augmented utility function of the form (\ref{formform}).  Such a characterization must involve a property stronger than both GAPP and GARP; indeed, related work that characterizes rationalization by weakly separable preferences in the context of the constrained-maximization model (see \citet{quah2014}) suggests that rationalization by an augmented utility function of the form (\ref{formform}) will involve a property {\em strictly} stronger than the combination of GAPP and GARP.  A special case of (\ref{formform}) is, of course, the quasilinear form, where $U(x,-e)=\widetilde U(x)-e$.  In this case, a full characterization is known and the rationalizing property is sometimes referred to as the {\em strong law of demand} (see \citet{brown2007}); obviously the strong law of demand implies both GAPP and GARP.

In our analysis of the Progresa data reported in \refsec{sec:progresa}, we find that 2061 out of 2488 houesholds pass GAPP (83\%),  2375 households pass GARP (95\%), and 35 households (a bit more than 1\%) fail both tests.  Interestingly, 1983 households (80\%) pass both GAPP and GARP, which is suggestive (but not conclusive) evidence that a very large proportion of households from the Progresa data could be rationalized by an augmented utility function of the form (\ref{formform}).

\subsection{Comparing demand predictions under GAPP and GARP} \label{sec:DDGPGP}

Suppose a data set $\mathcal{D}=\{(p^t,x^t)\}_{t=1}^T$ obeys GARP.  Then we know from Afriat's Theorem that there is a utility function $\widetilde U:\Re^L_+\to\Re$ for which $x^t$ is constrained optimal, for all $t$.  What could this model tell us about the demand at some price $\hat p$ that is not among the observed prices?  In this model, the predicted demand also depends on the level of total expenditure on the observed goods.  Suppose the expenditure is required to be some $w>0$; then the predicted demand will be those bundles $x$ with $\hat p\cdot x=w$ that are compatible with the model when combined with $\mathcal{D}$.  By Afriat's Theorem, this is means that $x$ is a predicted demand if and only if the following conditions are satisfied: $\hat p\cdot x=w$ and the data set $\mathcal{D}\cup \{(\hat p,x)\}$ obeys GARP.

Now suppose that $\mathcal{D}=\{(p^t,x^t)\}_{t=1}^T$ also obeys GAPP. Then we know it is also compatible with the augmented utility model and we could ask what the augmented utility model would say about demand at the price $\hat p$.  This is equivalent to identifying bundles $x$ such that  $\mathcal{D}\cup \{(\hat p,x)\}$ obeys GAPP.  Since $\mathcal{D}\cup \{(\hat p,x)\}$ obeys GAPP if and only if $\mathcal{D}\cup \{(\hat p,\lambda x)\}$ obeys GAPP for any $\lambda>0$ (see \refsec{sec:GAPPvsGARP}), we know that {\em the set of predicted demands at $\hat p$ forms a cone}.

Not surprisingly, these two models will typically have different predictions, even at the same expenditure level $w>0$.  To illustrate this, consider the following example.

\begin{example} \label{ex:yat}
Suppose $\mathcal{D}$ consists of the single observation $p^1=(1,1)$ and $x^1=(1,1)$.  What is the predicted demand at $\hat p=(1/4,3/2)$?   We study the predictions under the constrained-optimization model, with and without imposing homotheticity on the utility function, and the augmented utility model.  \medskip

Consider first the constrained-optimization model.\,  (a)\, Suppose that $w<\hat p\cdot x^1=7/4$; the line of points/bundles incurring this level of expenditure is depicted by $B'$ in \reffig{fig:CF_GARP}. In this case, any bundle with $\hat p\cdot x=w$ will {\em not} be revealed preferred to $x^1$ and so $x$ can be any bundle in gray shaded area without violating GARP.\, (b)\, Now suppose $w\geq \hat p\cdot (0,2)= 3$; the bundles with $\hat p\cdot x=w$ is depicted as $B'''$ in \reffig{fig:CF_GARP}. Then if $x\cdot\hat p=w$, we have $x\cdot p^1>2$. In other words, $x^1$ will never be revealed preferred to $x$. Once again, $x$ can be any bundle in the red shaded area (that extends indefinitely towards the north east) without GARP being violated.\, (c)\, Lastly, we turn to the case where $w\in [7/4,3)$; a line with bundles satisfying this property is $B''$. Then any bundle satisfying $\hat p\cdot x=w$ will be revealed preferred to $x^1$. So GARP requires that $x^1$ is {\em not} revealed preferred to $x$, that is, $p^1\cdot x>p^1\cdot x^1=2$ and therefore, all bundles in the blue shaded area will not violate GARP.

The shaded area in \reffig{fig:CF_GARP} gives the predicted demands at $\hat p$ using GARP.\medskip

\begin{figure}[h]
	\centering
	\begin{subfigure}[b]{\linewidth}
		\centering
		\includegraphics[trim=2.2cm 10.3cm 2.3cm 9.6cm,clip=true,scale=.7,page=1]{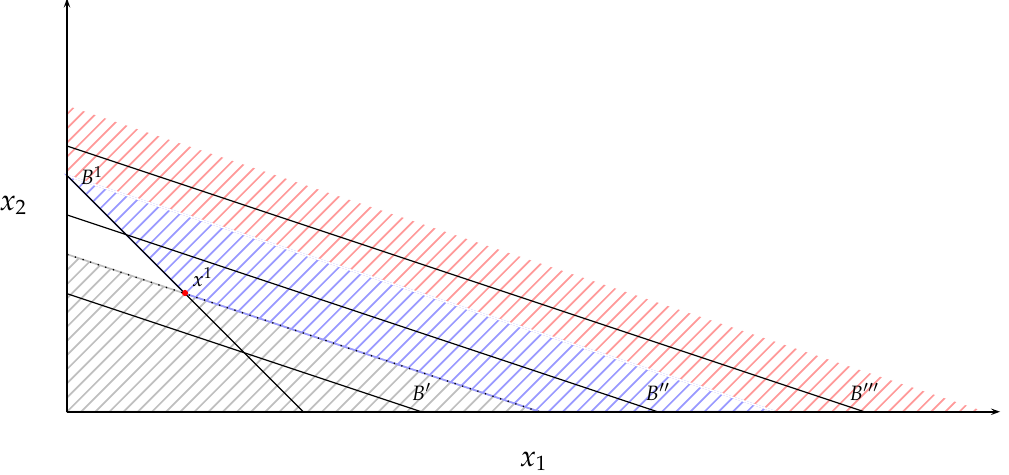}
		\caption{Counterfactuals using GARP}
		\label{fig:CF_GARP}
	\end{subfigure}\\
	\begin{subfigure}[b]{.49\linewidth}
		\includegraphics[trim=3.95cm 10.3cm 5.1cm 8.3cm,clip=true,scale=.7,page=2]{AppendixFigures.pdf}
		\caption{Counterfactuals using HARP}
		\label{fig:CF_HARP}
	\end{subfigure}
	\begin{subfigure}[b]{.49\linewidth}
		\includegraphics[trim=3.95cm 10.3cm 5.1cm 8.3cm,clip=true,scale=.7,page=3]{AppendixFigures.pdf}
		\caption{Counterfactuals using GAPP}	
		\label{fig:CF_GAPP}
	\end{subfigure}
	\caption{Counterfactuals with different consumption models}
	\label{fig:CF_HARP_GAPP}
\end{figure}

What happens to the predictions of the constrained-maximization model when the utility function is required to be homothetic?  It is well known that homothetic utility functions generate demand that is linear in cones. Therefore, for any $x\in\Re^2_+$, the data set $\{(p^1,x^1),(\hat p,x)\}$ can be rationalized (in the constrained-maximization sense) by a homothetic utility function if and only if $\{(p^1,x^1),(\hat p,\lambda x)\}$ can also be rationalized in this sense, for any $\lambda>0$. In other words, as in the augmented utility model, the set of predicted demands forms a cone.

The characterization of data sets that are constrained-optimal according to some homothetic preference is given in \citet{varian1983}, where the precise condition is known as the {\em homothetic axiom of revealed preference} or HARP, for short. In our simple case, to guarantee that  $\{(p^1,x^1),(\hat p,\lambda x)\}$ satisfies HARP, we set $w=\hat p\cdot x^1$ and consider the bundles with $\hat p\cdot x=w$; the bundles at this expenditure level are depicted by $\widetilde B$ in \reffig{fig:CF_HARP}.  At this expenditure level, GARP requires that $x$ satisfies $p^1\cdot x>p^1\cdot x^1$ and, for any such $x$, we have $\{(p^1,x^1),(\hat p,\lambda x)\}$ satisfying HARP; in other words, the set of predicted demands is the cone generated by these bundles of $x$.  This cone is depicted by the shaded region in \reffig{fig:CF_HARP}. \medskip

In the case of the augmented utility model, recall that if $x$ satisfies $\hat p\cdot x=p^1\cdot x^1=2$, then $\{(p^1,x^1),(\hat p,x)\}$ satisfies GAPP if and only if it satisfies GARP (see \refprop{prop:GAPPscaling}). The budget line with the property that $\hat p\cdot x=2$ is $\hat{B}$ in \reffig{fig:CF_GAPP} and, in this case, GARP (equivalently, GAPP) requires that $p^1\cdot x>p^1\cdot x^1=2$.  The shaded area gives the predicted demands at $\hat p$.  Notice that the cone in this case contains the cone in \reffig{fig:CF_HARP}, which is consistent with the fact that HARP is a stronger property than GAPP.  Furthermore, the predicted demands under GAPP is neither a subset nor a superset of that under GARP, which is again unsurprising given that these two properties are not comparable.
\end{example}

\subsection{Revealed preferences under GAPP and GARP} \label{sec:welGPGP}

Both GARP and GAPP forbids the existences of strict cycles over revealed preference relations: in the case of GARP, the revealed preference relation is defined over bundles and in  the case of GAPP it is defined over prices.  It is entirely possible for these revealed preference relations to disagree with each other; this occurrence should not be thought of as strange, nor is it an indication that one model is better of worse compared to the other. The two conclusions apply to different objects and either, or both, of them could be interesting to the analyst.

%$$U(x,-e)=h(\tilde U(x),-e),$$
%where $h$ is strictly increasing (in both its arguments) and $\tilde U:\Re^L_+\to\Re$ is strictly increasing.  In this case, obviously the data set obeys GAPP, but it also obeys GARP, because if $x^t$ maximizes $U$ then $x^t$ also maximizes $\tilde U(x)$ in the set $\{x\in\Re^L_+:p^t\cdot x\leq p^t\cdot x^t\}$.  For example, if $U$ has the quasilinear form, then $\mathcal{D}$ will obey both GARP and GAPP.

To be precise, suppose that a data $\mathcal{D}=\{(p^t,x^t)\}_{t=1}^T$ is collected from a consumer who is maximizing an augmented utility function of the form (\ref{formform}).   Such a data set will obey both GAPP and GARP.  It is possible for the price $p^t$ to be strictly revealed preferred to $p^s$ (whether directly or indirectly) and for the bundle $x^s$ to be revealed strictly preferred to $x^t$.  If this occurs, is the agent better off in observation $t$ or in observation $s$? The fact that $p^t$ is revealed strictly preferred to $p^s$ means that
$$U(x^t,-p^t\cdot x^t)> U(x^s,-p^s\cdot x^s)$$
while the fact that $x^s$ is revealed strictly preferred to $x^t$ means that
$$\widetilde U(x^t)<\widetilde U(x^s).$$
In other words, the consumer's augmented utility is higher in observation $t$ than in observation $s$, even though her sub-utility on the observed bundles is lower in observation $t$; these two phenomena are not mutually exclusive. \medskip

Another observation worth making is that it is sometimes possible to conclude that an out-of-sample price $\hat p$ is superior to some in-sample price $p^{t_1}$ observed in $\mathcal{D}$, even though one has no inkling what the demand will be at $\hat p$. Indeed, $\hat p$ is revealed preferred to $p^{t_1}$ whenever $\hat p\cdot x^{t_1}\leq p^{t_1}\cdot x^{t_1}$ (and, more generally, this relation between $\hat p$ and some other in-sample price $p^t$ can be extended via transitive closure). It follows that at the (unobserved) optimal bundle at $\hat p$, which we denote by $\hat x$, we must have
$$U(\hat x,-\hat p\cdot \hat x)\geq U(x^{t_1},-\hat p\cdot x^{t_1})>U(x^{t_1},-p^{t_1}\cdot x^{t_1}).$$
This is true even though, as we know from \refsec{sec:DDGPGP}, the predicted demand at $\hat p$ under the augmented utility model can be an arbitrarily small or large bundle.  On the other hand, without knowing the agent's expenditure level at $\hat p$, it is impossible to tell if the sub-utility $\widetilde U(\hat x)$ is greater or lower than $\widetilde U(x^{t_1})$. Put another way, while GAPP may allow the observer to rank $\hat p$ with $p^{t_1}$, it is impossible to rank the subutility of the demand bundle at these two observation using GARP, without some information or assumption on the expenditure level at $\hat p$.

\begin{example}
Suppose $\mathcal{D}$ consists of two observations,
\begin{eqnarray*}
(p^{t_1},x^{t_1})&=&((2,2),(2,2))\:\mbox{ and}\\
(p^{t_2},x^{t_2})&=&((1,1),(1,1)).
\end{eqnarray*}
It is straightforward to check that this data set can be generated by a consumer maximizing
$$U(x,-e)=\widetilde U(x)-f(e),$$
for strictly increasing functions $\widetilde U$ and $f$. Clearly, $p^{t_2}$ is revealed preferred to $p^{t_1}$ and $x^{t_1}$ is revealed preferred to $x^{t_2}$.  In this case, the consumer's augmented utility is higher at $t_2$ compared to $t_1$, even though her sub-utility on the observed goods is lower at $t_2$ compared to $t_1$.

Now suppose the data consists of just the observation $(p^{t_1},x^{t_1})$.  Obviously, we can still conclude that the consumer prefers $\hat p=(1,1)$ to $p^{t_1}$ and derives greater augmented utility from $\hat p$ than from $p^{t_1}$. However, nothing can be said about the consumer's subutility without further information on expenditure.   If the expenditure is lower than the expenditure at $t_1$, which is 8, then the subutility achieved at $\hat p$ must be lower than the subutility of $x^1$ and if the expenditure is higher than 8, then the sub-utility achieve must be lower than that of $x^{t_1}$.
\end{example}

\section{Proof of \refprop{prop:p_welfare}}   \label{sec:pricewelfare}
	
\noindent (1)  We have already shown the `only if' part of this claim, so we need to show the `if' part holds. From the proof of \refthm{thm:GAPP}, we know that for a large $M$, it is the case that $p^{t}\succeq_p p^{t'}$ if and only if $(x^t,M-p^t\cdot x^t) \succeq_x (x^{t'},M-p^{t'}\cdot x^{t'})$ and hence $p^{t}\succeq_p^* p^{t'}$ if and only if $(x^t,M-p^t\cdot x^t) \succeq_x^*  (x^{t'},M-p^{t'}\cdot x^{t'})$.  If $p^{t}\not\succeq_p^* p^{t'}$, then $(x^t,M-p^t\cdot x^t) \not\succeq_x^*  (x^{t'},M-p^{t'}\cdot x^{t'})$ and hence there is a utility function $\widetilde U:\Re_{+}^{L+1}\rightarrow \Re$ rationalizing the augmented data set $\widetilde D$ such that $\widetilde U(x^t,M-p^t\cdot x^t)<\widetilde U(x^{t'},M-p^{t'}\cdot x^{t'})$ (see Remark 2 in \refsec{sec:afriat}).  This in turn implies that the augmented utility function $U$ (as defined by (\ref{eq:UoX})), has the property that $U(x^t,-p^t\cdot x^t)<U(x^{t'},-p^{t'}\cdot x^{t'})$ or, equivalently, $V(p^t)<V(p^{t'})$.\medskip
	
\noindent (2)  Given part (1), we need only show that if $p^{t}\succeq_p^* p^{t'}$ but $p^{t}\not\succ_p^* p^{t'}$, then there is some augmented utility function $U$ such that $U(x^t,-p^t\cdot x^t)=U(x^{t'},-p^{t'}\cdot x^{t'})$.  To see that this holds, note that if $p^{t}\succeq_p^* p^{t'}$ but $p^{t}\not\succ_p^* p^{t'}$, then $(x^t,M-p^t\cdot x^t) \succeq^*_x  (x^{t'},M-p^{t'}\cdot x^{t'})$ but $(x^t,M-p^t\cdot x^t) \not\succ^*_x  (x^{t'},M-p^{t'}\cdot x^{t'})$. In this case there is a utility function $\widetilde U:\Re_{+}^{L+1}\rightarrow \Re$ rationalizing the augmented data set $\widetilde D$ such that $\widetilde U(x^t,M-p^t\cdot x^t)=\widetilde U(x^{t'},M-p^{t'}\cdot x^{t'})$.  This in turn implies that the augmented utility function $U$ (as defined by (\ref{eq:UoX})) satisfies $U(x^t,-p^t\cdot x^t)=U(x^{t'},-p^{t'}\cdot x^{t'})$ and so $V(p^t)=V(p^{t'})$. \hfill $\blacksquare$

\section{Price indices to deflate nominal expenditure}   \label{sec:indices}

In this section, we build on the discussion in \refsec{sec:deflprices}.  Suppose that, at observation $t$, the consumer chooses $(x^t,y^t)$ to maximize $\widetilde{U}(x,y)$, subject to $p^t \cdot x^t+q^t\cdot y^t\leq M^t$. We are interested in the conditions under which there is an index $k^t$, depending on the prices of the outside goods, such that the deflated data $\{(p^t/k^t,x^t)\}_{t=1}^T$ obeys GAPP (and hence can be rationalized as maximizing an augmented utility function).   In the main paper, we explained that this holds if prices of the outside goods move up and down proportionately (so there is no change to their prices relative to each other). When relative prices {\em are} allowed to change, it is still possible to obtain a deflator $k^t$ guaranteeing that $\{(p^t/k^t,x^t)\}_{t=1}^T$ obeys GAPP, but stronger assumptions will have to be imposed on the utility function $\widetilde{U}$.  We outline a set of sufficient conditions for this to hold.

Suppose that the outside goods are weakly separable from the observed goods, so the overall utility function has the form $\widetilde{U}(x,\tilde{u}(y))$, where $\tilde{u}(y)$ is the sub-utility of the bundle $y$ of outside goods.  Furthermore, we assume that $\tilde{u}$ has an indirect utility $\tilde{v}$ of the following form:
$$\tilde{v}(q,m)=h\left(\frac{m}{f(q)}+b(q),g_1(q),g_2(q),\ldots,g_N(q)\right)$$
where $f$, $b$, $g_1,\ldots,g_N$ are all real-valued functions of the prices $q$ of the outside goods, and $m$ is the expenditure devoted to those goods. This formulation covers a number of standard functional forms used in empirical analysis.  If $\tilde{v}(q,m)=m/f(q)$ where $f$ is one-homogeneous then the preference it generates is homothetic; if $\tilde{v}(q,m)=(m/f(q))+b(q)$, where $b$ is zero-homogeneous, then we obtain the Gorman polar form (see \citet{gorman1961}).  Another example is the form
\begin{equation} \label{weird-v}
\ln \tilde{v}(q,m)=\left\{\left[\frac{\ln m-\ln f(q)}{g_1(q)}\right]^{-1}+g_2(q)\right\}^{-1}
\end{equation}
where $g_1$ and $g_2$ are zero-homogeneous functions.  If $g_2\equiv 0$, the form (\ref{weird-v}) generates the Price Invariant Generalized Logarithmic (PIGLOG) demand system \cite{muellbauer1976}; if further functional form restrictions are imposed on $f$ and $g_1$, we obtain the Almost Ideal Demand System (AIDS) of \citet{deaton1980almost}.  The Quadratic Almost Ideal Demand System (QUAIDS) is a generalization of AIDS that has greater flexibility to model empirically relevant Engel curves (see \citet{banks1997}); it is a special case of (\ref{weird-v}) with functional form restrictions on $f$, $g_1$, and $g_2$.

We assume that the consumer's total wealth $M^t$ varies with $t$ in such a way that, should the consumer devote all of this wealth to the unobserved goods, then her utility is constant. This captures the idea that the consumer's {\em real wealth} (as measured by the indirect utility function $v$) is unchanged across  observations.  While we permit prices of the unobserved goods to change, we require that they change in such a way that  $g_1(q^t)$, $g_2(q^t),\ldots,g_N(q^t)$ remain constant at $\bar g_1,\bar g_2\ldots,\bar g_N$ (respectively) for all $t$.  Given the form of $\tilde{v}$, this implies that $(M^t/f(q^t))+b(q^t)$ is constant for all $t$; let this constant be $C$.  Thus we can think of the consumer as choosing $(x,c)$ to maximize $\widetilde{U}(x,\tilde{v}(c,\bar g_1,\bar g_2,\ldots,\bar g_N))$ subject to $p^t\cdot  x+(c-b(q^t))f(q^t)\leq (C-b(q^t))f(q^t)$.  This inequality can be written as
$$\frac{p^t\cdot x}{f(q^t)}+c\leq C.$$
It follows that the data set $\{(p^t/f(q^t), x^t)\}_{t=1}^T$ will obey GAPP.\vspace{0.1in}

\section{Nonlinear pricing, the rationality index, and related topics}\label{sec:omnibus}

In this section, we formulate and prove a rationalization result that allows for both imperfect rationalization and nonlinear pricing. This result generalizes \refthm{thm:GAPP_nonlinear} and \refthm{thm:GAPP} by allowing for imperfect rationality.  We explain how this result is crucial in helping us to calculate the rationality index (introduced in \refsec{sec:index}) and other variations on that index that provide a measure of departures from exact rationality.  We also use this result to show that the bounds on the compensating and equivalent variations obtained in \refsec{sec:compensation} are tight.

\subsection{$\overline{\vartheta}$-rationalization}\label{sec:ultimate}

We are in the setting of \refsec{sec:nonlinearGAPP}.  The consumer chooses her consumption from the space $X\subseteq \Re_+^L$.  A price system is a map $\psi: X\to\mathbb{R}_+$, where $\psi (x)$ is the cost of purchasing $x\in X$. Let $\overline{\vartheta}=(\vartheta^1,\vartheta^2,\dots,\vartheta^T)\in (0,1]^T$.  An augmented utility function $U:X\times\mathbb{R}_+\to\mathbb{R}$ provides a  $\overline{\vartheta}$-rationalization of a data set $\mathcal{D}=\{(\psi^t,x^t)\}_{t=1}^T$ if, at each observation $t$,
$$U(x^t,-\psi^t(x^t))\geq U(x,-(\vartheta^{t})^{-1}\psi^t(x))\:\mbox{ for all $x\in X$.}$$
Note that this definition of $\overline{\vartheta}$-rationalization generalizes the notion introduced in \refsec{sec:index}, which can be thought of as the special case where $\vartheta^t=\vartheta^{t'}$ for all $t$, $t'\in T$. The context here is also more general since we allow for nonlinear pricing (as introduced in \refsec{sec:nonlinearGAPP}).  Obviously, if a data set can be exactly rationalized then it is $\overline{\vartheta}$-rationalized with $\overline{\vartheta} =(1,1,\ldots,1)$; note also that if a data set can be $\overline{\vartheta}$-rationalized then it can also be $\overline{\vartheta} '$-rationalized for $\overline{\vartheta}'<\overline{\vartheta}$.   A consumer whose observations cannot be exactly rationalized but can be $\overline{\vartheta}$-rationalized for some $\overline{\vartheta}<(1,1,\ldots, 1)$ exhibits limited rationality in the sense discussed in \refsec{sec:index}.

The calculation of the rationality index hinges on our ability to ascertain whether a data set $\mathcal{D}$ has a $\overline{\vartheta}$-rationalization for a given $\overline{\vartheta}$. It is possible to characterize those data sets that can be $\overline{\vartheta}$-rationalized using a modified version of the GAPP test, as we now explain.

Let $\overline{\vartheta}\in (0,1]^T$. Define the relations $\succeq_{p,\overline{\vartheta}}$ and $\succ_{p,\overline{\vartheta}}$ in the following way:
\begin{quote}
$\psi^{t'}\succeq_{p,\overline{\vartheta}} \psi^{t}$ if $\psi^{t'}(x^t)\leq \vartheta^{t'} \psi^t(x^t)$ and $\psi^{t'}\succ_{p,\overline{\vartheta}} \psi^{t}$ if $\psi^{t'}(x^t)< \vartheta^{t'} \psi^t(x^t)$.
\end{quote}
Denote the transitive closure of $\succeq_{p,\overline{\vartheta}}$ by $\succeq^*_{p,\overline{\vartheta}}$.  Obviously these definitions generalize the ones given for revealed preference relations over prices provided in \refsec{sec:nonlinearGAPP}.

\medskip

The data set $\mathcal{D}$ obeys {\em $\overline{\vartheta}$-GAPP} if
\begin{quote}
{\em there do not exist observations $t,t'\in T$ such that $\psi^{t'}\succeq^*_{p,\overline{\vartheta}} \psi^{t}$ and $\psi^{t}\succ_{p,\overline{\vartheta}} \psi^{t'}$.}
\end{quote}
The next result states that $\overline{\vartheta}$-GAPP characterizes $\overline{\vartheta}$-rationalization.

\begin{theorem} \label{thm:ultimate}
A data set $\mathcal{D}=\{(\psi^t,x^t)\}_{t=1}^T$ can be $\overline{\vartheta}$-rationalized by an augmented utility function for some $\overline{\vartheta}\in (0,1]^T$ if and only if it satisfies $\overline{\vartheta}$-GAPP.
\end{theorem}\medskip

\noindent {\sc Remark 1.}\,  This theorem states, in particular, that $\mathcal{D}=\{(\psi^t,x^t)\}_{t=1}^T$ can be rationalized by an augmented utility function if and only if it satisfies GAPP, which corresponds to the special case where $\overline{\vartheta}=(1,1,\ldots,1)$. So it covers the first claim in \refthm{thm:GAPP_nonlinear} (the part before ``Furthermore,\ldots") and also the equivalence of statements (1) and (2) in \refthm{thm:GAPP}. For the proof of the second claim in \refthm{thm:GAPP_nonlinear} see the end of this subsection.  Unlike the proof we gave of \refthm{thm:GAPP} in the main body of the paper, our proof of \refthm{thm:ultimate} does not appeal to \refAfriat{thm:Afriat}, though it is clearly inspired by it. In particular, we show that $\overline{\vartheta}$-GAPP implies that there is a solution to a system of linear inequalities (see \reflem{lemma-Af} below), analogous to the so-called Afriat inequalities usually derived in the proof of \refAfriat{thm:Afriat} and then use those inequalities to explicitly construct a piecewise linear augmented utility function that rationalizes the data.
\medskip

\noindent {\sc Remark 2.}\,  Note that checking whether or not $\overline{\vartheta}$-GAPP holds for a given $\overline{\vartheta}$ is computationally undemanding: the relations  $\succeq_{p,\overline{\vartheta}}$ and $\succ_{p,\overline{\vartheta}}$ can be easily constructed; once this has been obtained, we can apply Warshall's algorithm to compute the transitive closure of the revealed preference relations and then check for violations of $\overline{\vartheta}$-GAPP.\medskip

\noindent {\sc Remark 3.}\,  Suppose we impose the mild restriction that every bundle that is an observed choice has a strictly positive value under any of the other price observation, that is, $\psi^{t'}(x^{t})>0$ whenever $\psi^{t'}\neq \psi^{t}$.  Then we can choose sufficiently small $\vartheta>0$ so that $\psi^{t'}(x^t)> \vartheta \psi^t(x^t)$ whenever $\psi^{t'}\neq \psi^{t}$.  If we let $\overline{\vartheta}=(\vartheta,\ldots,\vartheta)$, then $\mathcal{D}$ must obey $\overline{\vartheta}$-GAPP simply because the relation $\succ_{p,\overline{\vartheta}}$ is empty.  Thus every data set is $\overline{\vartheta}$-rationalizable for $\overline{\vartheta}$ sufficiently close to zero.   \bigskip

\noindent {\bf Proof of \refthm{thm:ultimate}.}\, Suppose $\mathcal{D}$ can be $\overline{\vartheta}$-rationalized by an augmented utility function for some $\overline{\vartheta}\in (0,1]^T$.  In that case, if $\psi^{t'}\succeq_{p,\overline{\vartheta}} \psi^{t}$, then $\psi^{t'}(x^t)\leq \vartheta^{t'} \psi^t(x^t)$ and so
\begin{equation}\label{phew}
U(x^{t'},-\psi^{t'}(x^{t'}))\geq U(x^{t},-(\vartheta^{t'})^{-1}\psi^{t'}(x^{t}))\geq U(x^{t},-\psi^{t}(x^t)),
\end{equation}
where the first inequality follows from the (imperfect) optimality of $x^{t'}$ and the second from the property that $U$ is strictly decreasing in expenditure. It follows that if $\psi^{t'}\succeq^*_{p,\overline{\vartheta}} \psi^{t}$, then $U(x^{t'},-\psi^{t'}(x^{t'}))\geq U(x^{t},-\psi^{t}(x^t))$.  Similarly, if $\psi^{t'}\succ_{p,\overline{\vartheta}} \psi^{t}$, then $\psi^{t'}(x^t)<\vartheta^{t'} \psi^t(x^t)$ and we obtain $U(x^{t'},-\psi^{t'}(x^{t'}))>U(x^{t},-\psi^{t}(x^t))$ since the second inequality in (\ref{phew}) will now be strict.  It is then clear that we cannot simultaneously have $\psi^{t'}\succeq^*_{p,\overline{\vartheta}} \psi^{t}$, and $\psi^{t}\succ_{p,\overline{\vartheta}} \psi^{t'}$, which establishes $\overline{\vartheta}$-GAPP.

Conversely, suppose that $\mathcal{D}$ obeys $\overline{\vartheta}$-GAPP.  Then there is a complete preorder $\succsim$ defined on the set $\{p^t\}_{t\in T}$ that extends $\succeq_{p,\overline{\vartheta}}$ and $\succ_{p,\overline{\vartheta}}$ in the sense that such $\psi^{t'}\succsim \psi^t$ if $\psi^{t'}\succeq^*_{p,\overline{\vartheta}} \psi^{t}$ and $\psi^{t'}\succ \psi^t$ if $\psi^{t'}\succ_{p,\overline{\vartheta}} \psi^{t}$, where $\succ$ is the asymmetric part of $\succsim$.   We first prove the following lemma.

\begin{lemma} \label{lemma-Af}
Suppose $\mathcal{D}$ obeys $\overline{\vartheta}$-GAPP and let $\succsim$ be a complete preorder that extends $\succeq_{p,\overline{\vartheta}}$ and $\succ_{p,\overline{\vartheta}}$.  Then there are numbers $\phi^t$ and $\lambda^t>0$ (for $t=1,2,\ldots, T$) with the following properties:
\begin{itemize}
\item[(a)] $\phi^{t'}>\phi^t$ if $\psi^{t'}\succ \psi^t$;
\item[(b)] $\phi^{t'}=\phi^t$ if $\psi^{t'}\sim \psi^t$; and
\item[(c)] $\phi^{t'}\leq \phi^{t}+\lambda^{t}(\psi^{t}(x^{t'})-\vartheta^t\psi^{t'}(x^{t'}))$ for all $t\neq t'$.
\end{itemize}
\end{lemma}\medskip

\noindent {\bf Proof.}\,  Let $z^{ij}=\psi^{i}(x^j)-\vartheta^i\psi^j(x^j)$ for $i,\, j\in T$.  Note that, for $i\neq j$, $z^{ij}<0$ implies that $\psi^i\succ \psi^j$ and $z^{ij}\leq 0$ implies that $\psi^i\succsim \psi^j$. We shall explicitly construct $\phi^t$ and $\lambda^t>0$ that satisfy the required conditions. With no loss of generality, suppose that $\psi^{t+1}\succsim \psi^t$ for $t=1,2,\ldots,T-1$.

First, choose $\phi^1$ to be any number and $\lambda^1$ to be any strictly positive number. Suppose $\psi^2\succ \psi^1$.  Then $\min_{j>1}z^{1j}>0$,   because if $z^{1j'}=\psi^1(x^{j'})-\vartheta^1\psi^{j'}(x^{j'})\leq 0$ for some $j'>1$, then $\psi^1\succsim \psi^{j'}$, which is a contradiction.   So there is $\phi^2$ such that
\begin{equation}\label{new-1}
\phi^1<\phi^2<\min_{j>1}\{\phi^1+\lambda^1z^{1j}\}.
\end{equation}
If $\psi^2\sim \psi^1$ then  $\min_{j>1}z^{1j}\geq 0$ because if $z^{1j'}=\psi^1(x^{j'})-\vartheta\psi^{j'}(x^{j'})<0$ for some $j'>1$, then $\psi^1\succ \psi^{j'}$, which is a contradiction.   Setting $\phi^2=\phi^1$, we obtain
\begin{equation}\label{new-2}
\phi^1=\phi^2 \leq \min_{j>1}\{\phi^1+\lambda^1z^{1j}\}.
\end{equation}
We claim that there is $\lambda^2>0$ such that
$$\phi^1\leq \phi^2+\lambda^2z^{21}.$$
Clearly this inequality holds if $z^{21}\geq 0$.  If $z^{21}= \psi^2(x^1)-\vartheta^2\psi^1(x^1)<0$, then $\psi^2\succ \psi^1$; this implies that  $\phi^1<\phi^2$ and thus the inequality holds for $\lambda^2$ sufficiently small.

We now go on to choose $\phi^3$ and $\lambda^3$.  Since $\psi^j\succsim \psi^i$ for all $j>2$ and $i=1,2$, we obtain $z^{ij}\geq 0$.  Consider two cases: when $\psi^3\succ \psi^2\succsim \psi^1$ and $\psi^3\sim \psi^2\succsim \psi^1$.  In the former case, both $\min_{j>2}z^{1j}>0$ and $\min_{j>2}z^{2j}>0$. Therefore
$$\phi^2<\min_{j>2}\{\phi^2+\lambda^2z^{2j}\}.$$
If $\phi^2=\phi^1$, obviously we also have
$$\phi^2<\min_{j>2}\{\phi^1+\lambda^1z^{1j}\}\mbox{;}$$
this inequality also holds if $\phi^2>\phi^1$ since in that case (\ref{new-1}) holds.  It follows that we can find $\phi^3$ such that
$$\phi^2<\phi^3<\min\left\{\min_{j>2}\{\phi^1+\lambda^1 z^{1j}\},\min_{j>2}\{\phi^2+\lambda^2z^{2j}\}\right\}.$$

We turn to the case where $\psi^3\sim \psi^2\succsim \psi^1$.  It follows from (\ref{new-1}) and (\ref{new-2}) that $\phi^2\leq \min_{j>2}\{\phi^1+\lambda^1z^{2j}\}$.  We also know that $z^{2j}\geq 0$ for all $j>2$.  Therefore, we can choose $\phi^3$ such that
$$\phi^2=\phi^3\leq \min\left\{\min_{j>2}\{\phi^1+\lambda^1 z^{1j}\},\min_{j>2}\{\phi^2+\lambda^2z^{2j}\}\right\}.$$
Now choose $\lambda^3>0$ sufficiently small so that
$$\phi^i\leq \phi^3+\lambda^3z^{3i}\:\mbox{ for $i=1,2$.}$$
Clearly that this inequality holds for any $\lambda^3>0$ if $z^{3i}\geq 0$. If $z^{3i}<0$ then $\psi^3\succ \psi^i$, in which case $\phi^3>\phi^i$ and the inequality will be satisfied for $\lambda^3$ sufficiently small.

Repeating this argument, we choose $\phi^t$ (for $t\leq T-1$) such that if $\psi^t\succ \psi^{t-1}$ then
\begin{equation}\label{new-3}
\phi^{t-1}<\phi^t<\min_{s\leq t-1}\left\{\min_{j>t-1}\{\phi^s+\lambda^sz^{sj}\}\right\}
\end{equation}
and if $\psi^t\sim \psi^{t-1}$ then
\begin{equation}\label{new-4}
\phi^{t-1}=\phi^t\leq \min_{s\leq t-1}\left\{\min_{j> t-1}\{\phi^s+\lambda ^s z^{sj}\}\right\}\mbox{;}
\end{equation}
and $\lambda^t>0$ (for $t=2,3,\ldots,T$) such that
\begin{equation}\label{new-5}
\phi^i\leq \phi^t+\lambda^t z^{ti}\:\mbox{ for $i\leq t-1$.}
\end{equation}
For a fixed $t'$, (\ref{new-3}) and (\ref{new-4}) guarantee that $\phi^{t'}\leq \phi^t+\lambda^t z^{tt'}$ for $t<t'$ while (\ref{new-5}) guarantees that this inequality holds for $t>t'$.  So we have found $\lambda^t$ and $\phi^t$ to obey condition (c), while the first two conditions hold by construction.  \hfill $\blacksquare$\medskip

We now return to the proof that (2) implies (3). Let $\succsim$ be a complete preorder that extends $\succeq_{p,\overline{\vartheta}}$ and $\succ_{p,\overline{\vartheta}}$ and let the numbers $\phi^t$ and $\lambda^t>0$ (for $t=1,2,\ldots, T$) satisfy properties (a) -- (c) in \reflem{lemma-Af}.   Define the function $U:X\times\Re_{-} \to\Re$ by
\begin{equation}\label{newyear}
U(x,-e)=\min_{t\in T}\{\phi^t+\lambda^t(\psi^t(x)-\vartheta^t e)\}.
\end{equation}
This function is an augmented utility function since it is strictly increasing in the last argument. We claim that this function also satisfies the property that, at each $t\in T$,
$$U(x^t,-\psi^t(x^t))\geq U(x,-(\vartheta^t)^{-1}\psi^t(x))\:\mbox{ for all $x\in X$.}$$
Indeed, at a given observation $s$, for any $t\neq s$, we have $\phi^t+\lambda^t(\psi^t(x^s)-\vartheta^t \psi^s(x^s))\geq \phi^s$ by condition (c); furthermore, $\phi^s+\lambda^s(\psi^s(x^s)-\vartheta^s \psi^s(x^s))\geq \phi^s$ since $\lambda^s>0$ and $\vartheta^s\in (0,1]$.  Therefore, $U(x^s,-\psi^s(x^s))\geq \phi^s$. On the other hand, by the definition of $U$,
$$U(x,-(\vartheta^s)^{-1}\psi^s(x))\leq \phi^s+\lambda^s(\psi^s(x)-\psi^s(x))\leq \phi^s.$$
So $U(x^s,-\psi^s(x^s))\geq U(x,-\vartheta^{-1}\psi^s (x))$ for all $x$. \hfill $\blacksquare$
\medskip

The augmented utility function $U$ at the price system $\psi$ induces an indirect utility given by $V(\psi)=\max_{x\in X}U(x,-\psi(x))$.  In the case where GAPP holds and exact rationalization is possible, one could also choose the rationalizing utility function $U$ so that its indirect utility $V$ agrees with any ordering over $\{\psi^t\}_{t=1}^T$ that is consistent with the revealed preference relations. (Note that this feature is also present in Afriat's Theorem; see Remark 2 in \refsec{sec:afriat}.)  The following result is used in \refsec{sec:morecomvar}.

\begin{theorem} \label{thm:ultimate2}
Suppose the data set $\mathcal{D}=\{(\psi^t,x^t)\}_{t=1}^T$ obeys GAPP and let $\succsim$ be a complete preorder on $\{\psi^t\}_{t=1}^T$ that extends $\succeq_{p}$ and $\succ_{p}$.  Then there is an augmented utility function $U:X\times\Re_{-}\to\Re$ that rationalizes $\mathcal{D}$ such that $V(\psi^{t'})=V(\psi^{t})$ if $\psi^{t'}\sim \psi^{t}$ and $V(\psi^t)>V(\psi^t)$ if $\psi^{t'}\succ \psi^t$ (where $\sim$ and $\succ$ are the symmetric and asymmetric parts of $\succsim$).
\end{theorem}\medskip

\noindent {\bf Proof.}\,  From the proof of \refthm{thm:ultimate}, we know that $U(x,-e)$ as given by (\ref{newyear}) (with $\theta^t=1$ for all $t$) rationalizes $\mathcal{D}$. We can then conclude that $V(\psi^t)=U(x^t,-\psi(x^t))=\phi^t$ because $\phi^t\leq \phi^{t'}+\lambda^{t'}(\psi^{t'}(x^t)-\psi^{t}(x^t))$ from part (c) of \reflem{lemma-Af}. Finally, $V$ satisfies the required properties because of (a) and (b) in \reflem{lemma-Af}.\hfill $\blacksquare$
\medskip

We end this subsection with the proof of \refthm{thm:GAPP_nonlinear}; this result is obtained as a corollary of \refthm{thm:ultimate}.
\medskip

\noindent {\bf Proof of \refthm{thm:GAPP_nonlinear}.}\,  Choosing  $\overline{\vartheta}=(1,1,\ldots,1)$, \refthm{thm:ultimate} states, in particular, that $\mathcal{D}=\{(\psi^t,x^t)\}_{t=1}^T$ can be rationalized by an augmented utility function if and only if it satisfies GAPP.  It remains for us to show that, under assumptions (i), (ii), and (iii), this utility function could be extended to one defined on a closed set $Y$ containing $X$ and that is increasing in $x_K$.  We know from the proof of \refthm{thm:GAPP_nonlinear} that the function $U:X\to\Re$ given by
\begin{equation*}
U(x,-e)=\min_{t\in T}\{\phi^t+\lambda^t(\psi^t(x)-e)\}.
\end{equation*}
rationalizes the data (see \ref{newyear}).  It suffices to show that each function $\psi^t$, which is defined on $X$ could be extended to a  continuous function on $Y$ that is strictly increasing in $x_K$, in which case we could correspondingly extend $U$ and the extension would be continuous and strictly increasing in $x_K$ (since $\lambda^T>0$).

That $\psi^t$ admits such an extension is guaranteed by (i), (ii), and (iii). A quick way of arriving at this conclusion is to appeal to Levin's Theorem, which is a version of Szpilrajn's Theorem for closed preorders (see \citet{nishimura2017} for a proof of Levin's Theorem).  Since $\psi^t$ is continuous, it induces a closed preorder $\succsim'$ on $X$ and therefore also on $Y$.\footnote{A preorder $\succsim'$ defined on a set $X$ is closed if $\{(a,b)\in X\times X: a\succsim' b\}$ is a closed subset of $X\times X$.}  For $K\subset L$, let $\geq_K$ be the partial order on $Y$ such that, for $x'$ and $x$ in $\Re^L$, we have $x'\geq_K x$ if $x'_i\geq x_i$ for all $i\in K$ and $x_i'=x_i$ for $i\notin K$. It is straightforward to check that, for any number $M$, the set
$$\{x\in Y:\mbox{there is $\tilde x\in X$ with $\tilde x\geq_K x$ and $M\geq \psi^t(\tilde x)$}\}$$
is a compact set in $Y$. (Recall that $Y$ is closed, contains $X$, and is contained in $\Re^L_+$.) Using this property, one could check that $\succsim''$, defined as the transitive closure of $\succsim'$ and $\geq_K$, is also a closed prorder on $Y$.  Levin's Theorem then guarantees that there is a {\em complete} and closed preorder $\succsim$ on $Y$ that extends $\succsim''$ and has a continuous representation $V:Y\to\Re$. In particular, $V$ must be strictly increasing in $x_K$ and satisfies the following property:  $V(x')\geq (>)\, V(x)$ if $\psi^t(x')\geq (>)\, \psi^t(x)$, for $x'$, $x\in X$. Furthermore, our assumptions guarantee that that $\{V(x):x\in X\}\subseteq\Re$ is a closed set.  These properties guarantee that we could choose a strictly increasing transformation $h$ defined on the range of $V$, i.e., the set $\{V(x):x\in Y\}$, so that $h(V(x))=\psi^t(x)$ for all $x\in X$.  Therefore the function $h\circ V:Y\to\Re$ is a continuous extension of $\psi^t:X\to\Re$ that is strictly increasing in $x_K$.  \hfill $\blacksquare$

\subsection{Rationality indices and their computation}\label{sec:rat-indices}

Given a data set $\mathcal{D}=\{(\psi^t,x^t)\}_{t=1}^T$, we know that it admits a $(\vartheta,\vartheta,\ldots,\vartheta)$-rationalization for some $\vartheta>0$ (see Remark 3 following \refthm{thm:ultimate}).  This guarantees that the {\em rationality index}, given by
\begin{equation*}
\vartheta^*=\sup\{\vartheta\in (0,1]:\mbox{$\mathcal{D}$ has a $(\vartheta,\vartheta,\ldots,\vartheta)$-rationalization}\},
\end{equation*}
is well-defined.  Note that this definition generalizes the definition provided in \refsec{sec:index} of the main paper, which applies to the linear price environment. A data set that can be rationalized exactly has a rationality index of 1 and we could use the closeness of $\vartheta^*$ to 1 as a measure of the data set's closeness to exactly rationality.

Given the characterization of $\overline{\vartheta}$-rationality stated in \refthm{thm:ultimate}, we also have
\begin{equation} \label{RI-GAPP}
\vartheta^*=\sup\{\vartheta\in (0,1]:\mbox{$\mathcal{D}$ satisfies $(\vartheta,\vartheta,\ldots,\vartheta)$-GAPP}\}.
\end{equation}
This identity provides us with a practical way of calculating $\vartheta^*$.  Indeed, $\vartheta^*$ can be obtained through a binary search algorithm that works as follows. We first set the lower and upper bounds on $\vartheta^*$ to be $\vartheta^L=0$ and $\vartheta^H=1$.  We then check (by checking $\overline{\vartheta}$-GAPP) whether the data set passes or fails the test at $\vartheta=(\vartheta^L+\vartheta^H)/2$ (to be precise, at $\overline{\vartheta}=(\vartheta,\vartheta,\ldots,\vartheta)$); if it passes the test, then we update both $\vartheta^*$ and its lower bound to $(\vartheta^L+\vartheta^H)/2$; if it fails the test, then we update $\vartheta^*$ to $\vartheta^L$ and the upper bound on $\vartheta^*$ to $(\vartheta^L+\vartheta^H)/2$. We then repeat the procedure, selecting and testing the new midpoint of the updated lower and upper bounds. The algorithm terminates when the lower and upper bounds are sufficiently close.

There are other plausible variations on the rationality index, based on the way one aggregates $\vartheta^t$ across observations.  Let $F:(0,1]^T\to \mathbb{R}_+$ be any weakly increasing function taking nonnegative values such that $F(1,1,\ldots,1)=1$.  We can then construct a generalized rationality index
$$F^*=\sup\{F(\overline{\vartheta}):\mbox{$\mathcal{D}$ has a $\overline{\vartheta}$-rationalization}\}.$$
The rationality index $\vartheta^*$ corresponds to the case where $F$ is defined by
$$F(\overline{\vartheta})=\min\{\vartheta^1,\vartheta^2,\ldots,\vartheta^T\}.$$
As an alternative to this, one could choose
$$F(\overline{\vartheta})=1-\sqrt{(1-\vartheta^1)^2+(1-\vartheta^2)^2+\ldots +(1-\vartheta^T)^2}\, ,$$
which leads to a measure of rationality based on the sum of square differences from the case of exact rationality (where $\overline{\vartheta}=(1,1,\ldots,1)$).

Computing these generalized rationality indices can be more demanding than computing the (basic) rationality index $\vartheta^*$ since in searching for those values of $\overline{\vartheta}$ that $\overline{\vartheta}$-rationalizes the data and maximizes $F(\overline{\vartheta})$, we would not in general be able to confine ourselves to the the case where $\vartheta^t=\vartheta^{t'}$ for all $t$, $t'$.  In the literature on measuring GARP violations, there are indices, such as the one proposed by \citet{varian1990}, that involve solving a maximization problem with the same mathematical structure. (In that case the problem is to find the best way to break up revealed preference cycles over consumption bundles rather than over price vectors.)  Algorithms that have been devised to compute Varian's index (see \citet{halevy2018} and \citet{polisson2020}) can also be used to compute $F^*$.

\medskip

\subsection{$\overline{\vartheta}$-GAPP and $\overline{\vartheta}$-GARP}\label{sec:thetaGPGP}

We confine our discussion to the environment where prices are linear, so the data set has the form $\mathcal{D}=\{(p^t,x^t)\}_{t=1}^T$.  Let $\overline{\vartheta}\in (0,1]^T$. We say that a utility function $\widetilde{U}:\mathbb{R}^L_+\to\mathbb{R}$ {\em $\overline{\vartheta}$-rationalizes $\mathcal{D}$ in the sense of Afriat} if $\widetilde{U}(x^t)\geq \widetilde{U}(x)$ for all $x\in B^t_{\overline{\vartheta}}$, where
$$B^t_{\overline{\vartheta}}=\{x\in\mathbb{R}^L_+:p^t\cdot x\leq \vartheta^t p^t\cdot x^t\}.$$
$\overline{\vartheta}$-rationalization in this sense admits a characterization similar to the one we gave for $\overline{\vartheta}$-rationalization in the augmented utility model.

Define the relations $\succeq_{x,\overline{\vartheta}}$ and $\succ_{x,\overline{\vartheta}}$ on the set $\{x^t\}_{t=1}^T$ in the following way:
\begin{quote}
$x^{t'}\succeq_{x,\vartheta} x^{t}$ if $p^{t'}\cdot x^t\leq \vartheta^{t'}p^{t'}\cdot x^{t'}$ and $x^{t'}\succ_{x,\vartheta} x^{t}$ if $p^{t'}\cdot x^t<\vartheta^{t'}p^{t'}\cdot 	x^{t'}$
\end{quote}
Denote the transitive closure of $\succeq_{x,\overline{\vartheta}}$ by $\succeq^*_{x,\overline{\vartheta}}$.  Obviously these definitions generalize the ones given for the revealed preference relations over bundles (see \refsec{sec:Afrsection} of the main paper). With these definitions in place, we can also generalize the definition of GARP.  We say that the data set $\mathcal{D}$ obeys {\em $\overline{\vartheta}$-GARP} if
\begin{quote}
{\em there do not exist observations $t,t'\in T$ such that $x^{t'}\succeq^*_{x,\overline{\vartheta}} x^{t}$ and $x^{t}\succ_{x,\overline{\vartheta}} x^{t'}$.}
\end{quote}

It is straightforward to show that $\overline{\vartheta}$-GARP is necessary for the $\overline{\vartheta}$-rationalization of $\mathcal{D}$ (in the sense of Afriat) by a locally nonsatiated utility function $\widetilde{U}:\mathbb{R}^L_+\to\mathbb{R}$.  It is also known (see \citet{halevy2018}) that $\overline{\vartheta}$-GARP is sufficient to guarantee the $\overline{\vartheta}$-rationalization of $\mathcal{D}$ (in Afriat's sense) by a continuous, strictly increasing and concave utility function $\widetilde{U}:\mathbb{R}^L_+\to\mathbb{R}$.\footnote{Indeed, we could obtain this result by modifying our proof of \refthm{thm:ultimate}.  First, $\overline{\vartheta}$-GARP guarantees that there is a complete preorder $\succsim$ on $\{x^t\}_{t=1}^T$ that extends $\succeq_{x,\overline{\vartheta}}$ and $\succ_{x,\overline{\vartheta}}$.  Then, by mimicking the proof of \reflem{lemma-Af}, one could guarantee the existence of numbers $\phi^t$ and $\lambda^t>0$ (for $t=1,2,\ldots, T$) with the following properties:\, (a) $\phi^{t'}>\phi^t$ if $x^{t'}\succ x^t$; (b) $\phi^{t'}=\phi^t$ if $x^{t'}\sim x^t$; and (c) $\phi^{t'}\leq \phi^{t}+\lambda^{t}p^t\cdot (x^{t'}-\vartheta^t x^t)$ for all $t\neq t'$.  The utility function $\widetilde{U}:\mathbb{R}^L_+\to\mathbb{R}$ given by
$$U(x)=\min_{t\in T}\{\phi^t+\lambda^tp^t\cdot (x-\vartheta^t x^t)\}$$
is a continuous, concave, and strictly increasing.  It is straightforward to check that property (c) guarantee that $\widetilde{U}$ rationalizes $\mathcal{D}$ in Afriat's sense.}  By definition, the critical cost efficiency index $c^*$ satisfies
$$c^*=\sup\{\vartheta\in (0,1]:\mbox{$\mathcal{D}$ has a $(\vartheta,\vartheta,\ldots,\vartheta)$-rationalization in the sense of Afriat}\}$$
and since $\overline{\vartheta}$-rationalization in Afriat's sense can be characterized by $\overline{\vartheta}$-GARP, we obtain
\begin{equation}\label{CCEI-GARP}
c^*=\sup\{\vartheta\in (0,1]:\mbox{$\mathcal{D}$ satisfies $(\vartheta,\vartheta,\ldots,\vartheta)$-GARP}\}.
\end{equation}
With these observations in place, the proof of \refprop{prop:GAPP-GARP} is now straightforward.\medskip

\noindent {\bf Proof of \refprop{prop:GAPP-GARP}.}\,  First we note that there is a generalization to Proposition \ref{prop:GAPPscaling}: it is straightforward to check $p^{t'}\succeq_{p,\overline{\vartheta}} p^t$ if and only if $\breve x^{t'}\succeq_{x,\overline{\vartheta}} \breve x^t$ and $p^{t'}\succ_{p,\overline{\vartheta}} p^t$ if and only if $\breve x^{t'}\succ_{x,\overline{\vartheta}} \breve x^t$. Thus, $\mathcal{D}$ satisfies $\overline{\vartheta}$-GAPP if and only if
$\breve{\mathcal{D}}$ satisfies $\overline{\vartheta}$-GARP. Then it follows immediately from (\ref{RI-GAPP}) and (\ref{CCEI-GARP}) that the critical cost efficiency index of $\breve{\mathcal{D}}$ is equal to the rationality index of $\mathcal{D}$.\hfill $\blacksquare$\medskip

\subsection{Allowing for variation in product characteristics across observations}\label{sec:differentiatedgoods}

In \refsec{sec:discretespace}(3) we considered a model of differentiated goods, where each product is represented by a vector of product characteristics in the space $\Re^L_+$. We assumed in that section that the set of available goods, $X$, is fixed across observations but that assumption is not crucial to our model or test.  We now allow the range of products available to the consumer to vary across observations.

The changes we have in mind include the introduction of new products and also changes to characteristics of an existing product.   The latter could be a substantive change --- for example, a change to the formula for a breakfast cereal --- or it could be a change (say) to the amount of money spent on advertising that alters a product's utility (in the broad sense).  All these cases could be formally captured by a data set  $\mathcal{D}=\{(\psi^t,x^t,X^t)\}_{t=1}^T$, where $X^t$ is the set of products available at observation $t$, $x^t$ (as usual) is the product chosen, and $\psi^t:X^t\to\mathbb{R}_+$ is the price system as observation $t$.  Notice that the price system at observation $t$ is defined on $X^t$ (the set of available products at observation $t$).  An augmented utility function $U:Y\times\Re_{-}\to\Re$, where $Y$ is a subset of $\Re^L_+$ containing $\cup_{t\in T}X^t$ rationalizes $\mathcal{D}$ if, at each observation $t$,
$$U(x^t,-\psi^t(x^t))\geq U(x,-\psi^t(x))\:\mbox{ for all $x\in X^t$;}$$
in other words, $x^t$ and its associated expenditure gives greater utility than any other product available at observation $t$. Sometimes, there is universal agreement that certain product characteristics $K\subset L$ will always make the product more desirable; in this case, we would also like the rationalizing utility function to be increasing in $x_K$.

Developing a test of whether $\mathcal{D}=\{(\psi^t,x^t,X^t)\}_{t=1}^T$ can be rationalized by an augmented utility function that is increasing in $x_K$ requires a modification of the notion of revealed preference.

We say that $\psi^{t'}$ is {\em directly revealed preferred} to $\psi^t$, and denote it by $\psi^{t'}\succeq_{vp}\psi^t$ if $\psi^{t'} (\hat x)\leq \psi^{t}(x^{t})$ where $\hat x\in X^{t'}$ and $\hat x\geq _K x^t$.\footnote{The partial order $\geq_K$ is defined as follows: $x''\geq_K x'$ if $x''_{-K}=x'_{-K}$ and $x''_{K}\geq x'_K$.} In other words, $\psi^{t'}$ is directly revealed preferred to $\psi^t$ if there is a product $\hat x$ available at $t'$ that is weakly superior to $x^t$ in the dimensions belonging to $K$, the same in the other dimensions, and costs less than $x^t$.   We say that $\psi^{t'}$ is {\em directly strictly revealed preferred} to $\psi^t$, and denote it by $\psi^{t'}\succ_{vp}\psi^t$ if $\psi^{t'}$ is directly revealed preferred to $\psi^t$ and, either $\psi^{t'} (\hat x)< \psi^{t}(x^{t})$ or $\hat x>_K x^t$. We denote the transitive closure of $\succeq_{vp}$ by $\succeq_{vp}^*$, that is, $\psi^{t'}\succeq_{vp}^* \psi^t$ if there are $t_1$,  $t_2,\ldots,t_N$ in $T$ such that $\psi^{t'}\succeq_{vp} \psi^{t_1}$, $\psi^{t_1}\succeq_{vp} \psi^{t_2},\ldots,\psi^{t_{N-1}}\succeq_{vp} \psi^{t_N}$, and $\psi^{t_N}\succeq_{vp} \psi^{t'}$; in this case we say that $\psi^{t'}$ is {\em revealed preferred} to $\psi^t$.  If anywhere along this sequence, it is possible to replace $\succeq_{vp}$ with $\succ_{vp}$ then we denote that relation by $\psi^{t'}\succ^*_{vp} \psi^{t}$ and say that $\psi^{t'}$ is {\em strictly revealed preferred} to $\psi^t$.

It is straightforward to check that if $\mathcal{D}$ can be rationalized by an augmented utility function that is strictly increasing in $x_K$ then it obeys {\em GAPP with respect to $\succeq^*_{vp}$ and $\succ^*_{vp}$}, in the following sense:
\begin{quote}
{\em there do not exist observations $t,t'\in T$ such that $\psi^{t'}\succeq^*_{vp} \psi^{t}$ and $\psi^{t}\succ^*_{vp} \psi^{t'}$.}
\end{quote}
The following theorem asserts that the converse is also true.\medskip

\begin{theorem}\label{thm:GAPP-advert}
Let the data set be $\mathcal{D}=\{(\psi^t,x^t,X^t)\}^T_{t=1}$, where $X^t$ is finite for all $t\in T$ and $\psi^t:X^t\to\Re_+$ is strictly increasing in $x_K$, i.e., if $x''>_K x'$ and both $x''$ and $x'$ are in $X^t$, then $\psi^t(x'')>\psi^t(x')$.  Let $Y$ be a closed set in $\Re^L_+$ containing $\cup_{t\in T}X^t$.

Then $\mathcal{D}$ can be rationalized by an augmented utility function $U:Y\times \Re_{-}\to\Re$  that is strictly increasing in $x_K$ if and only if satisfies GAPP with respect to $\succeq^*_{vp}$ and $\succ^*_{vp}$.
\end{theorem}

\noindent {\bf Proof.}\,  We skip the proof of the necessity of GAPP, which is straightforward, and turn to establishing its sufficiency. Let $X=\cup_{t\in T}X^t$.  We claim that we can extend the function $\psi^t:X^t\to\Re_+$ to a function $\underline\psi^t :X\to\Re$ that is increasing in $x_K$ and such that $\underline{\mathcal{D}}=\{(\underline\psi^t,x^t)\}_{t=1}^T$ satisfies GAPP (with respect to the revealed preference orders $\succeq^*_p$ and $\succ^*_p$ induced by $\underline{\mathcal{D}}$).   Then an application of \refthm{thm:GAPP_nonlinear} will guarantee that $\underline{\mathcal{D}}$, and thus also $\mathcal{D}$, can be rationalized by an augmented utility function $U:Y\times \Re_{-}\to\Re$ that is strictly increasing in $x_K$.

To guarantee that $\underline{\mathcal{D}}$ satisfies GAPP, with respect to $\succeq^*_p$ and $\succ^*_p$, we need to specify $\underline\psi^t(x)$, for $x\in X\setminus X^t$, in such a way that $\succeq^*_p=\succeq^*_{vp}$ and $\succ^*_p=\succ^*_{vp}$.  Then GAPP holds with respect to $\succeq^*_{p}$ and $\succ^*_{p}$ because GAPP holds with respect to $\succeq^*_{vp}$ and $\succ^*_{vp}$.  Because $X$ is finite, such an extension $\psi^t$ can be obtained with no technical difficulty. For $x\in X\setminus X^t$, if there is no $x'\in X^t$ such that $x'>_K x$, we choose  $\underline\psi^t (x)>\max\{\psi^s(x^s):s\in T\}$, while making sure that $\underline\psi^t$ remains increasing in $x_K$.  If there is $x'\in X^t$ such that $x'>_K x$, then choose $\underline\psi^t(x)$ to be strictly lower than $\psi^t(x')$, but if $x=x^s$ for some observation $s$, then choose $\underline\psi^t(x)=\underline\psi^t(x^s)>\psi^s(x^s)$ if $\psi^t(x')>\psi^s(x^s)$.  In this way, we guarantee $\succeq^*_p=\succeq^*_{vp}$ and $\succ^*_p=\succ^*_{vp}$.   \hfill $\blacksquare$

\section{More on compensating variation}\label{sec:morecomvar}

Our objective is to prove equation (\ref{justnice2}) from the body of the paper:
\begin{equation} \label{justnice-app}
\inf(\mu_c)= \max\{m_c^s:\mbox{$m_c^s$ satisfies (\ref{justnice}) for some $s\in S$}\}
\end{equation}
where (\ref{justnice}) requires $p^{t_2}x^s+m_c^s=p^sx^s$. \medskip

\noindent {\bf Proof.}\, Since $S$ is a finite set, there is $\bar s\in S$ that achieves the maximum on the right of (\ref{justnice-app}).  We have already shown that $\inf(\mu_c)\geq m_c^{\bar s}$, so it remains to show that they are equal.  We shall do this by producing, for any $\epsilon>0$, an augmented utility function rationalizing $\mathcal{D}$ for which the compensating variation is smaller than $m_c^{\bar s}+\epsilon$.

To this end, let $U$ be any augmented utility function that rationalizes $\mathcal{D}=\{(p^t,x^t)\}_{t=1}^T$; we know that $U$ exists since $\mathcal{D}$ obeys GAPP by assumption.  Let $\hat\psi:X\to\mathbb{R}_+$ be the nonlinear price system given by $\hat\psi (x)=p^{t_2}\cdot x+m_c^{\bar s}+\epsilon$ and suppose that $\hat x\in {\argmax}_{x\in X}U(x,-\hat\psi (x))$.  Now consider the data set $\mathcal{D}'=\mathcal{D}\cup \{(\hat\psi, \hat x)\}$.  Obviously this data set can be rationalized (in fact it is rationalized by $U$).  Furthermore, $\hat\psi \not\succeq_p p^s$ for any $s\in S$.  This is because
$$\hat\psi (x^s)=p^{t_2}x^s+m_c^{\bar s}+\epsilon > p^{t_2}\cdot x^s+m_c^s=p^sx^s$$
for any $s\in S$. (Recall that, be definition, $m_c^{\bar s}\geq m_c^s$ for all $s\in S$.)  Thus there is a complete preorder $\succsim$ on $\{p^t\}_{t=1}^T\cup \{\hat\psi\}$, completing the revealed preference relations on $\mathcal{D}'$ such that $p^{t_1}\succ \hat\psi$.
By \refthm{thm:ultimate2}, there is an augmented utility $\hat U$ rationalizing $\mathcal{D}'$ such that its indirect utility $\hat V$ satisfies  $\hat V(p^{t_1})>\hat V(\hat\psi)$. In other words,
$$\hat V(\hat\psi)=\max_{x\in X} \hat U(x,-p^{t_2}\cdot x-m_c^{\bar s}-\epsilon)<\hat U(x^{t_1},-p^{t_1}\cdot x^{t_1}).$$
So for the augmented utility function $\hat U$, the compensating variation must be smaller than $m_c^{\bar s}+\epsilon$. \hfill $\blacksquare$
\medskip

Our treatment of the compensating and equivalent variations can be easily extended to allow for nonlinear pricing.  We give a sketch of the procedure for calculating a bound on the compensating variation and leave the reader to fill in the details; this procedure is completely analogous to the one for linear prices described in \refsec{sec:compensation}

Let $U$ be the consumer's augmented utility function. Suppose that the initial price is $\psi^{t_1}$ and it changes to $\psi^{t_2}$, leading to a change in consumption from $x^{t_1}$ to $x^{t_2}$.   Then the compensating variation $\mu_c$ is, by definition, the variable that solves the equation
\begin{equation}\label{CVnonlinear}
\max\nolimits_{x\in \Re_+^L} U(x,-\psi^{t_2}(x)-\mu_c)=V(\psi^{t_1})=U(x^{t_1},-\psi^{t_1}(x^{t_1})).
\end{equation}
Note that $\mu_c$ is unique since $U$ is strictly increasing in the last argument. We could think of $\mu_c$ as the lump sum transferred {\em from} the consumer (if it is positive) or {\em to} the consumer (if it is negative) after the price change that will make her indifferent between the two situations.

Now suppose a data set $\mathcal{D}$ obeys GAPP and contains the observation $(\psi^{t_1},x^{t_1})$. How can we form a lower bound of the compensating variation of a price change from $\psi^{t_1}$ to $\psi^{t_2}$?  (Note that our discussion is valid whether or not $\psi^{t_2}$ is an observed price system in the $\mathcal{D}$.)  Formally, we wish to find
$$\inf\{\mu_c:\mbox{$\mu_c$ solves (\ref{CVnonlinear}) for some augmented utility function $U$ that rationalizes $\mathcal{D}$}\}.$$
Abusing terminology somewhat, we shall denote this term by $\inf({\mu}_{c})$.

We now describe how to compute this bound. Let $S\subset T$ be the set of observations such that $s\in S$ if $\psi^{s}\succeq^*_p \psi^{t_1}$.  This set is nonempty since it contains $p^{t_1}$ itself.  For each $s\in S$, there is $m_c^s$ such that
\begin{equation}\label{justnice-nl}
\psi^{t_2}(x^s) +m_c^s=\psi^s(x^s).
\end{equation}
For any $U$ that rationalizes $\mathcal{D}$, the compensating variation $\mu_c\geq m_c^s$. Indeed, if $m<m_c^s$, then $m\neq \mu_c$ for any utility function rationalizing $\mathcal{D}$ because
\begin{multline*}
\max\nolimits_{x\in \Re_+^L} U(x,-\psi^{t_2}(x)-m)\geq U(x^s,-\psi^{t_2}(x^s)-m)> U(x^s,-\psi^{t_2}(x^s)-m_c^s)\\
=U(x^s,-\psi^s(x^s))\geq U(x^{t_1},-\psi^{t_1}(x^{t_1}))=V(\psi^{t_1}).
\end{multline*}
Thus $\inf(\mu_c)\geq m_c^s$ for all $s\in S$.  In fact, by adapting the argument we provided for the case of linear prices in the earlier part of this section, we could show that
\begin{equation}  %\label{justnice2}
\inf(\mu_c)= \max\{m_c^s:\mbox{$m_c^s$ satisfies (\ref{justnice-nl}) for some $s\in S$}\}.
\end{equation}
Since the right side of this equation can be computed from the data, we have found a practical way of calculating $\inf(\mu_c)$.
\medskip

Notice that if $\psi^{t_2}$ is revealed preferred to $\psi^{t_1}$, in the sense that there is $s'\in S$ such that $m_c^{s'}\geq 0$, then $\inf(\mu_c)\geq 0$; in other words, a lump sum {\em tax} of $\inf(\mu_c)$ will leave the agent no worse off than at $t_1$ and potentially better off.  On the other hand, if $\psi^{t_2}$ is {\em not} revealed preferred to $\psi^{t_1}$, that is, for every $s\in S$, we have $m_c^s< 0$, then $\inf(\mu_c)< 0$; in other words, at $\psi=\psi^{t_2}$, a lump sum {\em transfer} of $\inf(\mu_c)$ to the agent will guarantee that the agent no worse off than at $t_1$ and potentially better off.

\section{Proof of \refthm{thm:StGAPPTest}}\label{sec:appendix_RAUM}

\noindent It remains for us to show that if there is $\nu\in \Re^{|\mathcal{A}|}_+$ such that $A\nu=\pi$ then there is a probability space $(\Omega, \mathcal{F},\mu)$ and a random variable $\chi:\Omega\to (\mathbb{R}^{L}_{+})^T$ such that $\{(p^t,\chi^t(\omega))\}_{t\in T}$ obeys GAPP almost surely and that (\ref{eqn:condi}) in the main paper holds, that is,
\begin{equation}\label{eqn:Appcondi}
\mathring{\pi}^t(Y)=\mu(\{\omega\in\Omega: \chi^t(\omega)\in Y\})\:\mbox{ for any measurable $Y\subset \Re^L_+$.}
\end{equation}

Given $\mathring\pi^t$ we define $\tilde\pi^{i_t,t}$ to be the conditional distribution of demand at observation $t$ when it restricted to the cone
$K^{i_t,t}=\{r\cdot x: x\in B^{i_t,t},\: r>0\}$.   Thus, if $Y$ is a measurable subset of $\Re^L_+$, then
$$\mathring\pi^t(Y\cap K^{i_t,t}) =\pi^{i_t,t}\,\tilde\pi^{i_t,t}(Y).$$
(Recall that, by definition, $\pi^{i_t,t}=\mathring\pi^t(K^{i_t,t})$.)  Of course, if $Y\cap K^{i_t,t}=\emptyset$ then $\tilde\pi^{i_t,t}(Y)=0$.

Given $a$ and $t$, there is a unique $i'_t$ such that $a^{i'_t,t}=1$; let $K^t_a=K^{i'_t,t}$ and let $\tilde\pi^t_a$ be the probability measure on $\Re^L_+$ such that $\tilde\pi^t_a=\tilde\pi^{i'_t,t}$.  Obviously, $\tilde\pi^t_a(K^t_a)=1$.

Let $\lambda_a$ be the product measure on $(\Re^L_+)^T$ given by $\lambda_a=\times_{t\in T} \tilde\pi^t_a$.  It follows from the definition of $a$ that
$$ \times_{t\in T} K^t_a \subset \left\{x\in (\Re^L_+)^T:\mbox{ $\{(p^t,x^t)\}_{t\in T}$ satisfies GAPP}\right\}$$
and since $\tilde\pi^t_a(K^t_a)=1$ for all $t$, we obtain
\begin{equation}\label{eqn:asu}
\lambda_a\left(\left\{x \in (\Re^L_+)^T:\mbox{$\{(p^t,x^t))\}_{t\in T}$ satisfies GAPP}\right\}\right)=1.
\end{equation}
Note that $x^t$ refers to the $t$th entry of $x$).

Define $\Omega=\mathcal{A}\times (\Re^L_+)^T$ and the probability measure $\mu$ on $\Omega$ by $\mu (\{a\}\times Y)=\nu_a\lambda_a(Y)$ for any measurable set $Y\subseteq (\Re^L_+)^T$, where $\nu_a$ refers to the $a$th entry of $\nu$. Lastly, define $\chi:\Omega\to (\Re^L_+)^T$ by $\chi((a,x))=x$.  Then, using (\ref{eqn:asu}), we obtain
\begin{multline*}
\mu\left(\left\{\,(a,x)\in\Omega: \{(p^t,\chi^t(a,x))\}_{t\in T}\mbox{ satisfies GAPP}\, \right\}\right) \\ =\sum_{a\in \mathcal{A}}\nu_a\lambda_{a}\left(\left\{x\in(\Re^L_+)^T: \{(p^t,\chi^t(a,x))\}_{t\in T}\mbox{ satisfies GAPP}\right\}\right)
= \sum_{a\in A}\nu_a =1.
\end{multline*}
It remains for us to show that (\ref{eqn:Appcondi}) holds.  Let $Y$ be a measurable set in $\Re^L_+$. For any $K^{i_t,t}$,
\begin{eqnarray*}
\mu(\{(a,x)\in\Omega: \chi^t(a,x)\in Y\cap K^{i_t,t}\})&=& \sum_{a\in\mathcal{A}}\nu_a\lambda_a(\{x\in (\Re^L_+)^T:\chi^t(a,x)\in Y\cap K^{i_t,t}\})  \\
&=&  \sum_{a\in\mathcal{A}}\nu_a\lambda_a(\{x\in (\Re^L_+)^T: x^t\in Y\cap K^{i_t,t}\}) \\
&=&\sum_{a\in \mathcal{A}}\nu_a\tilde\pi^t_a(\{x^t\in \Re^L_+: x^t\in Y\cap K^{i_t,t}\})
\end{eqnarray*}
Recall that $\mathcal{A}^{i_{t},t}=\{a\in\mathcal{A}: a^{i_{t},t}=1\}$, so for any $a\notin \mathcal{A}^{i_{t},t}$,
$$\tilde\pi^t_a(\{x^t\in \Re^L_+: x^t\in Y\cap K^{i_t,t}\})=0.$$
Thus
\begin{eqnarray*}
\sum_{a\in \mathcal{A}}\nu_a\tilde\pi^t_a(\{x^t\in \Re^L_+: x^t\in Y\cap K^{i_t,t}\})
&=& \sum_{a\in \mathcal{A}^{i_t,t}}\nu_a\tilde\pi^t_a(\{x^t\in \Re^L_+: x^t\in Y\cap K^{i_t,t}\}) \\
&=& \frac{\mathring\pi^t(Y\cap K^{i_t,t})}{\pi^{i_t,t}}\sum_{a\in \mathcal{A}^{i_t,t}}\nu_a \\
&=& \mathring\pi^t(Y\cap K^{i_t,t}),
\end{eqnarray*}
where the last equation follows from the fact that $A\nu=\pi$. Thus we have shown that, for all $K^{i_t,t}$,
$$\mu(\{(a,x)\in\Omega: \chi^t(a,x)\in Y\cap K^{i_t,t}\})=\mathring\pi^t(Y\cap K^{i_t,t}).$$
This in turn guarantees that (\ref{eqn:Appcondi}) holds.\hfill $\blacksquare$

\section{Omitted Details from \refsec{sec:metrics}} \label{sec:appendix_metrics}

In this section, we formally develop our bootstrap procedure from \refsec{sec:econ_wel}. We begin by describing Weyl-Minkowski duality\footnote{See, for example, Theorem 1.3 in \citet{Ziegler}.}%\citeasnounNew{Ziegler}.}
which is used for the equivalent (dual) restatement \eqref{eq:h-test} of our test \eqref{eq:v-test}. As we mentioned earlier, we will also appeal to this duality in the proof of the asymptotic validity of our testing procedure.

\begin{theorem}\label{thm:weyl-minkowsky} (Weyl-Minkowski Theorem for Cones)  A subset $\mathcal{C}$ of ${\Re}^I$
	is a finitely generated cone	
	\begin{equation}\label{eq:v-rep}
	\mathcal{C}=\{\nu _{1}a_{1}+...+\nu _{|A|}a_{|A|}:\nu _{h}\geq 0\} \text{ for some } A = [a_1,...,a_H] \in {\Re}^{I\times |A|}
	\end{equation}
	if, and only if, it is a finite intersection of closed half spaces
	\begin{equation}\label{eq:h-rep}
	\mathcal{C}=\{t\in \Re^{I}|Bt\leq 0\} 	\text{ for some } B \in {\Re}^{m\times I}.
	\end{equation}
\end{theorem}

The expressions in \eqref{eq:v-rep} and \eqref{eq:h-rep} are called a $\mathcal{V}$-representation (as in ``vertices'') and a $\mathcal{H}$-representation (as in ``half spaces'') of $\mathcal{C}$, respectively. In what follows, we use an $\mathcal{H}$-representation of $\mathrm{cone}(A)$ corresponding to  a $m \times I$ matrix $B$ as implied by \refthm{thm:weyl-minkowsky}.

We are now in a position to show that the bootstrap procedure defined in \refsec{sec:econ_wel} is asymptotically valid. Note first that ${\Theta} = [\underline{\theta},\widebar{\theta}]$, where
\begin{eqnarray}
\widebar{\theta}&=&\max_{\nu \in  \Delta^{|A|-1}} \rho\nu = \max_{1 \leq j \leq |A|} \rho_j     \\ \label{eq:utheta}
\underline{\theta}&=&\min_{\nu \in  \Delta^{|A|-1}} \rho\nu = \min_{1 \leq j \leq |A|} \rho_j,  \label{eq:ltheta}
\end{eqnarray}
where $\rho_j$ denotes the $j$th component of $\rho$. We normalize $(\rho,\theta)$ such that $\Theta = [\underline{\theta},\underline{\theta}+1]$. Next, define
\begin{eqnarray}
\mathcal H &:=&\{1,2,...,|A|\} \label{eq:H-def} \\
\widebar{\mathcal H} &:=& \{j \in \mathcal H\; |\; \rho_j = \widebar{\theta}\} \\
\underline{\mathcal H} &:=& \{j \in \mathcal H \; | \; \rho_j = \underline{\theta}\} \\
\mathcal H_0 &:=& \mathcal H  \setminus (\widebar{\mathcal{H}} \cup \underline{\mathcal{H}}). \label{eq:H0-def}
\end{eqnarray}
Recall that $\tau_N$ is a tuning parameter chosen such that $\tau _{N}\downarrow 0$ and $\sqrt{N}\tau _{N}\uparrow \infty $. For $\theta \in \Theta_I$, we now formally define the $\tau_N$-tightened version of $\mathcal S$ as
$$
\mathcal S_{\tau_N}(\theta) := \{A\nu\; | \; \rho\nu = \theta, \nu \in \mathcal V_{\tau_{N}}(\theta)\},
$$
where
$$\mathcal V_{\tau_{N}}(\theta) := \left\{\nu \in \Delta^{|A|-1} \; \left| \; \begin{array}{l}	\nu_j \geq
\frac
{(\widebar{\theta} - \theta)\tau_{N}}
{|\underline{\mathcal{H}}   \cup \mathcal{H}_0|},
j \in  \underline{\mathcal{H}}, \;
\nu_{j'} \geq
\frac {(\theta - \underline{\theta})\tau_{N}}
{|\widebar{\mathcal{H}} \cup \mathcal{H}_0  |}, \;
j' \in \widebar{\mathcal{H}}, \;
\\	\\
\nu_{j''} \geq
\left[
1 -
\frac
{(\widebar{\theta} - \theta) |\underline{\mathcal{H}}|}
{|\underline{\mathcal{H}}   \cup \mathcal{H}_0|}
-
\frac {(\theta - \underline{\theta}) |\widebar{\mathcal{H}}|}
{|\widebar{\mathcal{H}} \cup \mathcal{H}_0  |}
\right]
\frac {\tau_N}{|\mathcal{H}_0|}, \;
j'' \in {\mathcal{H}}_0
\end{array}\right.
\right\}.$$

In applications where $\rho$ is binary, the above notation simplifies. Specifically, in our empirical application on deriving the welfare bounds, $\rho =  \mathbb{1}_{t\succeq^*_p t'}$ and $\theta = \mathcal{N}_{t\succeq^*_p t'}$. Here, $\widebar \theta = 1$, $\underline{\theta} = 0$, and $\widebar \theta - \underline{\theta} = 1$ holds without any normalization. Also, $\widebar {\mathcal H}$ ($\underline{\mathcal H}$) is just the set of indices for the types that (do not) prefer price $p^t$ compared to $p^{t'}$, while $\mathcal H_0$ is empty. We then have:
$$
\mathcal S_{\tau_{N}}(\mathcal{N}_{t\succeq^*_p t'}) = \left\{A\nu\; \left| \;  \mathbb{1}_{t\succeq^*_p t'}'\nu = \mathcal{N}_{t\succeq^*_p t'}, \nu \in \mathcal V_{\tau_{N}}(\mathcal{N}_{t\succeq^*_p t'})\right. \right\},
$$
where
$$\mathcal V_{\tau_{N}}(\mathcal{N}_{t\succeq^*_p t'}) = \left\{\nu \in \Delta^{|A|-1} \; \left| \; \nu_j \geq 	\frac{(1 - \mathcal{N}_{t\succeq^*_p t'})\tau_{N}}{|\underline{\mathcal{H}}|},
j \in  \underline{\mathcal{H}}, \;  \nu_{j'} \geq \frac {\mathcal{N}_{t\succeq^*_p t'}\tau_{N}}{|\widebar{\mathcal{H}}|}, \;
j' \in \widebar{\mathcal{H}} \right. \right\}.$$

\medskip

We now state the mild data assumptions.
\begin{assumption}
	\label{ass:proportions} For all $t = 1,...,T$, $\frac {N_t}{N} \rightarrow \kappa_t$ as $%
	N \rightarrow \infty$, where $\kappa_t >0$, $1 \leq t \leq T$.
\end{assumption}

\begin{assumption}\label{ass:sampling}
	The econometrican observes $T$ independent cross-sections of i.i.d. samples $\left\{
	x_{n(t)}^t\right\} _{n(t)=1}^{N_{t}},t=1,...,T$ of consumers' choices corresponding to the known price vectors $\{p_t\}_{t=1}^T$.
\end{assumption}

Next, let ${\bf d}^{i,t}_{n(t)} :=\textbf{1}\{x_{n(t)}^t \in B^{i,t}\}$, ${\bf d}_{n(t)}^t = [{\bf d}^{1,t}_{n(t)},...,{\bf d}^{{I_t},t}_{n(t)}]$, and  ${\bf d}_{n}^t = [{\bf d}^{1,t}_{n},...,{\bf d}^{{I_t},t}_{n}]$.
Let ${\bf d}_t$ denote the choice vector of a consumer facing price $p^t$ (we can, for example, let  ${\bf d}_t = {\bf d}_1^t$).     Define ${\bf d} = [{\bf d}_1',...,{\bf d}_T']'$: note, ${\mathrm{E}[{\bf d}] = \pi}$ holds by definition.  Among the rows of $B$ some of them correspond to constraints that hold trivially by definition, whereas some are for non-trivial constraints.  Let ${\mathcal K}^R$ be the index set for the latter.  Finally, let
\begin{eqnarray*}
	g &=& B{\bf d}
	\\
	&=& [g_1,...,g_m]'.
\end{eqnarray*}
With these definitions, consider the following requirement:
\begin{condition}
	\label{condition 1} For each  $k \in {\mathcal K}^R$,  {\rm{var}}$(g_k)  > 0$ and $\mathrm{E}[|g_k/\sqrt{\mathrm{var}(g_k)}|^{2+c_1}] < c_2$ hold, where $c_1$ and $c_2$ are positive constants.
\end{condition}
\noindent
This  guarantees the Lyapunov condition for the triangular array CLT used in establishing asymptotic uniform validity.  This type of condition has been used widely in the literature of moment inequalities; see \citet{andrews-soares}.

\noindent {\sc \textbf{Proof of \refthm{thm:validity}.}}

Define
$$
\mathcal C = {\mathrm {cone}}(A)
$$
and
$$
\mathcal T(\theta) = \{\pi = A \nu: \rho'\nu = \theta, \nu \in \Re^{|A|}\},
$$
an affine  subspace of $\Re^{I}$.  It is convenient to rewrite $\mathcal T(\theta)$ as
$\mathcal T(\theta)= \{t \in \Re^{I}: \tilde B t = d(\theta)\} $ where $\tilde B  \in \tilde m \times \Re^{I} $, $d(\cdot) \in \tilde m \times 1$, and $\tilde m$ all depend on $(\rho,A)$.  We let $\tilde b_j$ denote the $j$-th row of $\tilde B$. Then
$$
\mathcal{S}(\theta) = \mathcal C \cap \Delta^{|A|-1} \cap \mathcal T(\theta).
$$
By \refthm{thm:weyl-minkowsky}, $\mathcal C = \{\pi: B \pi \leq 0\}$, therefore
\begin{eqnarray}\label{eq:altS}
\mathcal{S}(\theta) &=& \{t \in \Re^{|A|}:  Bt \leq 0, \tilde Bt = d(\theta), {\bf 1}_H't = 1 \}.
\end{eqnarray}
Let
$$
\psi(\theta) = [\psi_1(\theta),...,\psi_H(\theta)]' \quad \theta \in \Theta
$$
with
$$
\psi_j(\theta) =\left\{\begin{array}{cl}
\frac
{(\widebar{\theta} - \theta)}
{|\underline{\mathcal{H}}   \cup \mathcal{H}_0|}
& \text{ if } \; j \in  \underline{\mathcal{H}},
\\
\frac
{(\theta - \widebar{\theta})}
{|\widebar{\mathcal{H}}   \cup \mathcal{H}_0|}
& \text{ if } \; j \in  \widebar{\mathcal{H}},
\\
\left[
1 -
\frac
{(\widebar{\theta} - \theta) |\underline{\mathcal{H}}|}
{|\underline{\mathcal{H}}   \cup \mathcal{H}_0|}
-
\frac {(\theta - \underline{\theta}) |\widebar{\mathcal{H}}|}
{|\widebar{\mathcal{H}} \cup \mathcal{H}_0  |}
\right]
\frac {1}{|\mathcal{H}_0|}
& \text{ if } \; j \in \mathcal{H}_0,
\end{array}
\right.
$$
where terms are defined in \eqref{eq:H-def}-\eqref{eq:H0-def}. Then
$$
\mathcal{S}_{\tau_N}(\theta) = \{\pi = A\nu: \nu \geq \tau_N \psi(\theta), \nu \in \Delta^{|A|-1}, \rho '\nu  = \theta  \}.
$$
Finally, let
$$
\mathcal C_{\tau_N} = \{\pi = A\nu: \nu \geq \tau_N \psi(\theta)\}.
$$
Then
$$
\mathcal{S}_{\tau_N}(\theta) = \mathcal C_{\tau_N} \cap \Delta^{|A|-1} \cap \mathcal T(\theta).
$$
Proceeding as in the proof of Lemma 4.1 in \citetalias{kitamura2018}, we can express the set $\mathcal C_{\tau_N}$ as
$$
\mathcal C_{\tau_N} = \{t: Bt \leq -\tau_N \phi(\theta)\}
$$
where
$$
\phi(\theta) = -BA\psi(\theta).
$$
As in Lemma 4.1 in \citetalias{kitamura2018}, let the first $\bar m$ rows of $B$ represent inequality constraints and the rest equalities, and also let $\Phi_{kh}$ the $(k,h)$-element of the matrix $-BA$.   We have
$$
\phi_k = \sum_{h=1}^{|A|}\Phi_{kh} \psi_h(\theta)
$$
where, for each $k \leq \bar m$, $\{\Phi_{kh}\}_{h=1}^{|A|}$ are all nonnegative, with at least some of them being strictly positive, and $\Phi_{kh} = 0$ for all $h$ if $\bar m < k \leq m$.  Since  $ \psi_h(\theta) > 0, 1 \leq h \leq |A|$ for every $\theta \in \Theta$ by definition, we have $ \phi_j(\theta) \geq C, 1 \leq j \leq \bar m$ for some positive constant $C$, and $ \phi_j(\theta) = 0, \bar m < j \leq  m$ for every $\theta \in \Theta$.  Putting these together, we have
\begin{eqnarray*}
	\mathcal{S}_{\tau_N}(\theta) &=& \{t \in \Re^{|A|}:  Bt \leq -\tau_N \phi(\theta), \tilde Bt = d(\theta), {\bf 1}_H't = 1 \}
\end{eqnarray*}
where ${\bf 1}_H$ denotes the $|A|$-vector of ones.
Define the $\Re^{I}$-valued random vector
$$
\pi_{\tau_N}^* := \frac 1 {\sqrt N} \zeta + \hat \eta_{\tau_N}, \; \zeta \sim \mathrm{N}(0,\hat S)
$$
where $\hat S$ is a consistent estimator for the asymptotic covariance matrix of $\sqrt N(\hat \pi - \pi)$.  Then (conditional on the data) the distribution of
$$
\delta^*(\theta) := N\min_{\eta \in \mathcal S_{\tau_N}(\theta)}[\pi^*_{\tau_N} - \eta]'\Omega[\pi^*_{\tau_N} - \eta]
$$
corresponds to that of the bootstrap test statistics.  Let
$$
B_* :=
\begin{bmatrix}
B \\
\tilde B \\
{\bf 1}_H'
\end{bmatrix}
$$
Define $\ell = \mathrm{rank}(B_*)$ for the augmented matrix $B_*$ instead of $B$ in \citetalias{kitamura2018}, and let the $\ell \times m$-matrix $K$ be such that $KB_*$ is a matrix whose rows consist of a basis of the row space row $(B_*)$. Also let $M$ be an $(I-\ell )\times I$ matrix whose rows form an orthonormal basis of ker$B_*=$ ker$(KB_*)$, and define $P=\binom{KB_*}{M}$. Finally, let $\hat{g}=B_*\hat{\pi}$.

Define
\begin{eqnarray*}
	T(x,y)&:=&\binom{x}{y}^{\prime }{P^{-1}}^{\prime }\Omega P^{-1}\binom{x}{y},\quad x\in \Re^{\ell },y\in \Re^{I-\ell } \\
	t(x)&:=&\min_{y\in \Re^{I-\ell }}T(x,y) \\
	\quad s(g)&:=&\min_{\gamma = [{\gamma^\leq}',{\gamma^=}']', \gamma^\leq \leq 0, \gamma' \in \text{col}(B)}t(K[g-\gamma])
\end{eqnarray*}%
with
$$
\gamma^=  =
\begin{bmatrix}
{\bf 0}_{m - \bar m} \\
d(\theta) \\
1
\end{bmatrix}
$$
where ${\bf 0}_{m - \bar m}$ denotes the $(m - \bar m)$-vector of zeros.
It is easy to see that $t:\Re^{\ell }\rightarrow \Re_{+}$ is a
positive definite quadratic form.  By \eqref{eq:altS}, we can write
\begin{eqnarray*}
	\delta_{N}(\theta) =Ns(\hat{g})	=s(\sqrt{N}\hat{g}).
\end{eqnarray*}%
Likewise, for the bootstrapped version of $\delta$ we have
\begin{eqnarray*}
	\delta^*(\theta) &=& N\min_{\eta \in \mathcal S_{\tau_N}(\theta)}[\pi^*_{\tau_N} - \eta]'\Omega[\pi^*_{\tau_N} - \eta]  \\
	&=&s(\varphi_{N}(\hat \xi)   + \zeta),
\end{eqnarray*}%
where $\hat \xi = B_* \hat \pi/\tau_N$.  Note the function  $\varphi_N(\xi)   = [\varphi_N^1(\xi),...,\varphi^m_N(\xi) ]$  for $\xi =
(\xi_1,...,\xi_m)^{\prime } \in \text{col}(B_*)$.  Moreover, its $k$-th element $\varphi_N^k$ for $k \leq \bar m$ satisfies
\begin{equation*}
\varphi_N^k(\xi) = 0
\end{equation*}
if $|\xi^k| \leq \delta$ and $\xi^j \leq \delta, 1 \leq j \leq m$, $\delta > 0$, for large
enough $N$ and $\varphi_N^k(\xi) =  0$ for $k > \bar m$.   This follows (we use some notation in the proof of Theorem 4.2 in \citetalias{kitamura2018}, which the reader is referred to) by first noting that it suffices to show that for small enough $\delta >0$, every element $x^*$ that fulfills equation (9.2) in \citetalias{kitamura2018} with its RHS intersected with $\cap_{j=1}^{\tilde m} \tilde S_j(\delta), \tilde{S}_j(\delta) = \{x:|\tilde b_j'x - d_j(\theta)| \leq \tau \delta \}$ satisfies
$$
x^*|\mathcal S(\theta) \in \cap_{j=1}^q H_j^\tau \cap L \cap \mathcal{T}(\theta).
$$
If not, then there exists $(\tilde a,\tilde x) \in F \cap \mathcal{T}(\theta) \times\cap_{j=1}^q H_j \cap L \cap \mathcal{T}(\theta)$ such that
$$
(\tilde a - \tilde x)'(\tilde x|{\mathcal S}_{\tau}(\theta) - \tilde x) = 0,
$$
where $\tilde x|{\mathcal S}_{\tau}(\theta)$ denotes the orthogonal projection of
$\tilde x$ on ${\mathcal S}_{\tau}(\theta)$.
This, in turn, implies that there exists a triplet $(a_0,a_1,a_2) \in \mathcal A \times \mathcal A \times \mathcal A $ such that $(a_1 - a_0)'(a_2 - a_0) <0$.   But as shown in the proof of Theorem 4.2 in \citetalias{kitamura2018}, this cannot happen.  The conclusion then follows by Theorem 1 of \citet{andrews-soares}. \hfill $\blacksquare$

\bibliographystyle{econometrica}
\bibliography{GAPP_Refs3}

\end{document}